\documentclass[a4paper,aps,twocolumn,pre,intlimits,8pt,superscriptaddress]{revtex4-1}
\usepackage{amsmath,amsthm,amssymb,dsfont,mathrsfs,bbm}
\usepackage{mathtools,xfrac,xspace}
\usepackage{multirow}
\usepackage[usenames,dvipsnames]{xcolor}
\usepackage{tikz,pgfplots,bbm,bm,mathrsfs,dsfont}
\usetikzlibrary{shapes,snakes,arrows,pgfplots.groupplots,fadings,calc,decorations.markings,positioning,chains,fit,shapes,calc}
\usepackage[caption=false]{subfig}
\usepackage[english]{babel}
\usepackage{pgfplotstable}
\usepackage{hyperref}
\hypersetup{pdfauthor={Caracciolo Sicuro},pdftitle={One dimensional matching},%
            colorlinks, linktocpage=true, pdfstartpage=3, pdfstartview=FitV,%
    breaklinks=true, pdfpagemode=UseNone, pageanchor=true, pdfpagemode=UseOutlines,%
    plainpages=false, bookmarksnumbered, bookmarksopen=true, bookmarksopenlevel=1,%
    hypertexnames=true, pdfhighlight=/O,%
    urlcolor=orange, linkcolor=blue, citecolor=red, 
        }
\usepackage[T1]{fontenc}
\usepackage{lmodern}

\theoremstyle{plain}
\newtheorem{propcycle}{Proposition}[section]
\newtheorem{bischeraduept}[propcycle]{Proposition}
\newtheorem{pgone}[propcycle]{Proposition}
\newtheorem{lemma}[propcycle]{Lemma}
\newtheorem{pgoneth}[propcycle]{Theorem}
\newtheorem{plzeroth}[propcycle]{Theorem}
\newtheorem{corollary}[propcycle]{Corollary}
\newtheorem{corollary2}[propcycle]{Corollary}
\newtheorem{donskerth}[propcycle]{Theorem}

\theoremstyle{definition}
\newtheorem{cyclic}{Definition}[section]
\newtheorem{cooriented}[cyclic]{Definition}
\newtheorem{bischera}[cyclic]{Definition}
\newtheorem{crossing}[cyclic]{Definition}

\newcommand{\remp}{\textsc{rEmp}\xspace} 
\newcommand{\remps}{\textsc{rEmp}s\xspace} 
\newcommand{\reap}{\textsc{rEap}\xspace} 
\newcommand{\reaps}{\textsc{rEap}s\xspace}
\newcommand{\mG}{\mathcal{G}}
\newcommand{\mV}{\mathcal{V}}
\newcommand{\mE}{\mathcal{E}}
\newcommand{\mM}{\mathcal{M}}

\newcommand{\R}{\mathds{R}}

\DeclareMathOperator{\dd}{d}

\DeclareMathOperator{\e}{e}
\newcommand{\mC}[1]{{E^{(#1)}_N}}
\newcommand{\mc}[1]{{\varepsilon^{(#1)}_N}}

\newcommand{\Mod}[1]{\ (\mathrm{mod}\ #1)}

\allowdisplaybreaks
\begin{document}
\title{Random Euclidean matching problems in one dimension}
\author{Sergio Caracciolo}\email{sergio.caracciolo@mi.infn.it}
\affiliation{Dipartimento di Fisica, University of Milan and INFN, via Celoria 16, I-20133 Milan, Italy}
\author{Matteo D'Achille}\email{matteopietro.dachille@studenti.unimi.it}
\affiliation{Dipartimento di Fisica, University of Milan and INFN, via Celoria 16, I-20133 Milan, Italy}
\author{Gabriele Sicuro}\email{gabriele.sicuro@roma1.infn.it}
\affiliation{Dipartimento di Fisica, Sapienza Universit\`a di Roma, P.le A. Moro 2, I-00185, Rome, Italy}

\date{\today}
\begin{abstract}We discuss the optimal matching solution for both the assignment problem and the matching problem in one dimension for a large class of convex cost functions. We consider the problem in a compact set with the topology both of the interval and of the circumference. Afterwards, we assume the points' positions to be random variables identically and independently distributed on the considered domain. We analytically obtain the average optimal cost in the asymptotic regime of very large number of points $N$ and some correlation functions for a power-law type cost function in the form $c(z)=z^p$, both in the $p>1$ case and in the $p<0$ case. The scaling of the optimal mean cost with the number of points is $N^{-\sfrac{p}{2}}$ for the assignment and $N^{-p}$ for the matching when $p>1$, whereas in both cases it is a constant when $p<0$. 
Finally, our predictions are compared with the results of numerical simulations.
\end{abstract}
\maketitle

\section{Introduction}

After the seminal works of \textcite{Kirkpatrick1983}, \textcite{Orland1985}, and \textcite{Mezard1985}, random optimization problems have been successfully studied using statistical physics techniques, such as the replica trick or the cavity method \cite{Mezard1986a,Krauth1989a}. In a combinatorial optimization problem, we consider a finite set $\mathcal M$ of possible configurations $\mu$, and we associate a cost $E(\mu)\in\R$ to each configuration in $\mathcal M$. The goal is to find the optimal configuration $\mu$ such that $E(\mu)$ is minimized over $\mathcal M$. In a statistical physics approach, an ``inverse temperature'' $\beta$ is introduced, and we can write a partition function
\begin{equation}
Z(\beta)=\sum_{\mu\in\mM}\e^{-\beta E(\mu)},
\end{equation}
in such a way that the optimal cost is recovered as
\begin{equation}
\min_{\mu\in\mathcal M}E(\mu)=-\lim_{\mathclap{\beta\to+\infty}}\frac{\partial\ln Z(\beta)}{\partial \beta}.
\end{equation}
This statistical physics reformulation is particularly powerful when random combinatorial optimization problems are considered. In a random optimization problem, the set $\mathcal M$ depends on some parameters that are supposed to be random. In this case, $\mathcal M$ is therefore an instance of the problem in the space of parameters, and the average properties of the optimal solution are of a certain interest. In particular, denoting by $\overline{\bullet}$ the average over all instances $\mathcal M$,
\begin{equation}
\overline{\min_{\mu\in\mathcal M}E(\mu)}=-\lim_{\mathclap{\beta\to+\infty}}\frac{\partial\overline{\ln Z(\beta)}}{\partial \beta}.
\end{equation}
The average appearing in the previous equation can be tackled using the celebrated \textit{replica trick}, which allowed the derivation of fundamental results for many relevant random combinatorial optimization problems, like random matching problems \cite{Orland1985,Mezard1985,Mezard1987,Mezard1988} or the traveling salesman problem in its random formulation \cite{Orland1985,Mezard1986}.

In this paper, we will study a particular class of random optimization problems, namely random Euclidean matching problems (\remps). In the \remp, a set of $2N$ random points $\Xi\coloneqq \{\mathbf x_{i}\}_{i=1,\dots,2N}$ is given on a certain $d$-dimensional Euclidean domain. We associate a weight $w_{ij}$ to the couple $(\mathbf x_i,\mathbf x_j)$, typically in the form $w_{ij}=c\left(\|\mathbf x_i-\mathbf x_j\|\right)$ for some given function $c$. In the following, we will refer to the 
function $c$ as the cost function of the problem. We search therefore for the partition $\mu$ of $\Xi$ in $N$ sets of two elements such that
\begin{equation}
 E(\mu)=\sum_{\mathclap{(\mathbf x_i,\mathbf x_j)\in\mu}}w_{ij}
\end{equation}
is minimized. The object of interest is the optimal cost averaged over the points' positions,
\begin{equation}
\overline{\min_{\mu}E(\mu)}=\overline{\min_\mu\sum_{\mathclap{(\mathbf x_i,\mathbf x_j)\in\mu}}w_{ij}}.
\end{equation}
In a variation of the problem, called random Euclidean assignment problem (\reap), two sets of $N$ random points $\Xi\coloneqq \{\mathbf x_{i}\}_{i=1,\dots,N}$ and $\Upsilon\coloneqq \{\mathbf y_{i}\}_{i=1,\dots,N}$ are given, and we associate a weight $w_{ij}$ to the couple $(\mathbf x_i,\mathbf y_j)$, typically in the form $w_{ij}=c\left(\|\mathbf x_i-\mathbf y_j\|\right)$. In this case, only points of different sets can be coupled, and we search therefore for a permutation $\pi\in\mathcal S_N$ of $N$ elements such that
\begin{equation}
E(\pi)=\sum_{\mathclap{i=1}}^Nw_{i\,\pi(i)}
\end{equation}
is minimized. As before, the main object of interest is the average optimal cost,
\begin{equation}
 \overline{\min_{\pi\in\mathcal S_N}E(\pi)}=\overline{\min_{\pi\in\mathcal S_N}\sum_{\mathclap{i=1}}^N w_{i\,\pi(i)}}.
\end{equation}

A large physics literature exists about the properties of \remps and \reaps. In their seminal work, \textcite{Mezard1985} proposed a mean-field approximation of both the $d$-dimensional \remp and the $d$-dimensional \reap, obtaining the solution in the thermodynamical limit \cite{Mezard1985}. The finite-size corrections to the average optimal cost in the mean-field model have been also evaluated \cite{Mezard1987,*Ratieville2002,*Caracciolo2017}. Using the replica approach, it has been later shown that finite $d$ corrections can be included in the mean-field solution performing a diagrammatic expansion \cite{Mezard1988,*Lucibello2017}, whose resummation is, however, a challenging task. As an alternative to the classical methods, in Ref.~\cite{Caracciolo2015} a field theoretical approach has been proposed for the study of the so-called quadratic \reap in any dimension. The new approach was based on the deep connections with optimal transportation theory \cite{Villani2008,Ambrosio2016,bobkov2014}, and allowed the authors to give an exact analytical prediction for the average optimal cost for dimension $d=2$ and its finite-size corrections for dimension $d>2$ \cite{Caracciolo2014,Caracciolo2015}. Moreover, the \reap in one dimension on a compact domain has already been solved in the case of convex and increasing cost function~\cite{McCann1999, Boniolo2012, Caracciolo2014, Caracciolo2014c, Caracciolo2015, Caracciolo2015s}. Despite the numerous results on mean-field models, many properties of the corresponding Euclidean models remain to be investigated and, moreover, few exact results are available in finite dimension. The availability of analytical solutions is therefore of great importance to check the validity of the approximate results obtained correcting the mean-field theories, and the assumptions adopted to obtain them.

In the present work, we will restrict ourselves to the \remp and the \reap in one dimension, with the purpose of extending the analytical results of some previous investigations. Some results in the case of more general supports, as noncompact supports, and other general properties of the convergence rate can be found in the review of \textcite{bobkov2014}, in which the problem is treated in the context of optimal transportation theory. Despite their simple formulation, the one-dimensional \reap and the one-dimensional \remp have in general a nontrivial analytical treatment and they are related to many different problems in mathematics, physics, and biology. In Refs.~\cite{Boniolo2012,Caracciolo2014c} it has been shown that the optimal assignment in the case of a strictly increasing cost function can be interpreted as a stochastic process on a compact support, namely the Brownian bridge process, and therefore as a quadratic field theory \cite{Caracciolo2015}. On the other hand, if $c(z)$ is concave, it can be easily proved that, independently from the distribution adopted to generate the points, the optimal assignment is always \textit{planar} \cite{McCann1999}, in a sense that will be specified below. The relevance of planar matching configurations both in physics and in biology is due to the fact that they appear in the study of the secondary structure of single stranded DNA and RNA chains in solution \cite{Higgs2000}. These chains tend to fold in a planar configuration, in which complementary nucleotides are matched. The secondary structure of a RNA strand is therefore a problem of optimal matching on the line, with the restriction on the optimal configuration to be planar \cite{Orland2002,Nechaev2013}. The statistical physics of the folding process is not trivial and it has been investigated by many different techniques \cite{Bundschuh2002,Muller2003rna}, also in presence of disorder \cite{Higgs1996,Marinari2002,Bundschuh2002,Nechaev2013}. One-dimensional Euclidean matching problems  can be adopted therefore as toy models for different processes, depending on the properties of the cost function $c$, namely constrained Brownian processes for convex strictly increasing cost function, and folding processes for concave cost function. 

\paragraph*{The model.} Before proceeding further, let us recall some standard definitions of matching theory and rigorously specify our model. Given a generic graph $\mG=(\mV,\mE)$, with $\mV$ set of vertices and 
$\mE\subseteq\mV\times\mV$ set of edges, a \textit{matching} $\mu\subseteq\mE$ 
on $\mG$ is a subset of edges of $\mG$ such that, given two edges in $\mu$, they 
do not share a common vertex \cite{lovasz2009matching}. A matching $\mu$ is said 
to be \textit{maximal} if, for any $e\in\mE\setminus\mu$, $\mu\cup\{e\}$ is no 
longer a matching. Denoting by $|\mu|$ the cardinality of $\mu$, we define 
$\nu(\mG)\coloneqq\max_{\mu}|\mu|$ the \textit{matching number} of $\mG$, and we 
say that $\mu$ is \textit{maximum} if $|\mu|=\nu(\mG)$.  A perfect
matching (or $1$-factor) is a matching that matches all vertices
of the graph. Every perfect matching is also maximum and hence maximal. A perfect matching is a
minimum-size edge cover. We will denote by $\mM$ 
the set of perfect matchings.

Let us suppose now that a weight $w_e\geq 0$ is assigned to each edge $e\in\mE$ of the graph $\mG$. We can associate to each perfect matching $\mu$ a total cost
\begin{equation}\label{cost}
E(\mu)\coloneqq\sum_{e\in\mu}w_e,
\end{equation}
and a mean cost per edge
\begin{equation}\label{costedge}
\varepsilon(\mu)\coloneqq\frac{1}{\nu(\mG)}\sum_{e\in\mu}w_e.
\end{equation}
In the (weighted) matching problem we search for the perfect matching $\mu$ such that the total cost in Eq.~\eqref{cost} is minimized, i.e., the \textit{optimal} matching $\mu^\text{opt}$ is such that
\begin{equation}
E(\mu^\text{opt})=\min_{\mu\in\mM}E(\mu).
\end{equation}
Once the weights are assigned, the problem can be solved using efficient algorithms available in the literature \cite{Papadimitriou1998,*Kuhn,*Munkres1957,*Edmonds1965,*Jonker1987,*jungnickel2005, mezard2009information}. If the graph $\mathcal G$ is a bipartite graph, the matching problem is said to be an assignment problem.

In random matching problems, the costs $\{w_e\}_e$ are random quantities. In this case, the typical properties of the optimal solution are of a certain interest, and in particular the average optimal cost, $\overline{E}\coloneqq\overline{\min_{\mu\in\mM}E(\mu)}$, where we have denoted by $\overline{\bullet}$ the average over all possible instances of the costs set. The simplest way to introduce randomness in the problem is to consider the weights $\{w_e\}_e$ independent and identically distributed random variables \cite{Mezard1985,mezard1987spin}.  In random Euclidean matching problems, the graph $\mG$ is 
supposed to be embedded in a $d$-dimensional Euclidean domain 
$\Lambda\subseteq\R^d$ through an embedding function $\boldsymbol\Phi$, in such a way that 
each vertex $v\in\mV$ of the graph is associated to a random Euclidean point $v\mapsto 
\boldsymbol\Phi(v)\in\Lambda$. In this case, the cost $w_e$ of the edge 
$e=(u,v)$ is typically a function of the distance of the images of its 
corresponding endpoints in $\Lambda$, i.e., 
$w_{e}=c\left[\|\boldsymbol\Phi(u)-\boldsymbol\Phi(v)\|\right]$ \cite{Mezard1988,Sicuro2017}. Random Euclidean matching problems are usually more difficult to investigate, due to the presence of Euclidean correlations among the weights. The purely random case with independent edge weights plays the role of mean-field approximation of the Euclidean case \cite{Mezard1988}.

In the present paper, we will work on a specific toy model in one dimension. We will consider the case in which $\mathcal G=\mathcal K_{2N}$ complete graph with $2N$ vertices for the \remp, and $\mathcal G=\mathcal K_{N,N}$ complete bipartite graph with two partitions of the same size $N$ for the \reap. We will assume the points to be independently and uniformly generated both on the compact interval $\Lambda=[0,1]$ and on the unit circumference, and we will introduce a general class of cost functions, called here $\mathcal C$-functions, that determine a specific structure for the optimal matching solution for a given instance in the \reap. More precisely, once the two sets of points are labeled in increasing order according to their position along the line, the optimal assignment can be respresented as a periodic label shift. As particular application, we will consider the cost function $c(z)=z^p$. We will consider the finite size corrections to the average optimal cost for $p>1$ that had not been computed previously and we will also study the case $p<0$, corresponding to a long-range optimal assignment. The results obtained for the \reap will be extended to the \remp. The cost function $c(z)=z^p$ is of particular interest because the optimal assignment can be interpreted as a Gaussian stochastic process for $p>1$ and, as we will show, for $p<0$. On the other hand, for $0<p<1$ the solution is planar, and therefore the \reap is a model for the folding process mentioned above, the single parameter $p$ controlling the transition between different behaviors.

The paper is therefore subdivided in such a way to present in Sec.~\ref{sec:assignment} the case of the \reap and in Sec.~\ref{sec:matching} the case of the \remp. Finally, we will give our conclusions in Sec.~\ref{sec:conclusioni}.

\section{The random Euclidean assignment problem}\label{sec:assignment}
With reference to the definitions given in the Introduction, in the 
\textit{assignment problem}, we assume $\mG=\mathcal K_{N,N}$, the complete 
bipartite graph in which $\mV=\mV_1\cup\mV_2$, $\mV_1\cap\mV_2=\emptyset$, 
$|\mV_1|=|\mV_2|=N$. In the \reap in one 
dimension, we consider two sets of points $\Xi_N\coloneqq\{x_i\}_{i=1,\dots,N}$ 
and $\Upsilon_N\coloneqq \{y_j\}_{j=1,\dots,N}$, independently generated with 
uniform distribution density on $\Lambda=[0,1]$; we associate then the points in 
$\Xi_N$, respectively, $\Upsilon_N$, to the vertices in $\mV_1$, respectively, 
$\mV_2$. 
We will assume that the points are labeled in such a way that 
\begin{subequations}
\begin{align}
&0\leq x_1<x_2<\dots<x_{N}\leq 1,\\&0\leq y_1<y_2<\dots<y_{N}\leq 1.
\end{align}
\end{subequations}
A maximum matching $\mu\subset\mV_1\times\mV_2$ uniquely corresponds to a 
permutation $\pi\in\mathcal S_N$ of the $N$ elements 
$[N]\coloneqq\{1,\dots,N\}$, in such a way that, if $(i,j)\in\mu$, $j=\pi(i)$. 
We associate to $\pi$ a \textit{matching cost} and a \textit{mean cost per edge}, 
respectively, given by
\begin{subequations}
\begin{align}\label{costoass}
E_N(\pi)&\coloneqq \sum_{i=1}^Nw(x_i,y_{\pi(i)}),\\
\varepsilon_N(\pi)&\coloneqq\frac{1}{N}E_N(\pi).
\end{align}\end{subequations}
In the previous expressions, the cost function $w(x_i,y_j)$ depends on the points' positions $x_i$ and $y_j$. As anticipated, we will restrict ourselves to cost functions in the form \footnote{We will assume that $c(z)$ is finite almost everywhere on $\Lambda$.}
\begin{equation}\label{costoc}
w(x_i,y_j)\coloneqq c\left(|x_i-y_j|\right),\quad c\colon \Lambda\to\R.
\end{equation}
We are interested in the asymptotic behavior for $N\gg 1$ of the \textit{average} optimal (mean) cost
\begin{equation}
\varepsilon_N\coloneqq\overline{\min_{\pi}\varepsilon_N(\pi)},
\end{equation}
where we have denoted by $\overline{\bullet}$ the average over the points' 
positions. We will show that the typical properties of the solution strongly depend on the properties of the cost function $c(z)$. To study this dependency, we will assume in particular \begin{equation}\label{costocp} c(z)\coloneqq z^p,\quad p\in\R,\quad z\in\Lambda.\end{equation}As we will show below, the properties of the optimal solution will depend on the chosen value of $p$.

\subsection{On the structure of optimal matching}
We shall here discuss some general features of the optimal solution in a given instance at variance with the cost function.
\subsubsection{Preliminaries} Let us first introduce some preliminary definitions and results.
\begin{cyclic} Given a set of $n$ elements, we say that a permutation $\pi\in\mathcal S_n$ of $n$ elements belongs to $\mathcal C_n\subseteq\mathcal S_n$ if an integer number $k$ exists such that $0\leq k< n$ and
\begin{equation}
\pi(i)=i+k\Mod n,\quad i=1,\dots,n.
\end{equation}
\end{cyclic}
Observe that, for $k\neq 0$, a permutation $\pi\in\mathcal C_n$ is a cyclic permutation having one cycle only. For $k=0$ we have the identity permutation, which has $n$ cycles. The set $\mathcal C_n\subseteq\mathcal S_n$ is an Abelian subgroup corresponding to the cyclic group of the $n$ proper rotations in the plane which leave a regular polygon with $n$ vertices invariant. For $n=2$, $\mathcal C_2=\{(1)(2),(1,2)\}=\mathcal S_2$. The group $\mathcal C_3$ coincides with the alternating group of even permutations $\mathcal A_3$. 

\begin{cooriented}Given a triple of three integers $(i, j, k)\in[n]^3$, $[n]\coloneqq \{1,\dots,n\}$, we say that the ordered triple of integers $(a,b,c)$ is {\em cyclically co-oriented} with it when an even permutation $\pi\in\mathcal C_3$ exists such that $(\pi(a), \pi(b), \pi(c))$ is in the same order of $(i, j, k)$ respect to the order relation of the integers.
\end{cooriented}
The following Proposition, which appeared in Ref.~\cite{Boniolo2012}, will be fundamental in the following.
\begin{propcycle}\label{PropA1}
A permutation $\pi \in {\cal S}_n$ belongs to $\mathcal{C}_n$ if, and only if, for any triple $(i,j,k)$, the corresponding triple $(\pi(i), \pi(j), \pi(k))$ is cyclically co-oriented with $(i,j,k)$.
\end{propcycle}
\begin{proof}If $\pi \in \mathcal C_n$ all the triples $(\pi(i), \pi(j), \pi(k))$ are cyclically co-oriented with $(i,j,k)$, due to the fact that all permutations in $\mathcal C_n$ are ordering preserving. 

For the converse, let us assume that all triples are cyclically co-oriented with their image through the permutation $\pi$. Observe now that if for any couple $(i,j)\in [n]^2$ we have $\pi(i) - \pi(j) = i - j \Mod n$, then $\pi\in\mathcal C_n$. To prove this statement, we proceed by contradiction and we assume that there exists at least a couple $(i,j) \in [n]^2$ such that $\pi(i) - \pi(j) \neq i - j \Mod n$. It follows that the sequences 
\begin{subequations}
\begin{align}
I&= (i, i+1 \Mod n, \dots,j),\\
J&=(\pi(i), \pi(i)+1 \Mod n, \dots,\pi(j))
\end{align}
\end{subequations}
have not the same cardinality, and therefore there must exist a $k$ such that either $k \in I$ and  $\pi(k) \not\in  J$  or $k\not\in I$ and $\pi(k) \in J$. By consequence, the triples $(i, j, k)$ and $(\pi(i), \pi(j), \pi(k))$ are not cyclically co-oriented, that is in contradiction with the hypothesis and therefore the theorem is proved. 
\end{proof}

\subsubsection{$\mathcal C$-functions}
As anticipated, an assignment between two sets of $N$ points can be uniquely associated to a permutation of $N$ elements $\pi\in\mathcal S_N$. We will show below that the optimal permutations belongs to $\mathcal C_N\subseteq \mathcal S_N$ if the cost function $c(z)$ appearing in Eq.~\eqref{costoc} satisfies the following property.
\begin{bischera}
We shall say that a function $f:[0,1]\to\R$ is a {\em $\mathcal C$-function} if, given $0< z_1 < z_2 < 1$, for any $\eta\in (0,1-z_2)$, $\hat\eta\in (z_2,1)$
\begin{subequations}\label{bischera}
\begin{align}
f(z_2)-f(z_1)&\leq f(\eta+z_2)-f(\eta+z_1),\label{prima} \\
f(z_2)-f(z_1)&\leq f(\hat \eta-z_2)-f(\hat \eta-z_1).\label{seconda}
\end{align}\end{subequations}
\end{bischera}
Eq.~\eqref{prima} implies that $\Psi_\eta(z)=f(\eta+z)-f(z)$ is an increasing function in the interval $(0,1-\eta)$ for any value of $\eta\in(0,1)$. Moreover, if $f$ is continuous, Eq.~\eqref{prima} is equivalent to convexity (see Appendix \ref{app:convexity}).

Eq.~\eqref{seconda} implies that the function $\Phi_\eta(z)\coloneqq f(\eta-z)-f(z)$ is increasing in the interval $(0,\eta)$ for any value of $\eta\in(0,1)$. If $f$ is differentiable, this fact can be written as
\begin{subequations}\begin{equation}
 f'(\eta-z)+f'(z)\leq 0,\quad z\in(0,\eta),\quad \eta\in(0,1),
\end{equation}
which for $\eta\to 1$ becomes
\begin{equation}
 f'(1-z)+f'(z)\leq 0.
\end{equation}\label{decreasing}\end{subequations}
This implies, for example, that the convex function
\begin{equation}
 f_\alpha(x)=\left(x-\alpha\right)^2,\quad\alpha\in\R,
\end{equation}
is a $\mathcal C$-function on $\Lambda$ for $\alpha\geq\sfrac{1}{2}$ only.

\subsubsection{Optimal matching on a segment} Let us now discuss the structure of the optimal assignment on the line. We start with the following Definition to fix our nomenclature.
\begin{crossing}[Crossing and planar matching] Let us consider two sets of points $\Xi_N=\{x_i\}_{i=1,\dots,N}$ and $\Upsilon_N=\{y_i\}_{i=1,\dots,N}$ on the interval $\Lambda\coloneqq [0,1]$ and let us assume that they are labeled in such a way that if $i<j$ then $x_i<x_j$ and $y_i<y_j$. Then a matching between $\Xi_N$ and $\Upsilon_N$ is said to be \textit{planar}, or \textit{non-crossing} if, given the corresponding permutation $\pi$ and any two pairs of matched points $(x_i,y_{\pi(i)})$ and $(x_j,y_{\pi(j)})$, $i<j$, the corresponding intervals are either disjoint, $x_i<y_{\pi(i)}<x_j<y_{\pi(j)}$, or nested, $x_i<x_j<y_{\pi(j)}<y_{\pi(i)}$. The matching is otherwise said to be \textit{crossing}.
\end{crossing}
From the pictorial point of view, drawing the points on a rightward oriented horizontal line, in a planar matching it is always possible to draw semi-arcs in the upper semi-plane joining the couples of matched points which do not intersect, e.g., 
\begin{center}
\begin{tikzpicture}[scale=0.6]
\draw[style=help lines,line width=1pt,black] (0.5,0) to (8.5,0);
\node[draw,circle,inner sep=1.5pt,fill=black,text=white,label=below:{\footnotesize $x_1$}] (x1) at (1,0) {};
\node[draw,circle,inner sep=1.5pt,fill=black,text=white,label=below:{\footnotesize $x_2$}] (x2) at (2,0) {};
\node[draw,circle,inner sep=1.5pt,fill=black,text=white,label=below:{\footnotesize $x_3$}] (x3) at (4,0) {};
\node[draw,circle,inner sep=1.5pt,fill=black,text=white,label=below:{\footnotesize $x_4$}] (x4) at (6,0) {};
\node[draw,circle,inner sep=1.5pt,fill=white,label=below:{\footnotesize $y_1$}] (y1) at (3,0) {};
\node[draw,circle,inner sep=1.5pt,fill=white,label=below:{\footnotesize $y_2$}] (y2) at (5,0) {};
\node[draw,circle,inner sep=1.5pt,fill=white,label=below:{\footnotesize $y_3$}] (y3) at (7,0) {};
\node[draw,circle,inner sep=1.5pt,fill=white,label=below:{\footnotesize $y_4$}] (y4) at (8,0) {};
\draw[line width=1pt,gray] (x1) to[bend left=90] (y4);
\draw[line width=1pt,gray] (x2) to[bend left=90] (y2);
\draw[line width=1pt,gray] (x3) to[bend right=90] (y1);
\draw[line width=1pt,gray] (x4) to[bend left=90] (y3);
\end{tikzpicture}\end{center}
It is well known that if the cost function $c$ in Eq.~\eqref{costoc} is concave the optimal matching configuration is planar \cite{McCann1999,Orland2002,Nechaev2013}. 

In the following, we will restrict ourselves to two classes of convex cost functions, namely $\mathcal C$-functions and strictly increasing convex functions. We can start studying in detail the matching problem for $N=2$, considering two white points and two black points on the line. We can assume, without loss of generality, that the first point along the line is black. There are therefore $3$ possible orderings of the points, namely
\begin{center}
\begin{tikzpicture}[scale=0.6]
\draw[style=help lines,line width=1pt,black] (0.5,0) to (4.5,0);
\node[draw,circle,inner sep=1.5pt,fill=black,text=white,label=below:{\footnotesize $x_1$}] (x1) at (1,0) {};
\node[draw,circle,inner sep=1.5pt,fill=black,text=white,label=below:{\footnotesize $x_2$}] (x2) at (2,0) {};
\node[draw,circle,inner sep=1.5pt,fill=white,label=below:{\footnotesize $y_1$}] (y1) at (3,0) {};
\node[draw,circle,inner sep=1.5pt,fill=white,label=below:{\footnotesize $y_2$}] (y2) at (4,0) {};
\node (c) at (2.5,0.75) {\footnotesize A};
\end{tikzpicture}\quad
\begin{tikzpicture}[scale=0.6]
\draw[style=help lines,line width=1pt,black] (0.5,0) to (4.5,0);
\node[draw,circle,inner sep=1.5pt,fill=black,text=white,label=below:{\footnotesize $x_1$}] (x1) at (1,0) {};
\node[draw,circle,inner sep=1.5pt,fill=black,text=white,label=below:{\footnotesize $x_2$}] (x2) at (3,0) {};
\node[draw,circle,inner sep=1.5pt,fill=white,label=below:{\footnotesize $y_1$}] (y1) at (2,0) {};
\node[draw,circle,inner sep=1.5pt,fill=white,label=below:{\footnotesize $y_2$}] (y2) at (4,0) {};
\node (c) at (2.5,0.75) {\footnotesize B};
\end{tikzpicture}\quad
\begin{tikzpicture}[scale=0.6]
\draw[style=help lines,line width=1pt,black] (0.5,0) to (4.5,0);
\node[draw,circle,inner sep=1.5pt,fill=black,text=white,label=below:{\footnotesize $x_1$}] (x1) at (1,0) {};
\node[draw,circle,inner sep=1.5pt,fill=black,text=white,label=below:{\footnotesize $x_2$}] (x2) at (4,0) {};
\node[draw,circle,inner sep=1.5pt,fill=white,label=below:{\footnotesize $y_1$}] (y1) at (2,0) {};
\node[draw,circle,inner sep=1.5pt,fill=white,label=below:{\footnotesize $y_2$}] (y2) at (3,0) {};
\node (c) at (2.5,0.75) {\footnotesize C};
\end{tikzpicture}
\end{center}
Each configuration allows two possible matchings, namely $\pi(i)=i$ and $\pi(i)=i+1\pmod 2$ for $i=1,2$. The following Propositions hold.
\begin{bischeraduept}\label{PropA2}
Given the assignment problem on the interval $\Lambda$ with $N=2$, if the cost function $c(z)\colon\Lambda\to\R$ appearing in Eq.~\eqref{costoc} is a $\mathcal C$-function, then the optimal matching is the crossing one, whenever a crossing matching is available.
\end{bischeraduept}
\begin{proof}The proof of the Proposition is straightforward. In the $N=2$ case only two of the three possible configurations allow a crossing solution, namely the configuration \textsc{a} and the configuration \textsc{c}. The configuration \textsc{a} allows the two possible matchings
\begin{center}
\begin{tikzpicture}[scale=0.6]
\draw[style=help lines,line width=1pt,black] (0.5,0) to (4.5,0);
\node[draw,circle,inner sep=1.5pt,fill=black,text=white,label=below:{\footnotesize $x_1$}] (x1) at (1,0) {};
\node[draw,circle,inner sep=1.5pt,fill=black,text=white,label=below:{\footnotesize $x_2$}] (x2) at (2,0) {};
\node[draw,circle,inner sep=1.5pt,fill=white,label=below:{\footnotesize $y_1$}] (y1) at (3,0) {};
\node[draw,circle,inner sep=1.5pt,fill=white,label=below:{\footnotesize $y_2$}] (y2) at (4,0) {};
\draw[line width=1pt,gray] (x1) to[bend left=90] (y1);
\draw[line width=1pt,gray] (x2) to[bend left=90] (y2);
\node (c) at (2.5,-1) {\footnotesize A-I};
\end{tikzpicture}\quad 
\begin{tikzpicture}[scale=0.6]
\draw[style=help lines,line width=1pt,black] (0.5,0) to (4.5,0);
\node[draw,circle,inner sep=1.5pt,fill=black,text=white,label=below:{\footnotesize $x_1$}] (x1) at (1,0) {};
\node[draw,circle,inner sep=1.5pt,fill=black,text=white,label=below:{\footnotesize $x_2$}] (x2) at (2,0) {};
\node[draw,circle,inner sep=1.5pt,fill=white,label=below:{\footnotesize $y_1$}] (y1) at (3,0) {};
\node[draw,circle,inner sep=1.5pt,fill=white,label=below:{\footnotesize $y_2$}] (y2) at (4,0) {};
\draw[line width=1pt,gray] (x1) to[bend left=90] (y2);
\draw[line width=1pt,gray] (x2) to[bend left=90] (y1);
\node (c) at (2.5,-1) {\footnotesize A-II};
\end{tikzpicture}
\end{center}
With reference to the picture above,
\begin{subequations}
\begin{align}
E_{\text{I}}^{\textsc{a}}&=c(y_1-x_1)+c(y_2-x_2),\\
E_{\text{II}}^{\textsc{a}}&=c(y_2-x_1)+c(y_1-x_2).
\end{align}\end{subequations}
But, $c$ being a $\mathcal C$-function,
\begin{equation}
c(y_1-x_1)-c(y_1-x_2)\leq c(y_2-x_1)-c(y_2-x_2),
\end{equation}
where we have used Eq.~\eqref{prima} with $\eta=y_2-y_1$, and therefore we have that $E_{\text{I}}^{\textsc{a}}\leq E_{\text{II}}^{\textsc{a}}$. Observe that this result for the \textsc{a} case holds for any continuous convex function $c(z)$ on $\Lambda$. Similarly, the configuration \textsc{c} allows the two possible matchings
\begin{center}
\begin{tikzpicture}[scale=0.6]
\draw[style=help lines,line width=1pt,black] (0.5,0) to (4.5,0);
\node[draw,circle,inner sep=1.5pt,fill=black,text=white,label=below:{\footnotesize $x_1$}] (x1) at (1,0) {};
\node[draw,circle,inner sep=1.5pt,fill=black,text=white,label=below:{\footnotesize $x_2$}] (x2) at (4,0) {};
\node[draw,circle,inner sep=1.5pt,fill=white,label=below:{\footnotesize $y_1$}] (y1) at (2,0) {};
\node[draw,circle,inner sep=1.5pt,fill=white,label=below:{\footnotesize $y_2$}] (y2) at (3,0) {};
\draw[line width=1pt,gray] (x1) to[bend left=90] (y1);
\draw[line width=1pt,gray] (x2) to[bend right=90] (y2);
\node (c) at (2.5,-1) {\footnotesize C-I};
\end{tikzpicture}\quad 
\begin{tikzpicture}[scale=0.6]
\draw[style=help lines,line width=1pt,black] (0.5,0) to (4.5,0);
\node[draw,circle,inner sep=1.5pt,fill=black,text=white,label=below:{\footnotesize $x_1$}] (x1) at (1,0) {};
\node[draw,circle,inner sep=1.5pt,fill=black,text=white,label=below:{\footnotesize $x_2$}] (x2) at (4,0) {};
\node[draw,circle,inner sep=1.5pt,fill=white,label=below:{\footnotesize $y_1$}] (y1) at (2,0) {};
\node[draw,circle,inner sep=1.5pt,fill=white,label=below:{\footnotesize $y_2$}] (y2) at (3,0) {};
\draw[line width=1pt,gray] (x1) to[bend left=90] (y2);
\draw[line width=1pt,gray] (x2) to[bend right=90] (y1);
\node (c) at (2.5,-1) {\footnotesize C-II};
\end{tikzpicture}
\end{center}
with corresponding costs
\begin{subequations}
\begin{align}
E_{\text{I}}^{\textsc{c}}&=c(y_1-x_1)+c(x_2-y_2),\\
E_{\text{II}}^{\textsc{c}}&=c(y_2-x_1)+c(x_2-y_1).
\end{align}\end{subequations}
Again, $c$ being a $\mathcal C$-function,
\begin{equation}
c(y_1-x_1)-c(y_1-x_2)\leq c(y_2-x_1)-c(y_2-x_2),
\end{equation}
where we have used Eq.~\eqref{seconda} with $\eta=y_2-x_1$, and therefore $E_{\text{II}}^{\textsc{c}}\leq E_{\text{I}}^{\textsc{c}}$. This completes the proof.
\end{proof}

\begin{pgone}[Convex increasing function]Given the assignment problem on the interval $\Lambda$, if the cost function $c(z)\colon\Lambda\to\R$ appearing in Eq.~\eqref{costoc} is a strictly convex increasing function, then the optimal permutation in the $N=2$ case is the identity permutation $\pi(i)=i$, $i=1,2$.\label{PropA3}
\end{pgone}
\begin{proof}Let us first observe that a strictly increasing function on $\Lambda$ cannot be a $\mathcal C$-function, due to property in Eq.~\eqref{decreasing}. Moreover, due to the strict convexity hypothesis, $c(0)$ must be finite. We have already proved that, if $c(z)$ is convex, the optimal matching in the configuration \textsc{a} is the ordered one. If we consider now the configuration \textsc{b}, we have to evaluate two possible matchings, namely
\begin{center}
\begin{tikzpicture}[scale=0.6]
\draw[style=help lines,line width=1pt,black] (0.5,0) to (4.5,0);
\node[draw,circle,inner sep=1.5pt,fill=black,text=white,label=below:{\footnotesize $x_1$}] (x1) at (1,0) {};
\node[draw,circle,inner sep=1.5pt,fill=black,text=white,label=below:{\footnotesize $x_2$}] (x2) at (3,0) {};
\node[draw,circle,inner sep=1.5pt,fill=white,label=below:{\footnotesize $y_1$}] (y1) at (2,0) {};
\node[draw,circle,inner sep=1.5pt,fill=white,label=below:{\footnotesize $y_2$}] (y2) at (4,0) {};
\draw[line width=1pt,gray] (x1) to[bend left=90] (y1);
\draw[line width=1pt,gray] (x2) to[bend left=90] (y2);
\node (c) at (2.5,-1) {\footnotesize B-I};
\end{tikzpicture}\quad 
\begin{tikzpicture}[scale=0.6]
\draw[style=help lines,line width=1pt,black] (0.5,0) to (4.5,0);
\node[draw,circle,inner sep=1.5pt,fill=black,text=white,label=below:{\footnotesize $x_1$}] (x1) at (1,0) {};
\node[draw,circle,inner sep=1.5pt,fill=black,text=white,label=below:{\footnotesize $x_2$}] (x2) at (3,0) {};
\node[draw,circle,inner sep=1.5pt,fill=white,label=below:{\footnotesize $y_1$}] (y1) at (2,0) {};
\node[draw,circle,inner sep=1.5pt,fill=white,label=below:{\footnotesize $y_2$}] (y2) at (4,0) {};
\draw[line width=1pt,gray] (x1) to[bend left=90] (y2);
\draw[line width=1pt,gray] (x2) to[bend right=90] (y1);
\node (c) at (2.5,-1) {\footnotesize B-II};
\end{tikzpicture}
\end{center}
In this case,
\begin{subequations}
\begin{align}
E_{\text{I}}^{\textsc{b}}&=c(y_1-x_1)+c(x_2-y_2),\\
E_{\text{II}}^{\textsc{b}}&=c(y_2-x_1)+c(x_2-y_1).
\end{align}\end{subequations}
By hypothesis, the quantity $c(y_1-x_1)-c(x_2-y_1)$ is monotonously increasing respect to the variable $y_1\in(x_1,x_2)$ and therefore
\begin{equation}
c(y_1-x_1)-c(x_2-y_1)\leq c(x_2-x_1)-c(0).
\end{equation}
Similarly, convexity implies that $c(y_2-x_1)-c(y_2-x_2)$ is increasing in its argument $y_2$ and therefore
\begin{equation}
c(y_2-x_1)-c(y_2-x_2)\geq c(x_2-x_1)-c(0).
\end{equation}
It follows that
\begin{multline}
c(y_1-x_1)-c(x_2-y_1)\leq c(y_2-x_1)-c(y_2-x_2)\Rightarrow\\
\Rightarrow E_{\text{I}}^{\textsc{b}}\leq E_{\text{II}}^{\textsc{b}}.
\end{multline}
If we finally consider the configuration \textsc{c}, we have
\begin{equation}
E_{\text{I}}^{\textsc{c}}=c(y_1-x_1)+c(x_2-y_2)\leq 
c(y_2-x_1)+c(x_2-y_1)=E_{\text{II}}^{\textsc{c}},
\end{equation}
due to the fact that the cost function $c$ is strictly increasing.\end{proof}

If we consider now the $N=3$ case, we can derive the following fundamental Lemma.
\begin{lemma} Let $c(z)\colon\Lambda\to\R$, cost function for the assignment problem on the interval $\Lambda$, be a $\mathcal C$-function. Then the optimal permutation $\pi$ in the assignment problem for $N=3$ belongs to the set $\mathcal C_3$.\label{Lemma}
\end{lemma}
\begin{proof}For the proof of this Lemma, observe that, by Proposition \ref{PropA2}, for $N=2$ the crossing matching is the optimal one whenever it is available, i.e., we can obtain the optimal matching maximizing the number of crossings given a certain configuration. In the assignment problem with $N=3$ ``white'' points to be matched with $N=3$ ``black'' points on the line, there are $\frac{1}{2}\binom{6}{3}=10$ distinct configurations assuming that the first point is always of a given type, e.g., black. Each configuration allows six possible matching permutations. Let us start with the following three configurations,
\begin{center}
\begin{tikzpicture}[scale=0.45]
\draw[style=help lines,line width=1pt,black] (0.5,0) to (6.5,0);
\node[draw,circle,inner sep=1.5pt,fill=black,text=white] (x1) at (1,0) {};
\node[draw,circle,inner sep=1.5pt,fill=black,text=white] (x2) at (2,0) {};
\node[draw,circle,inner sep=1.5pt,fill=black,text=white] (x3) at (3,0) {};
\node[draw,circle,inner sep=1.5pt,fill=white] (y1) at (4,0) {};
\node[draw,circle,inner sep=1.5pt,fill=white] (y2) at (5,0) {};
\node[draw,circle,inner sep=1.5pt,fill=white] (y3) at (6,0) {};
\node (c) at (3.5,1) {\footnotesize A};
\end{tikzpicture}\ 
\begin{tikzpicture}[scale=0.45]
\draw[style=help lines,line width=1pt,black] (0.5,0) to (6.5,0);
\node[draw,circle,inner sep=1.5pt,fill=black,text=white] (x1) at (1,0) {};
\node[draw,circle,inner sep=1.5pt,fill=black,text=white] (x2) at (2,0) {};
\node[draw,circle,inner sep=1.5pt,fill=black,text=white] (x3) at (6,0) {};
\node[draw,circle,inner sep=1.5pt,fill=white] (y1) at (3,0) {};
\node[draw,circle,inner sep=1.5pt,fill=white] (y2) at (4,0) {};
\node[draw,circle,inner sep=1.5pt,fill=white] (y3) at (5,0) {};
\node (c) at (3.5,1) {\footnotesize B};
\end{tikzpicture}\ 
\begin{tikzpicture}[scale=0.45]
\draw[style=help lines,line width=1pt,black] (0.5,0) to (6.5,0);
\node[draw,circle,inner sep=1.5pt,fill=black,text=white] (x1) at (1,0) {};
\node[draw,circle,inner sep=1.5pt,fill=black,text=white] (x2) at (5,0) {};
\node[draw,circle,inner sep=1.5pt,fill=black,text=white] (x3) at (6,0) {};
\node[draw,circle,inner sep=1.5pt,fill=white] (y1) at (2,0) {};
\node[draw,circle,inner sep=1.5pt,fill=white] (y2) at (3,0) {};
\node[draw,circle,inner sep=1.5pt,fill=white] (y3) at (4,0) {};
\node (c) at (3.5,1) {\footnotesize C};
\end{tikzpicture}
\end{center}
All of them can be pictorially represented, in a compact way, on a 
circumference as
\begin{center}
\begin{tikzpicture}[thick]
\newdimen\R
\R=0.5cm
\draw (0,0) circle (\R);
\node[draw,circle,inner sep=1.5pt,fill=black] (x1) at (60:\R) {};
\node[draw,circle,inner sep=1.5pt,fill=black] (x2) at (120:\R) {};
\node[draw,circle,inner sep=1.5pt,fill=black] (x3) at (180:\R) {};
\node[draw,circle,inner sep=1.5pt,fill=white] (y1) at (240:\R) {};
\node[draw,circle,inner sep=1.5pt,fill=white] (y2) at (300:\R) {};
\node[draw,circle,inner sep=1.5pt,fill=white] (y3) at (360:\R) {};
\end{tikzpicture}
\end{center}
For example, we can represent
\begin{center}
\begin{tikzpicture}[scale=0.45]
\draw[style=help lines,line width=1pt,black] (0.5,0) to (6.5,0);
\node[draw,circle,inner sep=1.5pt,fill=black,text=white,label=below:{\footnotesize $1$}] (x1) at (1,0) {};
\node[draw,circle,inner sep=1.5pt,fill=black,text=white,label=below:{\footnotesize $2$}] (x2) at (2,0) {};
\node[draw,circle,inner sep=1.5pt,fill=black,text=white,label=below:{\footnotesize $3$}] (x3) at (3,0) {};
\node[draw,circle,inner sep=1.5pt,fill=white,label=below:{\footnotesize $1$}] (y1) at (4,0) {};
\node[draw,circle,inner sep=1.5pt,fill=white,label=below:{\footnotesize $2$}] (y2) at (5,0) {};
\node[draw,circle,inner sep=1.5pt,fill=white,label=below:{\footnotesize $3$}] (y3) at (6,0) {};
\node (c) at (3.5,1) {\footnotesize A};
\draw [-stealth,thick] (7,0) to (9,0);
\begin{scope}[scale=2.22,thick,shift={(5,0)}]
\newdimen\R
\R=0.5cm
\draw (0,0) circle (\R);
\node[draw,circle,inner sep=1pt,fill=black,text=white] (x1) at (60:\R) {\tiny $1$};
\node[draw,circle,inner sep=1pt,fill=black,text=white] (x2) at (120:\R) {\tiny $2$};
\node[draw,circle,inner sep=1pt,fill=black,text=white] (x3) at (180:\R) {\tiny $3$};
\node[draw,circle,inner sep=1pt,fill=white] (y1) at (240:\R) {\tiny $1$};
\node[draw,circle,inner sep=1pt,fill=white] (y2) at (300:\R) {\tiny $2$};
\node[draw,circle,inner sep=1pt,fill=white] (y3) at (360:\R) {\tiny $3$};
\node[draw,circle,inner sep=1pt,fill=white,white] (z) at (30:\R) {};
\end{scope}\end{tikzpicture}
\end{center}
In this representation, each matching on a configuration on the line corresponds 
to a set of three chords joining the points on the circumference, and a 
crossing in a matching corresponds to an actual crossing between chords. We will 
use this representation to evaluate all configurations in a compact way. By 
Proposition \ref{PropA2}, we can order the possible matchings using the fact 
that, applying a transposition to a given permutation, if a new crossing 
appears, then the new permutation has a lower cost. Pictorially we can order the 
six permutations as
\begin{center}
\begin{tikzpicture}[thick]
\begin{scope}
\newdimen\R
\R=0.5cm
\draw (0,0) circle (\R);
\node[draw,circle,inner sep=1.5pt,fill=black] (x1) at (60:\R) {};
\node[draw,circle,inner sep=1.5pt,fill=black] (x2) at (120:\R) {};
\node[draw,circle,inner sep=1.5pt,fill=black] (x3) at (180:\R) {};
\node[draw,circle,inner sep=1.5pt,fill=white] (y1) at (240:\R) {};
\node[draw,circle,inner sep=1.5pt,fill=white] (y2) at (300:\R) {};
\node[draw,circle,inner sep=1.5pt,fill=white] (y3) at (360:\R) {};
\draw[line width=1pt,gray] (x1) to[out=240,in=180] (y3);
\draw[line width=1pt,gray] (x2) to[out=300,in=120] (y2);
\draw[line width=1pt,gray] (x3) to[out=0,in=60] (y1);
\end{scope}
\draw [-stealth] (0.6,0) to [out=0,in=90] (1,-0.4);
\draw [-stealth] (-0.6,0) to [out=180,in=90] (-1,-0.4);
\begin{scope}[shift={(-1,-1)}]
\newdimen\R
\R=0.5cm
\draw (0,0) circle (\R);
\node[draw,circle,inner sep=1.5pt,fill=black] (x1) at (60:\R) {};
\node[draw,circle,inner sep=1.5pt,fill=black] (x2) at (120:\R) {};
\node[draw,circle,inner sep=1.5pt,fill=black] (x3) at (180:\R) {};
\node[draw,circle,inner sep=1.5pt,fill=white] (y1) at (240:\R) {};
\node[draw,circle,inner sep=1.5pt,fill=white] (y2) at (300:\R) {};
\node[draw,circle,inner sep=1.5pt,fill=white] (y3) at (360:\R) {};
\draw[line width=1pt,gray] (x1) to[out=240,in=120] (y2);
\draw[line width=1pt,gray] (x2) to[out=300,in=180] (y3);
\draw[line width=1pt,gray] (x3) to[out=0,in=60] (y1);
\end{scope}
\begin{scope}[shift={(1,-1)}]
\newdimen\R
\R=0.5cm
\draw (0,0) circle (\R);
\node[draw,circle,inner sep=1.5pt,fill=black] (x1) at (60:\R) {};
\node[draw,circle,inner sep=1.5pt,fill=black] (x2) at (120:\R) {};
\node[draw,circle,inner sep=1.5pt,fill=black] (x3) at (180:\R) {};
\node[draw,circle,inner sep=1.5pt,fill=white] (y1) at (240:\R) {};
\node[draw,circle,inner sep=1.5pt,fill=white] (y2) at (300:\R) {};
\node[draw,circle,inner sep=1.5pt,fill=white] (y3) at (360:\R) {};
\draw[line width=1pt,gray] (x1) to[out=240,in=180] (y3);
\draw[line width=1pt,gray] (x2) to[out=300,in=60] (y1);
\draw[line width=1pt,gray] (x3) to[out=0,in=120] (y2);
\end{scope}
\draw [-stealth] (1,-1.6) to (1,-1.9);
\draw [-stealth] (-1,-1.6) to (-1,-1.9);
\begin{scope}[shift={(-1,-2.5)}]
\newdimen\R
\R=0.5cm
\draw (0,0) circle (\R);
\node[draw,circle,inner sep=1.5pt,fill=black] (x1) at (60:\R) {};
\node[draw,circle,inner sep=1.5pt,fill=black] (x2) at (120:\R) {};
\node[draw,circle,inner sep=1.5pt,fill=black] (x3) at (180:\R) {};
\node[draw,circle,inner sep=1.5pt,fill=white] (y1) at (240:\R) {};
\node[draw,circle,inner sep=1.5pt,fill=white] (y2) at (300:\R) {};
\node[draw,circle,inner sep=1.5pt,fill=white] (y3) at (360:\R) {};
\draw[line width=1pt,gray] (x1) to[out=240,in=60] (y1);
\draw[line width=1pt,gray] (x2) to[out=300,in=180] (y3);
\draw[line width=1pt,gray] (x3) to[out=0,in=120] (y2);
\end{scope}
\begin{scope}[shift={(1,-2.5)}]
\newdimen\R
\R=0.5cm
\draw (0,0) circle (\R);
\node[draw,circle,inner sep=1.5pt,fill=black] (x1) at (60:\R) {};
\node[draw,circle,inner sep=1.5pt,fill=black] (x2) at (120:\R) {};
\node[draw,circle,inner sep=1.5pt,fill=black] (x3) at (180:\R) {};
\node[draw,circle,inner sep=1.5pt,fill=white] (y1) at (240:\R) {};
\node[draw,circle,inner sep=1.5pt,fill=white] (y2) at (300:\R) {};
\node[draw,circle,inner sep=1.5pt,fill=white] (y3) at (360:\R) {};
\draw[line width=1pt,gray] (x1) to[out=240,in=120] (y2);
\draw[line width=1pt,gray] (x2) to[out=300,in=60] (y1);
\draw[line width=1pt,gray] (x3) to[out=0,in=180] (y3);
\end{scope}
\draw [-stealth] (-1,-3.1) to [out=-90,in=180] (-0.6,-3.5);
\draw [-stealth] (1,-3.1) to [out=-90,in=0] (0.6,-3.5);
\begin{scope}[shift={(0,-3.5)}]
\newdimen\R
\R=0.5cm
\draw (0,0) circle (\R);
\node[draw,circle,inner sep=1.5pt,fill=black] (x1) at (60:\R) {};
\node[draw,circle,inner sep=1.5pt,fill=black] (x2) at (120:\R) {};
\node[draw,circle,inner sep=1.5pt,fill=black] (x3) at (180:\R) {};
\node[draw,circle,inner sep=1.5pt,fill=white] (y1) at (240:\R) {};
\node[draw,circle,inner sep=1.5pt,fill=white] (y2) at (300:\R) {};
\node[draw,circle,inner sep=1.5pt,fill=white] (y3) at (360:\R) {};
\draw[line width=1pt,gray] (x1) to (y1);
\draw[line width=1pt,gray] (x2) to (y2);
\draw[line width=1pt,gray] (x3) to (y3);
\end{scope}
\end{tikzpicture}
\end{center}
where the arrows denote the transition from a given matching to another one with lower cost, as a consequence of a single transposition. In this case the final, and cheapest, matching corresponds to the optimal permutation $\pi(i)=i$ for the configuration \textsc{a}, $\pi(i)=i+1\Mod 3$ for the configuration \textsc{b}, and $\pi(i)=i+2\Mod 3$ for the configuration \textsc{c}.

The following six configurations of points
\begin{center}
\begin{tikzpicture}[scale=0.45]
\draw[style=help lines,line width=1pt,black] (0.5,0) to (6.5,0);
\node[draw,circle,inner sep=1.5pt,fill=black,text=white] (x1) at (1,0) {};
\node[draw,circle,inner sep=1.5pt,fill=black,text=white] (x2) at (2,0) {};
\node[draw,circle,inner sep=1.5pt,fill=black,text=white] (x3) at (4,0) {};
\node[draw,circle,inner sep=1.5pt,fill=white] (y1) at (3,0) {};
\node[draw,circle,inner sep=1.5pt,fill=white] (y2) at (5,0) {};
\node[draw,circle,inner sep=1.5pt,fill=white] (y3) at (6,0) {};
\node (c) at (3.5,1) {\footnotesize D};
\end{tikzpicture}\ 
\begin{tikzpicture}[scale=0.45]
\draw[style=help lines,line width=1pt,black] (0.5,0) to (6.5,0);
\node[draw,circle,inner sep=1.5pt,fill=black,text=white] (x1) at (1,0) {};
\node[draw,circle,inner sep=1.5pt,fill=black,text=white] (x2) at (2,0) {};
\node[draw,circle,inner sep=1.5pt,fill=black,text=white] (x3) at (5,0) {};
\node[draw,circle,inner sep=1.5pt,fill=white] (y1) at (3,0) {};
\node[draw,circle,inner sep=1.5pt,fill=white] (y2) at (4,0) {};
\node[draw,circle,inner sep=1.5pt,fill=white] (y3) at (6,0) {};
\node (c) at (3.5,1) {\footnotesize E};
\end{tikzpicture}\ 
\begin{tikzpicture}[scale=0.45]
\draw[style=help lines,line width=1pt,black] (0.5,0) to (6.5,0);
\node[draw,circle,inner sep=1.5pt,fill=black,text=white] (x1) at (1,0) {};
\node[draw,circle,inner sep=1.5pt,fill=black,text=white] (x2) at (3,0) {};
\node[draw,circle,inner sep=1.5pt,fill=black,text=white] (x3) at (4,0) {};
\node[draw,circle,inner sep=1.5pt,fill=white] (y1) at (2,0) {};
\node[draw,circle,inner sep=1.5pt,fill=white] (y2) at (5,0) {};
\node[draw,circle,inner sep=1.5pt,fill=white] (y3) at (6,0) {};
\node (c) at (3.5,1) {\footnotesize F};
\end{tikzpicture}\\\vspace{0.2cm}
\begin{tikzpicture}[scale=0.45]
\draw[style=help lines,line width=1pt,black] (0.5,0) to (6.5,0);
\node[draw,circle,inner sep=1.5pt,fill=black,text=white] (x1) at (1,0) {};
\node[draw,circle,inner sep=1.5pt,fill=black,text=white] (x2) at (4,0) {};
\node[draw,circle,inner sep=1.5pt,fill=black,text=white] (x3) at (5,0) {};
\node[draw,circle,inner sep=1.5pt,fill=white] (y1) at (2,0) {};
\node[draw,circle,inner sep=1.5pt,fill=white] (y2) at (3,0) {};
\node[draw,circle,inner sep=1.5pt,fill=white] (y3) at (6,0) {};
\node (c) at (3.5,1) {\footnotesize G};
\end{tikzpicture}\ 
\begin{tikzpicture}[scale=0.45]
\draw[style=help lines,line width=1pt,black] (0.5,0) to (6.5,0);
\node[draw,circle,inner sep=1.5pt,fill=black,text=white] (x1) at (1,0) {};
\node[draw,circle,inner sep=1.5pt,fill=black,text=white] (x2) at (3,0) {};
\node[draw,circle,inner sep=1.5pt,fill=black,text=white] (x3) at (6,0) {};
\node[draw,circle,inner sep=1.5pt,fill=white] (y1) at (2,0) {};
\node[draw,circle,inner sep=1.5pt,fill=white] (y2) at (4,0) {};
\node[draw,circle,inner sep=1.5pt,fill=white] (y3) at (5,0) {};
\node (c) at (3.5,1) {\footnotesize H};
\end{tikzpicture}\ 
\begin{tikzpicture}[scale=0.45]
\draw[style=help lines,line width=1pt,black] (0.5,0) to (6.5,0);
\node[draw,circle,inner sep=1.5pt,fill=black,text=white] (x1) at (1,0) {};
\node[draw,circle,inner sep=1.5pt,fill=black,text=white] (x2) at (4,0) {};
\node[draw,circle,inner sep=1.5pt,fill=black,text=white] (x3) at (6,0) {};
\node[draw,circle,inner sep=1.5pt,fill=white] (y1) at (2,0) {};
\node[draw,circle,inner sep=1.5pt,fill=white] (y2) at (3,0) {};
\node[draw,circle,inner sep=1.5pt,fill=white] (y3) at (5,0) {};
\node (c) at (3.5,1) {\footnotesize I};
\end{tikzpicture}
\end{center}
can be treated similarly. In particular, all six configurations can be represented as
\begin{center}
\begin{tikzpicture}[thick]
\newdimen\R
\R=0.5cm
\draw (0,0) circle (\R);
\node[draw,circle,inner sep=1.5pt,fill=black] (x1) at (60:\R) {};
\node[draw,circle,inner sep=1.5pt,fill=black] (x2) at (120:\R) {};
\node[draw,circle,inner sep=1.5pt,fill=black] (x3) at (300:\R) {};
\node[draw,circle,inner sep=1.5pt,fill=white] (y1) at (180:\R) {};
\node[draw,circle,inner sep=1.5pt,fill=white] (y2) at (240:\R) {};
\node[draw,circle,inner sep=1.5pt,fill=white] (y3) at (360:\R) {};
\end{tikzpicture}
\end{center}
We have therefore
\begin{center}
\begin{tikzpicture}[thick]
\begin{scope}
\newdimen\R
\R=0.5cm
\draw (0,0) circle (\R);
\node[draw,circle,inner sep=1.5pt,fill=black] (x2) at (60:\R) {};
\node[draw,circle,inner sep=1.5pt,fill=black] (x1) at (120:\R) {};
\node[draw,circle,inner sep=1.5pt,fill=black] (x3) at (300:\R) {};
\node[draw,circle,inner sep=1.5pt,fill=white] (y3) at (180:\R) {};
\node[draw,circle,inner sep=1.5pt,fill=white] (y2) at (240:\R) {};
\node[draw,circle,inner sep=1.5pt,fill=white] (y1) at (360:\R) {};
\draw[line width=1pt,gray] (x1) to[out=300,in=0] (y3);
\draw[line width=1pt,gray] (x2) to[out=240,in=60] (y2);
\draw[line width=1pt,gray] (x3) to[out=120,in=180] (y1);
\end{scope}
\begin{scope}[shift={(1.5,0)}]
\newdimen\R
\R=0.5cm
\draw (0,0) circle (\R);
\node[draw,circle,inner sep=1.5pt,fill=black] (x2) at (60:\R) {};
\node[draw,circle,inner sep=1.5pt,fill=black] (x1) at (120:\R) {};
\node[draw,circle,inner sep=1.5pt,fill=black] (x3) at (300:\R) {};
\node[draw,circle,inner sep=1.5pt,fill=white] (y3) at (180:\R) {};
\node[draw,circle,inner sep=1.5pt,fill=white] (y2) at (240:\R) {};
\node[draw,circle,inner sep=1.5pt,fill=white] (y1) at (360:\R) {};
\draw[line width=1pt,gray] (x1) to[out=300,in=0] (y3);
\draw[line width=1pt,gray] (x2) to[out=240,in=180] (y1);
\draw[line width=1pt,gray] (x3) to[out=120,in=60] (y2);
\end{scope}
\begin{scope}[shift={(3,0)}]
\newdimen\R
\R=0.5cm
\draw (0,0) circle (\R);
\node[draw,circle,inner sep=1.5pt,fill=black] (x2) at (60:\R) {};
\node[draw,circle,inner sep=1.5pt,fill=black] (x1) at (120:\R) {};
\node[draw,circle,inner sep=1.5pt,fill=black] (x3) at (300:\R) {};
\node[draw,circle,inner sep=1.5pt,fill=white] (y3) at (180:\R) {};
\node[draw,circle,inner sep=1.5pt,fill=white] (y2) at (240:\R) {};
\node[draw,circle,inner sep=1.5pt,fill=white] (y1) at (360:\R) {};
\draw[line width=1pt,gray] (x1) to[out=300,in=180] (y1);
\draw[line width=1pt,gray] (x2) to[out=240,in=0] (y3);
\draw[line width=1pt,gray] (x3) to[out=120,in=60] (y2);
\end{scope}
\draw [-stealth] (0,-0.6) to (0,-0.9);
\draw [-stealth] (3,-0.6) to (3,-0.9);
\draw [-stealth] (1.5,-0.6) to (1.5,-0.9);
\draw [-stealth] (2.1,-1.5) to (2.4,-1.5);
\begin{scope}[shift={(0,-1.5)}]
\newdimen\R
\R=0.5cm
\draw (0,0) circle (\R);
\node[draw,circle,inner sep=1.5pt,fill=black] (x2) at (60:\R) {};
\node[draw,circle,inner sep=1.5pt,fill=black] (x1) at (120:\R) {};
\node[draw,circle,inner sep=1.5pt,fill=black] (x3) at (300:\R) {};
\node[draw,circle,inner sep=1.5pt,fill=white] (y3) at (180:\R) {};
\node[draw,circle,inner sep=1.5pt,fill=white] (y2) at (240:\R) {};
\node[draw,circle,inner sep=1.5pt,fill=white] (y1) at (360:\R) {};
\draw[line width=1pt,gray] (x1) to[out=300,in=60] (y2);
\draw[line width=1pt,gray] (x2) to[out=240,in=0] (y3);
\draw[line width=1pt,gray] (x3) to[out=120,in=180] (y1);
\end{scope}
\begin{scope}[shift={(1.5,-1.5)}]
\newdimen\R
\R=0.5cm
\draw (0,0) circle (\R);
\node[draw,circle,inner sep=1.5pt,fill=black] (x2) at (60:\R) {};
\node[draw,circle,inner sep=1.5pt,fill=black] (x1) at (120:\R) {};
\node[draw,circle,inner sep=1.5pt,fill=black] (x3) at (300:\R) {};
\node[draw,circle,inner sep=1.5pt,fill=white] (y3) at (180:\R) {};
\node[draw,circle,inner sep=1.5pt,fill=white] (y2) at (240:\R) {};
\node[draw,circle,inner sep=1.5pt,fill=white] (y1) at (360:\R) {};
\draw[line width=1pt,gray] (x1) to[out=300,in=60] (y2);
\draw[line width=1pt,gray] (x2) to[out=240,in=180] (y1);
\draw[line width=1pt,gray] (x3) to[out=120,in=0] (y3);
\end{scope}
\begin{scope}[shift={(3,-1.5)}]
\newdimen\R
\R=0.5cm
\draw (0,0) circle (\R);
\node[draw,circle,inner sep=1.5pt,fill=black] (x2) at (60:\R) {};
\node[draw,circle,inner sep=1.5pt,fill=black] (x1) at (120:\R) {};
\node[draw,circle,inner sep=1.5pt,fill=black] (x3) at (300:\R) {};
\node[draw,circle,inner sep=1.5pt,fill=white] (y3) at (180:\R) {};
\node[draw,circle,inner sep=1.5pt,fill=white] (y2) at (240:\R) {};
\node[draw,circle,inner sep=1.5pt,fill=white] (y1) at (360:\R) {};
\draw[line width=1pt,gray] (x1) to[out=300,in=180] (y1);
\draw[line width=1pt,gray] (x2) to[out=240,in=60] (y2);
\draw[line width=1pt,gray] (x3) to[out=120,in=0] (y3);
\end{scope}
\end{tikzpicture}
\end{center}
This implies that, for each configuration, there are two possible optimal permutations, namely $\pi(i)=i$ or $\pi(i)=i+1\Mod 3$ for the configurations \textsc{d}, \textsc{e}, \textsc{f}, and $\pi(i)=i+1\Mod 3$ or $\pi(i)=i+2\Mod 3$ for the configurations \textsc{g}, \textsc{h}, \textsc{i}.

Finally, let us consider the configuration
\begin{center}
\begin{tikzpicture}[scale=0.45]
\draw[style=help lines,line width=1pt,black] (0.5,0) to (6.5,0);
\node[draw,circle,inner sep=1.5pt,fill=black,text=white] (x1) at (1,0) {};
\node[draw,circle,inner sep=1.5pt,fill=black,text=white] (x2) at (3,0) {};
\node[draw,circle,inner sep=1.5pt,fill=black,text=white] (x3) at (5,0) {};
\node[draw,circle,inner sep=1.5pt,fill=white] (y1) at (2,0) {};
\node[draw,circle,inner sep=1.5pt,fill=white] (y2) at (4,0) {};
\node[draw,circle,inner sep=1.5pt,fill=white] (y3) at (6,0) {};
\node (c) at (3.5,1) {\footnotesize J};
\draw [-stealth,thick] (7,0) to (9,0);
\begin{scope}[scale=2.22,thick,shift={(5,0)}]
\newdimen\R
\R=0.5cm
\draw (0,0) circle (\R);
\node[draw,circle,inner sep=1.5pt,fill=black] (x1) at (60:\R) {};
\node[draw,circle,inner sep=1.5pt,fill=black] (x2) at (180:\R) {};
\node[draw,circle,inner sep=1.5pt,fill=black] (x3) at (300:\R) {};
\node[draw,circle,inner sep=1.5pt,fill=white] (y1) at (120:\R) {};
\node[draw,circle,inner sep=1.5pt,fill=white] (y2) at (240:\R) {};
\node[draw,circle,inner sep=1.5pt,fill=white] (y3) at (360:\R) {};
\node[draw,circle,inner sep=1pt,fill=white,white] (z) at (30:\R) {};
\end{scope}\end{tikzpicture}
\end{center}
In this case we have
\begin{center}
\begin{tikzpicture}[thick]
\begin{scope}
\newdimen\R
\R=0.5cm
\draw (0,0) circle (\R);
\node[draw,circle,inner sep=1.5pt,fill=black] (x1) at (60:\R) {};
\node[draw,circle,inner sep=1.5pt,fill=black] (x2) at (180:\R) {};
\node[draw,circle,inner sep=1.5pt,fill=black] (x3) at (300:\R) {};
\node[draw,circle,inner sep=1.5pt,fill=white] (y1) at (120:\R) {};
\node[draw,circle,inner sep=1.5pt,fill=white] (y2) at (240:\R) {};
\node[draw,circle,inner sep=1.5pt,fill=white] (y3) at (360:\R) {};
\draw[line width=1pt,gray] (x1) to[out=240,in=180] (y3);
\draw[line width=1pt,gray] (x2) to[out=0,in=60] (y2);
\draw[line width=1pt,gray] (x3) to[out=120,in=300] (y1);
\node[draw,circle,inner sep=1pt,fill=white,white] (z) at (30:\R) {};
\end{scope}
\begin{scope}[shift={(1.5,0)}]
\newdimen\R
\R=0.5cm
\draw (0,0) circle (\R);
\node[draw,circle,inner sep=1.5pt,fill=black] (x1) at (60:\R) {};
\node[draw,circle,inner sep=1.5pt,fill=black] (x2) at (180:\R) {};
\node[draw,circle,inner sep=1.5pt,fill=black] (x3) at (300:\R) {};
\node[draw,circle,inner sep=1.5pt,fill=white] (y1) at (120:\R) {};
\node[draw,circle,inner sep=1.5pt,fill=white] (y2) at (240:\R) {};
\node[draw,circle,inner sep=1.5pt,fill=white] (y3) at (360:\R) {};
\draw[line width=1pt,gray] (x1) to[out=240,in=300] (y1);
\draw[line width=1pt,gray] (x2) to[out=0,in=180] (y3);
\draw[line width=1pt,gray] (x3) to[out=120,in=60] (y2);
\node[draw,circle,inner sep=1pt,fill=white,white] (z) at (30:\R) {};
\end{scope}
\begin{scope}[shift={(3,0)}]
\newdimen\R
\R=0.5cm
\draw (0,0) circle (\R);
\node[draw,circle,inner sep=1.5pt,fill=black] (x1) at (60:\R) {};
\node[draw,circle,inner sep=1.5pt,fill=black] (x2) at (180:\R) {};
\node[draw,circle,inner sep=1.5pt,fill=black] (x3) at (300:\R) {};
\node[draw,circle,inner sep=1.5pt,fill=white] (y1) at (120:\R) {};
\node[draw,circle,inner sep=1.5pt,fill=white] (y2) at (240:\R) {};
\node[draw,circle,inner sep=1.5pt,fill=white] (y3) at (360:\R) {};
\draw[line width=1pt,gray] (x1) to[out=240,in=60] (y2);
\draw[line width=1pt,gray] (x2) to[out=0,in=300] (y1);
\draw[line width=1pt,gray] (x3) to[out=120,in=180] (y3);
\node[draw,circle,inner sep=1pt,fill=white,white] (z) at (30:\R) {};
\end{scope}
\begin{scope}[shift={(4.5,-1)}]
\newdimen\R
\R=0.5cm
\draw (0,0) circle (\R);
\node[draw,circle,inner sep=1.5pt,fill=black] (x1) at (60:\R) {};
\node[draw,circle,inner sep=1.5pt,fill=black] (x2) at (180:\R) {};
\node[draw,circle,inner sep=1.5pt,fill=black] (x3) at (300:\R) {};
\node[draw,circle,inner sep=1.5pt,fill=white] (y1) at (120:\R) {};
\node[draw,circle,inner sep=1.5pt,fill=white] (y2) at (240:\R) {};
\node[draw,circle,inner sep=1.5pt,fill=white] (y3) at (360:\R) {};
\draw[line width=1pt,gray] (x1) to[out=240,in=300] (y1);
\draw[line width=1pt,gray] (x2) to[out=0,in=60] (y2);
\draw[line width=1pt,gray] (x3) to[out=120,in=180] (y3);
\node[draw,circle,inner sep=1pt,fill=white,white] (z) at (30:\R) {};
\end{scope}
\begin{scope}[shift={(6,-1)}]
\newdimen\R
\R=0.5cm
\draw (0,0) circle (\R);
\node[draw,circle,inner sep=1.5pt,fill=black] (x1) at (60:\R) {};
\node[draw,circle,inner sep=1.5pt,fill=black] (x2) at (180:\R) {};
\node[draw,circle,inner sep=1.5pt,fill=black] (x3) at (300:\R) {};
\node[draw,circle,inner sep=1.5pt,fill=white] (y1) at (120:\R) {};
\node[draw,circle,inner sep=1.5pt,fill=white] (y2) at (240:\R) {};
\node[draw,circle,inner sep=1.5pt,fill=white] (y3) at (360:\R) {};
\draw[line width=1pt,gray] (x1) to[out=240,in=180] (y3);
\draw[line width=1pt,gray] (x2) to[out=0,in=300] (y1);
\draw[line width=1pt,gray] (x3) to[out=120,in=60] (y2);
\node[draw,circle,inner sep=1pt,fill=white,white] (z) at (30:\R) {};
\end{scope}
\draw [-stealth] (0,-0.6) to [out=-90,in=180] (0.9,-1.5);
\draw [-stealth] (3,-0.6) to [out=-90,in=0] (2.1,-1.5);
\draw [-stealth] (1.5,-0.6) to (1.5,-0.9);
\begin{scope}[shift={(1.5,-1.5)}]
\newdimen\R
\R=0.5cm
\draw (0,0) circle (\R);
\node[draw,circle,inner sep=1.5pt,fill=black] (x1) at (60:\R) {};
\node[draw,circle,inner sep=1.5pt,fill=black] (x2) at (180:\R) {};
\node[draw,circle,inner sep=1.5pt,fill=black] (x3) at (300:\R) {};
\node[draw,circle,inner sep=1.5pt,fill=white] (y1) at (120:\R) {};
\node[draw,circle,inner sep=1.5pt,fill=white] (y2) at (240:\R) {};
\node[draw,circle,inner sep=1.5pt,fill=white] (y3) at (360:\R) {};
\draw[line width=1pt,gray] (x1) to (y2);
\draw[line width=1pt,gray] (x2) to (y3);
\draw[line width=1pt,gray] (x3) to (y1);
\node[draw,circle,inner sep=1pt,fill=white,white] (z) at (30:\R) {};
\end{scope}
\end{tikzpicture}
\end{center}
There are three possible optimal permutations, namely $\pi(i)=i$, $\pi(i)=i+1\Mod 3$ and $\pi(i)=i+2\Mod 3$.

Collecting our results, we have that, in the $N=3$ case, the optimal permutation $\pi$ is such that $\pi\in\mathcal C_3$.
\end{proof}

Proposition \ref{PropA3} allows us to state the following Theorem, that generalizes an analogous one proven by \textcite{Boniolo2012} in the particular case of the cost function $c(z)=z^p$ with $p>1$. An equivalent statement for general convex increasing functions can be found, for example, in Ref.~\cite{McCann1999}. 
\begin{pgoneth}[Optimal matching with convex increasing cost function]Given the assignment problem on $\Lambda$, if the cost function $c(z)\colon\Lambda\to\mathds R$ appearing in Eq.~\eqref{costoc} is a strictly convex increasing function, then the optimal permutation is the identity permutation.\label{TheoA5}
\end{pgoneth}
\begin{proof}
Let us assume by contradiction that the optimal matching $\pi$ is not the one corresponding to the identity permutation. Therefore, there exists at least a couple of matched pairs $(i,\pi(i))$, $(j,\pi(j))$ such that $i<j$ and $\pi(i)>\pi(j)$. But, by Proposition \ref{PropA2}, this implies that the cost can be decreased considering instead the matched pairs $(i,\pi(j))$ and $(j,\pi(i))$, hence the absurd and therefore the optimal matching is given by the identity permutation $\pi(i)=i$ for all values of $i=1,\dots,N$.
\end{proof} 

If the cost function is a $\mathcal{C}$-function, the following Theorem holds.
\begin{plzeroth}[Optimal matching with $\mathcal C$-function]Given the assignment problem on the interval $\Lambda$, if the cost function $c(z)\colon\Lambda\to\mathds R$ appearing in Eq.~\eqref{costoc} is a $\mathcal C$-function, then the optimal permutation $\pi$ is such that $\pi\in\mathcal C_N$.\label{TheoA6}
\end{plzeroth}
\begin{proof}
Given the optimal permutation $\pi$, because of Lemma~\ref{Lemma}, in every subset of three matched couples they must be co-oriented. Therefore, because of Proposition~\ref{PropA1}, $\pi\in\mathcal C_N$.
\end{proof}

\subsection{Average properties of the optimal solution}
In the following, we will apply the previous results to the \reap with cost function $c(z)=z^p$, $p\in\mathds R \setminus[0,1]$. We will study the average properties of the optimal solution assuming that the points are uniformly and independently distributed on the unit interval. To stress the dependency on $p$, we will denote the matching cost and the mean cost per edge corresponding to the permutation $\pi\in\mathcal S_N$ by
\begin{subequations}
\begin{align}
 \mC{p}(\pi)&\coloneqq\sum_{i=1}^N|x_i-y_{\pi(i)}|^p,\\ 
 \mc{p}(\pi)&\coloneqq\frac{1}{N}\mC{p}(\pi),
\end{align}\end{subequations}
respectively. The average optimal cost will be given by
\begin{equation}
 \mc{p}\coloneqq\overline{\min_{\pi\in\mathcal S_N}\mc{p}(\pi).}
\end{equation} 
For this particular case, Theorem \ref{TheoA6} allows us to state the following
\begin{corollary}Given the assignment problem on the interval $\Lambda$ with cost function $c(z)=z^p$, denoting by $\pi$ the optimal permutation, then $\pi(i)=i$ for $p>1$, and $\pi\in\mathcal C_N$ for $p<0$.\label{CorrII7}
\end{corollary}
\begin{proof}
It is enough to observe that, for $p>1$, the function $c(z)=z^p$ is a strictly increasing convex function, and apply therefore Theorem \ref{TheoA5}. On the other hand, for $p<0$, the function $c(z)$ is a $\mathcal C$-function, and we can apply Theorem \ref{TheoA6}.
\end{proof}
If periodic boundary conditions are assumed (i.e., the problem is considered on the unit circumference), we can derive a similar result. Indeed, the assignment problem with cost function $c(z)$ on the circumference can be restated as an assignment problem on $\Lambda$ with a modified cost function taking into account the periodicity. In particular, the assignment problem on the circumference with cost function $c(z)$ can be thought as an assignment problem on $\Lambda$ with cost function
\begin{equation}\label{cperiodic}
 \hat c(z)=c(z)\theta\left(\frac{1}{2}-z\right)\theta(z)+c(1-z)\theta(1-z)\theta\left(z-\frac{1}{2}\right),
\end{equation}
where $\theta(x)$ is the Heaviside step function. In this case, the following Corollary holds.
\begin{corollary2}\label{CorrII8}Let us consider the assignment problem on the circumference with cost function $c(z)=z^p$. Denoting by $\pi$ the optimal permutation, then $\pi\in\mathcal C_N$ for $p<0$ or $p>1$. In particular, for $p=1$ there exists a cyclic optimal solution. 
\end{corollary2}
\begin{proof}
As stated above, the assignment problem on the circumference with cost function $c(z)$ corresponds to the assignment problem on $\Lambda$ with cost function
\begin{equation}
 \hat c(z)=z^p\theta(z)\theta\left(\frac{1}{2}-z\right)+(1-z)^p\theta(1-z)\theta\left(z-\frac{1}{2}\right).
\end{equation}
The function above is a $\mathcal{C}$-function for $p<0$. Indeed, it is easily seen that $\Psi_\eta(z)=\hat c(\eta+z)-\hat c(z)$ is an increasing function on the interval $(0,1-\eta)$ for any value of $\eta\in(0,1)$, and therefore Eq.~\eqref{prima} holds. Moreover, the function $\Phi_\eta(z)=\hat c(\eta-z)-\hat c(z)$ is monotonically increasing on the interval $(0,\eta)$, and therefore Eq.~\eqref{seconda} is satisfied. The proof for the $p\geq 1$ case has been given in Ref.~\cite{Boniolo2012}.
\end{proof}

\subsubsection{Donsker's theorem and the Brownian bridge process} In Corollary~\ref{CorrII7} and Corollary~\ref{CorrII8} we have proved that, for $p\in\mathds R\setminus[0,1]$, the optimal permutation has the form $\pi(i)=i+k\Mod N$, for some $k\in [N]$ depending on the instance of our problem. In the case of open boundary conditions (i.e., of the problem on the interval), the optimal cost can be written as
\begin{equation}\label{sommapneg}
\min_{\pi\in\mathcal S_N}\mc{p}(\pi) = \min_{k\in [N]} \frac{1}{N}\sum_{i=1}^N\left|x_i - y_{i+k\Mod{N}}\right|^p,
\end{equation}
In particular, for $p>1$, the optimal permutation of the assignment problem on $\Lambda$ is always $\pi(i)=i$, independently from both the instance and the specific values of $p$, and therefore the optimal cost is simply given by
\begin{equation}\label{sommapmag1}
\min_{\pi\in\mathcal S_N}\mc{p}(\pi) = \frac{1}{N}\sum_{i=1}^N\left|x_i - y_{i}\right|^p.
\end{equation}
These results imply that the optimal solution is related, in the $N\to\infty$ limit, to a linear combination of two Brownian bridge processes, a fact that follows from Donsker's theorem \cite{Donsker1952,*Dudley1999}. 
\begin{donskerth}[Donsker] For any $N\in\mathds N$, there exists a probability space $\Omega_N$ such that we can define on it the random variable $\boldsymbol{\mathsf X}_N\coloneqq(\mathsf X_i)_i$, $\boldsymbol{\mathsf X}_N\colon\Omega_N\to \Lambda^N$, each component $\mathsf X_i$ being a random variable uniformly distributed on the unit interval $\Lambda$. Moreover, let us consider the corresponding $N$th empirical process,
\begin{equation}
 \mathsf F_N(t,\boldsymbol{\mathsf X}_N)\coloneqq\frac{1}{N}\sum_{i=1}^N\theta(t-\mathsf X_i)-t.
\end{equation}
Then we can find a sample-continuous Brownian bridge process on $\Lambda$, defined on the same probability space $\Omega_N$, $\mathsf B_N\colon \Lambda\times\Omega_N\to \mathsf B_N(t;\omega)$, such that, for all $\epsilon>0$,
\begin{equation}
 \lim_N\Pr\left[\sup_{t\in[0,1]}\left|\sqrt{N}\mathsf F_N(t,\boldsymbol{\mathsf X}_N(\omega))-\mathsf B_N(t;\omega)\right|>\epsilon\right]=0.
\end{equation}
\end{donskerth}
Donsker's theorem expresses the (weak) convergence of the process $\mathsf F_N$ to a Brownian bridge process in the $N\to\infty$ limit. The convergence rate has been studied by \textcite{Komlos1975}, that proved that
\begin{equation}\label{KMT}
 \sup_{0\leq t\leq 1}\left|\sqrt{N}\mathsf F_N(t,\boldsymbol{\mathsf X}_N(\omega))-\mathsf B_N(t;\omega)\right|=O\left(\frac{\ln N}{\sqrt{N}}\right)
\end{equation}
almost surely \cite{csorgo1981}.

Let $\Xi_N=\{x_i\}_i\equiv \boldsymbol{\mathsf X}_N(\omega)$ be a realization of $\boldsymbol{\mathsf X}_N$ for a given instance of the problem $\omega\in\Omega_N$. The empirical process $\mathsf F_N$ is given by
\begin{equation}
\mathsf F_N(t,\boldsymbol{\mathsf X}_N(\omega))=\frac{1}{N}\sum_{i=1}^N\theta(t-x_i)-t
\end{equation}
Supposing now that the elements of $\Xi_N$ are labeled in such a way that $x_i<x_j\Leftrightarrow i<j$, we have \begin{equation}
\lim_{t\to x_i^+}\mathsf F_N(t,\boldsymbol{\mathsf X}_N(\omega))=\frac{i}{N}-x_i.\end{equation}
Given therefore two realizations $\omega$ and $\tilde\omega$, corresponding to $\Xi_N=\{x_i\}_i$ and $\Upsilon_N=\{y_i\}_i\equiv \boldsymbol{\mathsf X}_N(\tilde\omega)$ respectively, both generated as above, we can write
\begin{equation}
 y_{j}-x_i=\frac{j-i}{N}+\lim_{t\to x_i^+}\mathsf F_N(t,\boldsymbol{\mathsf X}_N(\omega))-\lim_{t\to y_{j}^+}\mathsf F_N(t,\boldsymbol{\mathsf X}_N(\tilde \omega)).
\end{equation}
Denoting by $i=Nu+\sfrac{1}{2}$ and $j=Nv+\sfrac{1}{2}$, $u,v\in\Lambda$, and observing that, for large $N$ and $\epsilon>0$,
\begin{equation}
 \Pr\left[\left|x_i-u\right|>\epsilon\right]\simeq \exp\left(-\frac{N\epsilon^2}{2u(1-u)}\right)
\end{equation}
Donsker's theorem allows us to write
\begin{multline}\label{limitDonsker}
 \sqrt{N}\left(y_{Nv+\sfrac{1}{2}}-x_{Nu+\sfrac{1}{2}}+u-v\right)\\
 \xrightarrow{N\to +\infty}\mathsf B(u;\omega)-\mathsf B(v;\tilde \omega),
\end{multline}
where the limit is intended in probability. This results implies that we can write the arguments of the sums in Eq.~\eqref{sommapneg} and Eq.~\eqref{sommapmag1} in terms of Brownian bridge processes in the large $N$ limit.
\subsubsection{Open boundary conditions}
To be more specific, we start analyzing the \reap on the interval, assuming a cost function $c(z)=z^p$ with $p\in\mathds R\setminus[0,1]$.
\paragraph{The $p>1$ case.} Many aspects of the \reap in one dimension for $p>1$ have been analyzed in Refs.~\cite{McCann1999,Boniolo2012,Caracciolo2014c}. As observed above, the optimal permutation in this case is always the identity one, $\pi(i)=i$. Denoting by
\begin{equation}
\varphi_i\coloneqq y_i-x_i,
\end{equation}
we can write, for any instance and any value of $N$,
\begin{equation}
\min_{\pi\in\mathcal S_N}\mc{p}(\pi)=\frac{1}{N}\sum_{i=1}^N|\varphi_i|^p.
\end{equation}
From Donsker's theorem, we know that in the $N\to+\infty$ limit,
\begin{equation}
 \phi(s)\coloneqq\sqrt{N}\varphi_{Ns+\sfrac{1}{2}}\xrightarrow{N\to +\infty}\mathsf B(s;\omega)-\mathsf B(s;\tilde \omega),
\end{equation}
where we have introduced the new variable $s$ such that $i=Ns+\sfrac{1}{2}$. The authors of Refs.~\cite{Boniolo2012,Caracciolo2014c} used this correspondence between the Brownian bridge process and the optimal solution of the Euclidean assignment problem to derive the expression of the average optimal cost and the correlation function in the $N\to+\infty$ limit. They obtained
\begin{equation}\label{costoBoniolo}
N^{\frac{p}{2}}\mc{p}=\frac{\Gamma\left(1+\frac{p}{2}\right)}{p+1}+O\left(\frac{1}{N}\right),
\end{equation} 
whereas the correlation function has been studied in Refs.~\cite{Boniolo2012}. Here we derive
\begin{multline}\label{corrpmag1}
c_p(r)\coloneqq \frac{1}{\mathcal N(r)}\iint_0^1\overline{\phi(s)\phi(t)}\delta\left(|s-t|-r\right)\dd s\dd t\\
=\iint_0^1\frac{\overline{\phi(s)\phi(t)}-\overline{\phi(s)}\,\overline{\phi(t)}}{\mathcal N(r)}\delta\left(|s-t|-r\right)\dd s\dd t\\=\frac{1}{3}\left(1-r\right)^2,\quad r\in\Lambda,
\end{multline}
where
\begin{equation}
 \mathcal N(r)\coloneqq\iint_0^1\delta\left(|s-t|-r\right)\dd s\dd t=2(1-r).
\end{equation}
The correlation function has been obtained using the following fundamental property of the Brownian bridge process,
\begin{equation}
 \overline{\mathsf B(s;\omega)\mathsf B(t;\omega)}=\min\{s,t\}-st.\label{covbb}
\end{equation}

In the following, we derive the finite size corrections to the asymptotic cost in Eq.~\eqref{costoBoniolo}, through a straightforward computation on the optimal matching solution and following the approach of \textcite{Boniolo2012}. Let us first observe that the probability of finding $x_k$ in the interval $(x,x+\dd x)$ is
\begin{equation}\label{prbin}
\Pr\left[x_k\in\dd x\right]=B(k;N,x)\frac{k\dd x}{x},\quad k=1,\dots,N,
\end{equation}
where we have introduced the binomial distribution
\begin{equation}
B(k;N,p)\coloneqq \binom{N}{k}p^k(1-p)^{N-k}.
\end{equation}
In Eq.~\eqref{prbin} we have used the notation $x_k\in\dd x\Leftrightarrow x_k\in(x,x+\dd x)$. Due to the fact that the random variables $x_k$ and $y_k$ are independent, we can write
\begin{multline}\label{distphipmag}
 \Pr[\varphi_k\in\dd\varphi]=\dd\varphi\binom{N}{k}^2 k^2\times\\
 \times\iint_0^1\! \delta(\varphi-y+x)(xy)^{k-1}[(1-x)(1-y)]^{N-k}\dd x\dd y.
\end{multline}
Eq.~\eqref{distphipmag} allows in principle the calculation of the average optimal cost for any value $N$. For example, in the $p=2$ case, we obtain
\begin{equation}
\min_{\pi}\mc{2}(\pi)
=\frac{1}{3(1+N)}
=\frac{1}{3N}-\frac{1}{3N^2}+o\left(\frac{1}{N^2}\right).
\end{equation}
The calculations, however, greatly simplify in the $N\to+\infty$ limit. Introducing the variable $\phi(s)\coloneqq \sqrt{N}\varphi_{Ns+\sfrac{1}{2}}$, Eq.~\eqref{distphipmag} can be written, up to higher order terms, as
\begin{multline}
\Pr[\phi(s)\in\dd\phi]=\\
 =\dd\phi\frac{\e^{-\frac{\phi^2}{4s(1-s)}}}{2\sqrt{\pi s(1-s)}}
 \left\{1+\frac{s(1-s)+1}{8Ns(1-s)}+\frac{7s(1-s)-2}{8N s^2(1-s)^2}\phi^2\right.\\
 \left.+\frac{1-5s(1-s)}{32Ns^3(1-s)^3}\phi^4+o\left(\frac{1}{N}\right)\right\},\label{distphi}
\end{multline}
see Appendix~\ref{app:distphi}. As expected, the leading term is the distribution of a Brownian bridge process on the domain $\Lambda$ \footnote{This result has been obtained, for example, in Ref.~\cite{Boniolo2012}. Observe, however, that in Ref.~\cite{Boniolo2012} one set of points is supposed to be fixed and, for this reason, the variance of the distribution is one half of the variance of the leading term in Eq.~\eqref{distphi}.}. From Eq.~\eqref{distphi} we can easily obtain
\begin{multline}\label{costopmag1}
N^{\frac{p}{2}}\mc{p}= \int_0^1\overline{|\phi(s)|^p}\dd s\\
=\frac{\Gamma\left(1+\sfrac{p}{2}\right)}{p+1}\left(1-\frac{1}{N}\frac{p(p+2)}{8}\right)+o\left(\frac{1}{N}\right).
\end{multline} 
The expression of the leading term in Eq.~\eqref{costopmag1} has been numerically verified, for example, in Refs.~\cite{Boniolo2012,Caracciolo2014c}. In Fig.~\ref{fig:pmag1} we compare the results of our simulations with the theoretical prediction given in Eq.~\eqref{costopmag1}.

\paragraph{The $p<0$ case.} Let us now consider the $p<0$ and let us define
\begin{multline}\label{phiN}
 \frac{1}{\sqrt N}\phi_t^{(N)}(s)\coloneqq \\
 =y_{N[s+t\Mod{1}]+\sfrac{1}{2}}-x_{Ns+\sfrac{1}{2}}-\sigma(s,t),
\end{multline}
where $Ns+\sfrac{1}{2}=k\in[N]$ and $Nt\in[N]$ and
\begin{equation}
 \sigma(s,t)\coloneqq [s+t\Mod{1}]-s,\quad s,t\in\Lambda.
\end{equation}
Corollary \ref{CorrII7} states that, for a given instance of our problem, the optimal solution corresponds to a certain value $t$ such that
\begin{multline}
 \min_\pi\varepsilon^{(p)}_N(\pi)=\\
 =\frac{1}{N}\sum_{k=1}^N\left|\sigma\left(\frac{k-\sfrac{1}{2}}{N},t\right)+\frac{1}{\sqrt{N}}\phi_t^{(N)}\left(\frac{k-\sfrac{1}{2}}{N}\right)\right|^p.
\end{multline}
From Donsker's theorem, we have
\begin{equation}\label{bblimit}
 \phi_t^{(N)}(s)\xrightarrow{N\to+\infty}\phi_t(s)\coloneqq \mathsf B(s;\omega)-\mathsf B\left(s+t\Mod 1;\tilde \omega\right),
\end{equation}
and the optimal cost can be written, in the large $N$ limit, as
\begin{multline}\label{costshift}
 \frac{1}{N}\sum_{k=1}^N\left|\sigma\left(\frac{k-\sfrac{1}{2}}{N},t\right)+\frac{1}{\sqrt{N}}\phi_t^{(N)}\left(\frac{k-\sfrac{1}{2}}{N}\right)\right|^p\xrightarrow{N\gg 1}\\
 \int_0^1\left|\sigma(s,t)+\frac{\phi_t^{(N)}(s)}{\sqrt{N}}\right|^p\dd s,
\end{multline}
for some value of $t$ depending both on $p$ and on the specific instance of the problem. The value of $t$ for the optimal solution can be found by minimizing the expression above respect to $t$. We proceed perturbatively, observing that
\begin{equation}
 \lim_{N}\int_0^1\left|\sigma(s,t)+\frac{\phi_t^{(N)}(s)}{\sqrt{N}}\right|^p\dd s=t^p(1-t)+t(1-t)^p,
\end{equation}
which is minimized by $t=\sfrac{1}{2}$. To evaluate the nontrivial finite-size corrections, we assume therefore
\begin{equation}\label{tau0}
 t\coloneqq\frac{1}{2}+\frac{\tau}{\sqrt N},
\end{equation}
where $\tau$ depends both on $p$ and on the instance of the problem. Performing a large $N$ expansion, the cost can be written as
\begin{multline}\int_0^1\left|\sigma(s,t)+\frac{\phi_t^{(N)}(s)}{\sqrt{N}}\right|^p\,d s\\
=\frac{1}{2^p}+p\int_0^{1}\frac{\mathrm{sign}\left(\sfrac{1}{2}-s\right)}{2^{p-1}\sqrt{N}}\phi^{(N)}_{\sfrac{1}{2}}(s)\dd s\\-\frac{p\tau\left(\tau+\phi_{\sfrac{1}{2}}\left(\sfrac{1}{2}\right)\right)}{2^{p-2}N}+\frac{p(p-1)}{2^{p-1}N}\int_0^{1}\left(\tau+\phi_{\sfrac{1}{2}}(s)\right)^2\dd s\\
+\tau p\int_0^{1}\frac{\mathrm{sign}\left(\sfrac{1}{2}-s\right)}{2^{p-1}N}\left.\partial_t\phi_t(s)\right|_{t=\frac{1}{2}}\dd s+o\left(\frac{1}{N}\right).
\end{multline}
Here and in the following we adopt the convention $\mathrm{sign}(0)=1$. Observe that, due to Eq.~\eqref{KMT}, we can neglect the corrections to the asymptotic limit given in Eq.~\eqref{bblimit} in all terms appearing in the expansion above, except in the second one, that must be treated differently when the average will be performed, due to the different scaling of the coefficient. Being $\phi_t(s)=\mathsf B(s;\omega)-\mathsf B\left(s+t\Mod 1;\hat\omega\right)$, we can formally write in the last term $\partial_t\phi_t(s)=-\partial_t\mathsf B\left(s+t\Mod 1,\hat\omega\right)=-\partial_s\mathsf B\left(s+t\Mod 1,\hat\omega\right)$ and therefore, after an integration by parts, the expression above becomes
\begin{multline}\label{costopmin1-tau}
 \int_0^1\left|\sigma(s,t)+\frac{\phi^{(N)}_t(s)}{\sqrt{N}}\right|^p\dd s=\\
 =\frac{1}{2^p}-\frac{p\tau\left(\tau+\phi_{\sfrac{1}{2}}\left(\sfrac{1}{2}\right)+\phi_{\sfrac{1}{2}}(0)\right)}{2^{p-2}N}\\
 +\frac{p(p-1)}{2^{p-1}N}\int_0^{1}\left(\tau+\phi_{\sfrac{1}{2}}(s)\right)^2\dd s\\
 +\frac{p}{2^{p-1}\sqrt{N}}\int_0^{1}\mathrm{sign}\left(\sfrac{1}{2}-s\right)\phi_{\sfrac{1}{2}}^{(N)}(s)\dd s+o\left(\frac{1}{N}\right).
\end{multline}
Minimizing respect to $\tau$, we obtain, up to higher terms,
\begin{equation}\label{tau}
 \tau=-\frac{p-1}{p-3}\int_0^1\phi_{\sfrac{1}{2}}(s)\dd s+\frac{\phi_{\sfrac{1}{2}}\left(\sfrac{1}{2}\right)+\phi_{\sfrac{1}{2}}(0)}{p-3}.\end{equation}
We have verified our assumption in Eq.~\eqref{tau0} and therefore Eq.~\eqref{tau}. In particular, Eq.~\eqref{tau} implies $\overline{\tau}=0$ and \footnote{It is amusing to remark that if we expand $ 6\overline{\tau^2}$ around $p=-\infty$ we get a series with integer coefficients
\[
 6 \overline{\tau^2} = \sum_{n\ge 0} \frac{a_n}{p^n},\quad a_n=  
\begin{cases} 
1 & \hbox{for } n=0,\\
3^{n-2}(n+2) & \hbox{for } n>0. 
\end{cases}
\]
The sequence $(a_n)_{n\in\mathds N}$ can have a combinatorial interpretation. It is indeed the number of spanning trees in a $(n-2)$-book graph \cite{oeis,*Doslic2013}.}
\begin{equation}\label{tau2th}
\overline{\tau^2}=\frac{p^2-5p+7}{6(p-3)^2}.
\end{equation}
The results of our numerical simulations, given in the inset in Fig.~\ref{fig:pmin0}, show a good agreement between the prediction in Eq.~\eqref{tau2th} and simulations.

To obtain the average optimal cost, we have to substitute Eq.~\eqref{tau} into Eq.~\eqref{costopmin1-tau}, and then average over the possible realizations, using the fact that, as consequence of Eq.~\eqref{covbb}, the following property holds:
\begin{multline}
 \overline{\phi_{\sfrac{1}{2}}(s)\phi_{\sfrac{1}{2}}(t)}=\min\left\{s,t\right\}-st\\
 +\min\left\{s+\frac{\mathrm{sign}\left(\sfrac{1}{2}-s\right)}{2},t+\frac{\mathrm{sign}\left(\sfrac{1}{2}-t\right)}{2}\right\}\\-\left(s+\frac{\mathrm{sign}\left(\sfrac{1}{2}-s\right)}{2}\right)\left(t+\frac{\mathrm{sign}\left(\sfrac{1}{2}-t\right)}{2}\right).\label{covbb2}
\end{multline}
The average of the last term in Eq.~\eqref{costopmin1-tau} requires the 
evaluation of $\overline{\phi_{\sfrac{1}{2}}^{(N)}(s)}$ that must be performed, 
as anticipated, including the corrections to the limiting Brownian bridge 
distribution. Introducing for the sake of brevity 
$\varsigma=s+\sigma(s,\sfrac{1}{2})$, we get
\begin{multline}\label{distphipminz1}
\Pr[\phi^{(N)}_{\sfrac{1}{2}}(s)\in\dd\phi]
=\frac{\exp\left(-\frac{\phi^2}{2s(1-s)+2\varsigma(1-\varsigma)}\right)\dd\phi}{\sqrt{2\pi} \sqrt{s(1-s)+\varsigma(1-\varsigma)}}\\
\times \left[1+\frac{s-\varsigma}{\sqrt{N}}\left(\frac{(s+\varsigma-1)^2}{(s(1-s)+\varsigma(1-\varsigma))^2} \phi\right.\right.\\
\left.\left.-\frac{1-(s-\varsigma)^2-3s(1-s)-3\varsigma(1-\varsigma)}{3(s(1-s)+\varsigma(1-\varsigma))^3}\phi^3\right)+o\left(\frac{1}{\sqrt N}\right)\right],
\end{multline}
which provides
\begin{equation}
 \overline{\phi^{(N)}_{\sfrac{1}{2}}(s)}=-\frac{\sigma(s,\sfrac{1}{2})}{\sqrt{N}}.
 \end{equation}
Collecting the results above, we finally obtain
\begin{equation}\label{cpminzero2}
\mc{p}
=\frac{1}{2^p}\left[1+\frac{1}{3N}\frac{p(p-2)(p-4)}{p-3}\right]
+o\left(\frac{1}{N}\right).
\end{equation}
We verified the previous formula numerically. The numerical results show a good agreement with the theoretical prediction in Eq.~\eqref{cpminzero2}, see Fig.~\ref{fig:pmin0}. 

Given the optimal permutation $\pi$, such that $\pi(i)=i+k\Mod{N}$, the correlation function for the matching field
\begin{equation}
 \mu_i=y_{\pi(i)}-x_{i}\xrightarrow[i=Ns+\sfrac{1}{2}]{N\to+\infty}\mu(s)=\begin{cases}\frac{1}{2}&\text{if }0\leq s<\frac{1}{2}\\ -\frac{1}{2}&\text{if }\frac{1}{2}<s\leq1,\end{cases}
\end{equation}
is easily calculated as
\begin{multline}\label{corrfunpmin0intera}
c_p(r)\coloneqq\frac{1}{\mathcal N(r)}\iint_0^1\overline{\mu(s)\mu(t)}\delta\left(|s-t|-r\right)\dd s\dd t\\
=\begin{cases}\frac{3}{4}-\frac{1}{2(1-r)}&\text{if }0\leq r<\frac{1}{2}\\
  -\frac{1}{4}&\text{if }\frac{1}{2}<r\leq1.
 \end{cases}
\end{multline}
More interestingly, we introduce the correlation function for the field
\begin{equation}
 \hat \mu_{i}\coloneqq\sqrt{N}\left[y_{\pi(i)}-x_{i}-\frac{\mathrm{sign}\left(N-k-i\right)}{2}\right].
\end{equation}
For $i=Ns+\sfrac{1}{2}$, in the $N\to\infty$ limit keeping $s$ fixed, we have
\begin{equation}
 \hat\mu_{Ns+\sfrac{1}{2}}\xrightarrow{N\to+\infty}\phi_{\sfrac{1}{2}}(s)+\tau\eqqcolon\hat\mu(s).
\end{equation}
For $0\leq r\leq 1$,
\begin{multline}\label{corrfunpmin0}
 \hat c_p(r)\coloneqq\iint_0^1\frac{\overline{\hat\mu(s)\hat\mu(t)}}{\mathcal N(r)}\delta\left(|s-t|-r\right)\dd s\dd t\\
 =\iint_0^1\frac{\overline{\left(\tau+\phi_{\sfrac{1}{2}}(s)\right)\left(\tau+\phi_{\sfrac{1}{2}}(t)\right)}}{\mathcal N(r)}\delta\left(|s-t|-r\right)\dd s\dd t\\
 =\begin{cases}
   \frac{p^2-6p+10}{6(p-3)^2}+\frac{10r^2-30r+17}{6(p-3)(1-r)}r-\frac{pr(1-r)}{p-3}&\text{if }0\leq r\leq\sfrac{1}{2},\\
   \frac{p^2-3p+1}{6(p-3)^2}-\frac{3p-4}{3(p-3)}r+\frac{3p-5}{3(p-3)}r^2&\text{if }\sfrac{1}{2}<r\leq 1.
  \end{cases}
\end{multline}
Observe that
\begin{equation}
 \lim_{p\to-\infty}\hat c_{p}(r)=\frac{1}{6}-r(1-r),
\end{equation}
that coincides with the correlation function for the assignment problem with $p=2$ on the circumference \cite{Caracciolo2014c,Caracciolo2015}. This fact is not a coincidence. Indeed, for $p\to -\infty$, we have that (see below)
\begin{equation}
 \lim_{p\to-\infty}\tau=-\int_0^1\phi_{\sfrac{1}{2}}(s)\dd s.
\end{equation}
By comparison with Eq.~\eqref{taup2} below, it will be clear that $\hat\mu(s)$ 
for $p\to-\infty$ coincides with the solution of the assignment problem for 
$p=2$ on the circumference in which one set of points is translated by 
$\sfrac{1}{2}$. In Fig.~\ref{fig:corrpmin0obc} we compare the predictions above 
for $c_p$ and $\hat c_p$ with our numerical results.

\begin{figure}
\includegraphics[width=0.45\textwidth]{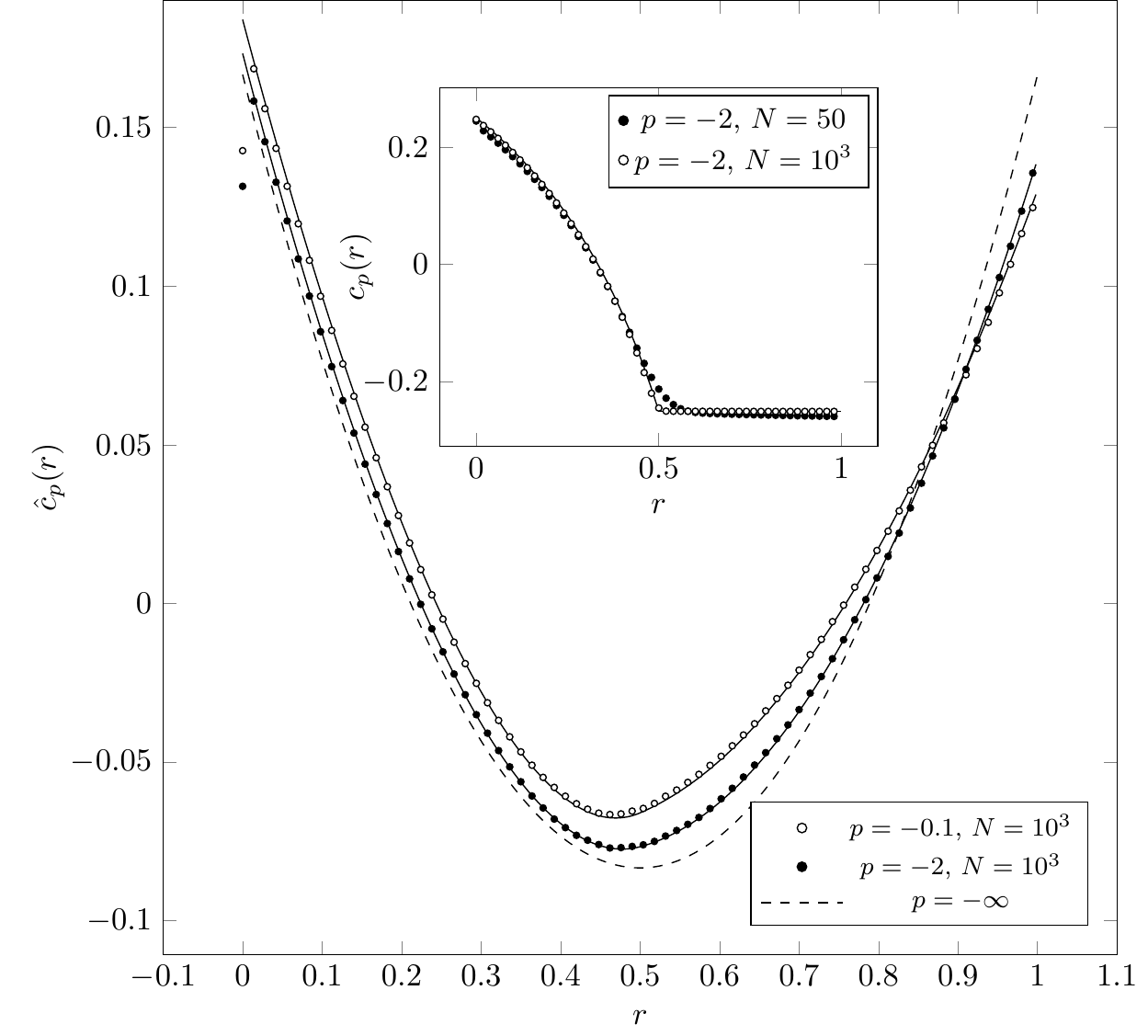}
\caption{Correlation functions $c_p(r)$ (inset) and $\hat c_p(r)$ for the assignment problem on the interval $\Lambda$ with $p<0$, as defined in Eq.~\eqref{corrfunpmin0intera} and Eq.~\eqref{corrfunpmin0}, respectively. Observe that a finite-size effect appears for $r\to 0$ in the numerical data for $\hat c_p(r)$. In all cases, the theoretical predictions are represented in solid line.\label{fig:corrpmin0obc}}
\end{figure}

   \begin{figure*}[!ht]
     \subfloat[Numerical results for the average optimal cost (inset) and its finite-size corrections in the assignment problem with $p>1$ on $\Lambda$. We fitted the obtained numerical results using the fitting function $f(N)=a+\sfrac{b}{N}$ for each value of $p$, $a$ and $b$ being fitting parameters corresponding to the leading cost and to the finite-size corrections, respectively.  We present our numerical results and we compare them with our prediction given in Eq.~\eqref{costopmag1}.\label{fig:pmag1}]{ \includegraphics[width=0.45\textwidth]{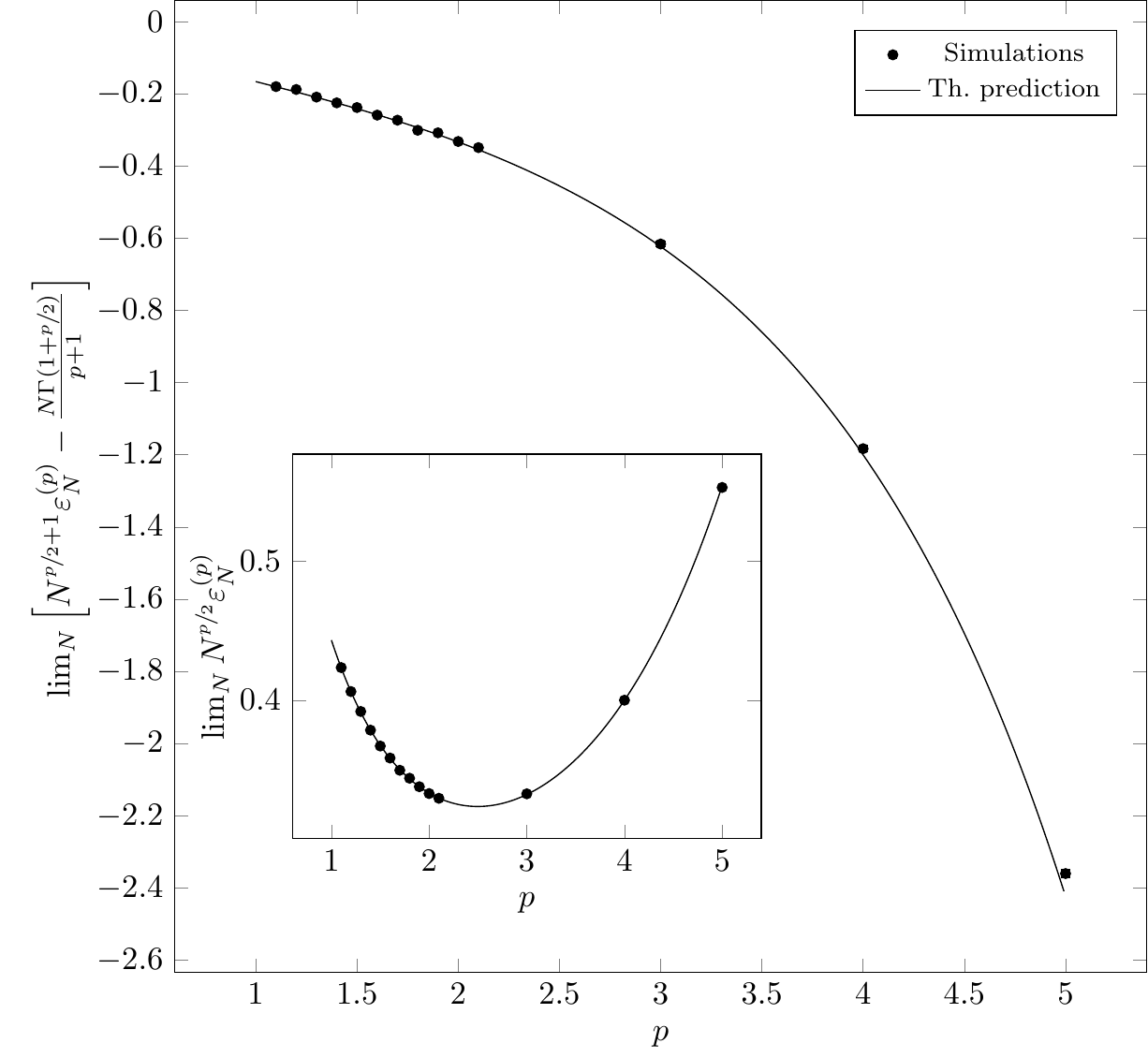}}
     \hfill
     \subfloat[Numerical results for the average optimal cost in the assignment problem on $\Lambda$ with $p<0$. 
     The theoretical predictions are given in Eq.~\eqref{cpminzero2} and they are represented in solid lines. In the upper inset, numerical results for $\overline{\tau^2}$ for different values of $p$, obtained using two-parameters fitting function in the form $f(N)=\alpha+\sfrac{\beta}{N}$, $\alpha$ being the numerical estimation for $\overline{\tau^2}$. We compare them with the prediction in Eq.~\eqref{tau2th}.\label{fig:pmin0}]{\includegraphics[width=0.45\textwidth]{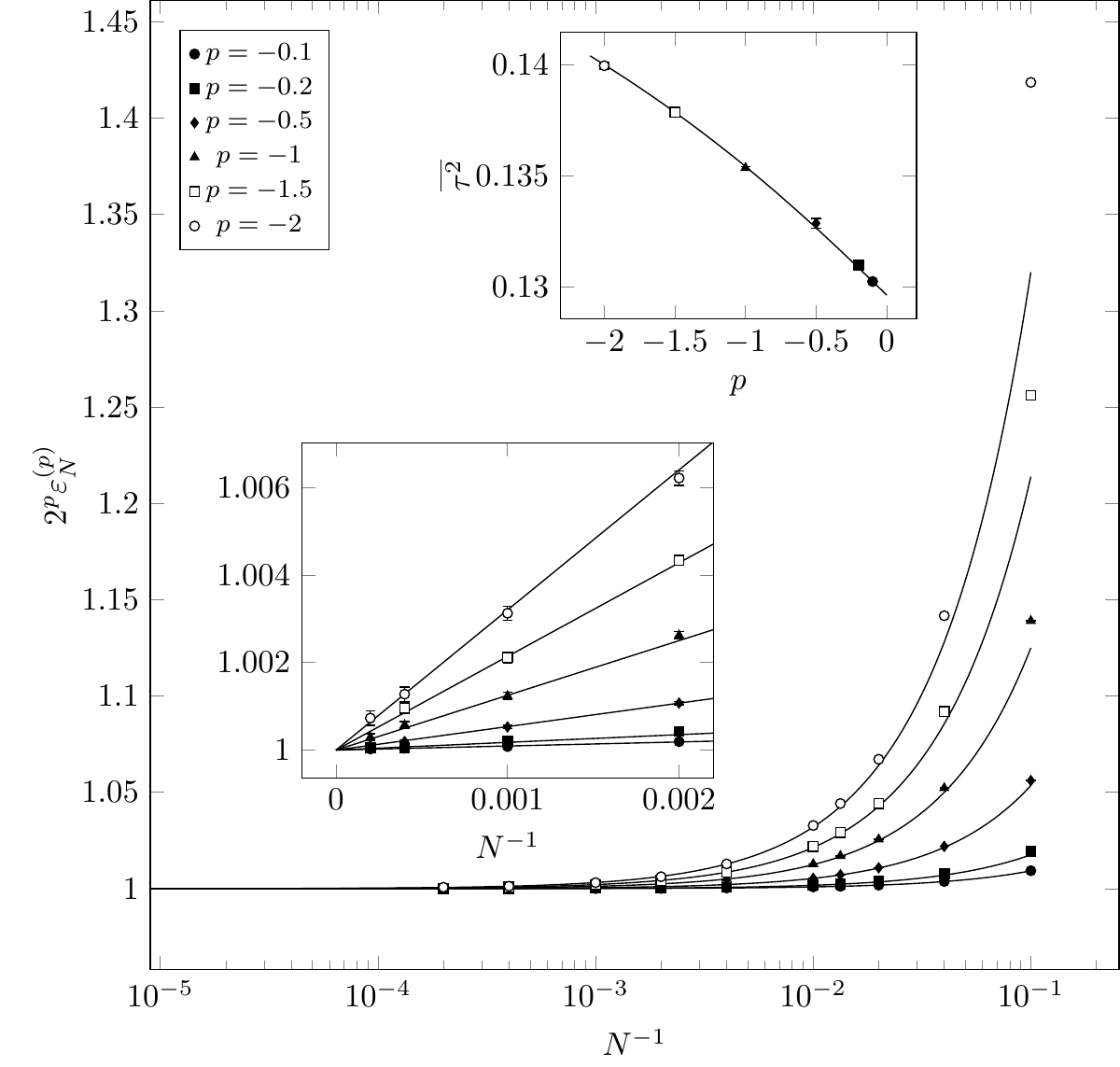}}\\
     \subfloat[Numerical results for the average optimal cost (inset) and its finite size corrections in the assignment problem with $p>1$ on the circumference. We fitted the numerical results obtained for different values of $N$ at fixed $p$ using the fitting function in the form given in Eq.~\eqref{scalingpbc}. See also Table~\ref{tab:datipbc}. 
     \label{fig:pmag1pbc}]{%
       \includegraphics[width=0.45\textwidth]{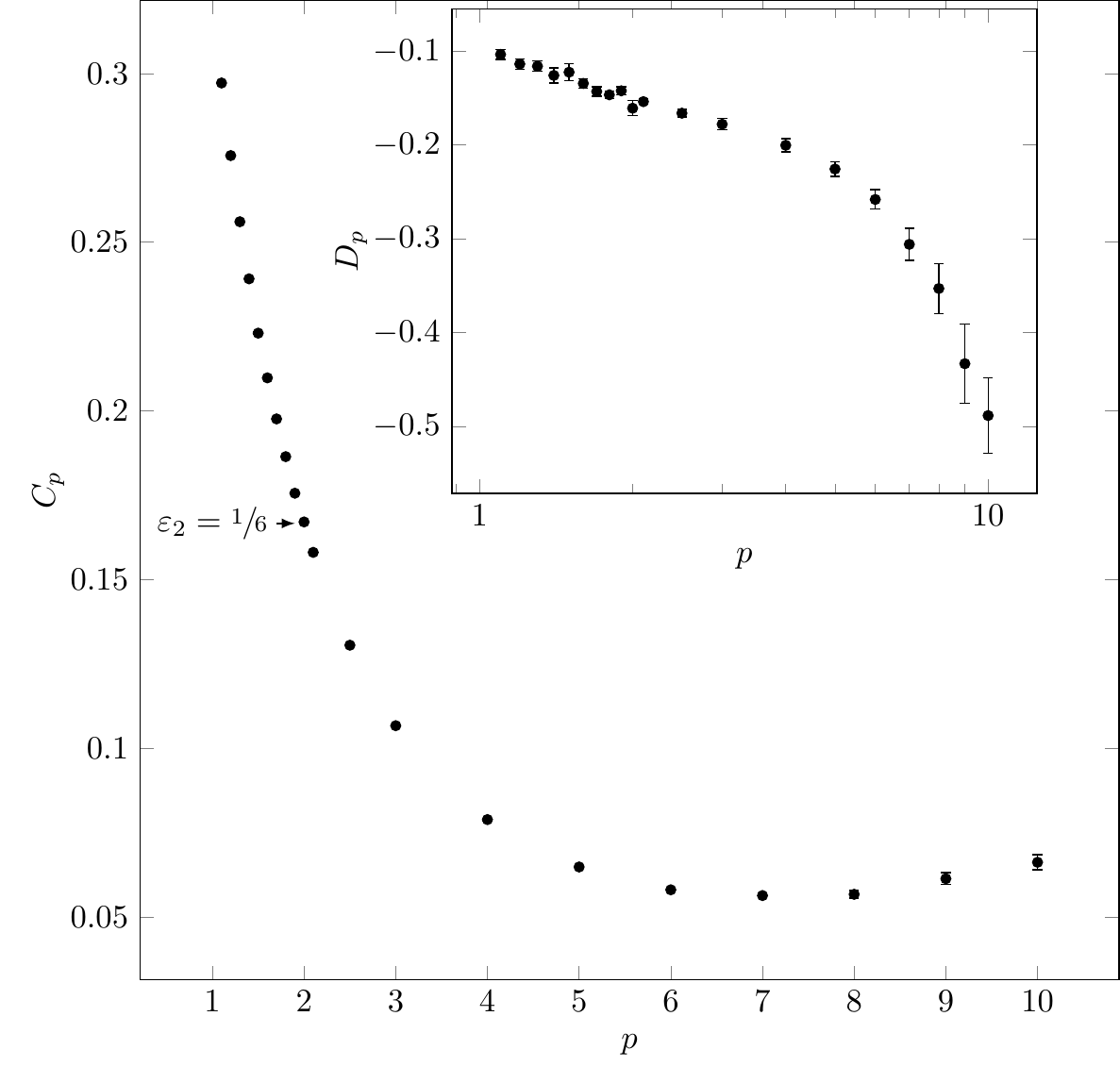}
     }\hfill
     \subfloat[Numerical results for the average optimal cost and its finite size corrections in the assignment problem with $p<0$ on the circumference. We present our numerical results and we compare them with our theoretical predictions for the asymptotic behavior given in Eq.~\eqref{costopmin0PBC} (solid lines).\label{fig:pmin0PBC}]{%
       \includegraphics[width=0.45\textwidth]{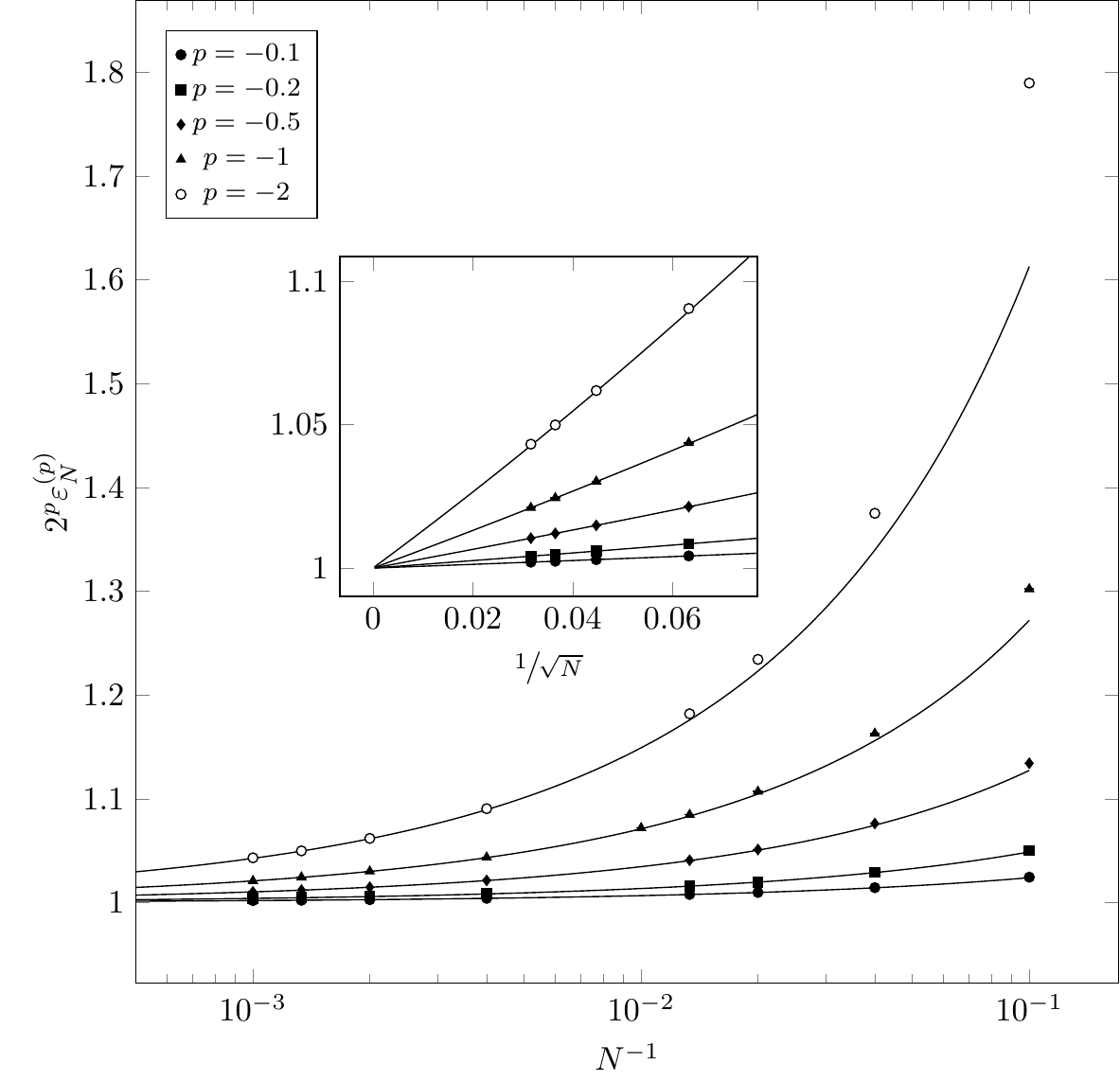}
     }
     \caption{Numerical results for the assignment problem. Error bars are typically smaller than the markers in the figures.}
     \label{fig:assignment}
   \end{figure*}
\subsubsection{Periodic boundary conditions} 
Corollary~\ref{CorrII8} states that, in the case of periodic boundary conditions, both for $p>1$ and for $p<0$, the optimal matching can be found searching for the optimal permutation in the set $\mathcal C_N$. The calculation above for the assignment problem on the interval, however, has to be slightly modified, due to the fact that the cost function is replaced by the one in Eq.~\eqref{cperiodic}, whereas Eq.~\eqref{limitDonsker} still holds. In particular, the average optimal cost has the form given in Eq.~\eqref{costshift}, with $\sigma(s,t)\equiv\sigma(t)=t\in[0,1)$, value of the global shift depending on the specific instance and on the value of $p$. The optimal cost can be written, in the large $N$ limit, as
\begin{equation}
 \int_0^1\left[\left|t+\frac{\phi_t(s)}{\sqrt N}\right|\Mod{\sfrac{1}{2}}\right]^p\dd s,
\end{equation}
for a certain value of $t$ obtained by minimization.

\paragraph{The $p>1$ case.} The $p>1$ case has been analyzed in 
Refs.~\cite{Caracciolo2014c,Boniolo2012,Caracciolo2015}. In this case at the 
leading order we obtain $t=0$ and therefore $t= o(1)$. An explicit 
expression of $t$ is known for $p=2$ and $p\to+\infty$ only. In general, 
we can assume that $t=\sfrac{\tau}{\sqrt N}+o\left(\sfrac{1}{\sqrt N}\right)$. 
The value of $\tau$ is obtained minimizing, in the $N\to+\infty$ limit, the 
expression
\begin{equation}
 \int_0^1\left|\tau+\phi_0(s)\right|^p\dd s.
\end{equation}
For $p=2$ we have \cite{Boniolo2012,Caracciolo2014c}
\begin{subequations}\label{taup2}
\begin{align}
  \mc{2}&=\frac{1}{6N}+o\left(\frac{1}{N}\right),\\
  \tau&=-\int_0^1\phi_0(s)\dd s.
\end{align}
\end{subequations}
Unfortunately, no general expression for $\tau$ is available to our knowledge. Numerical simulations suggests that the average optimal cost scales as
\begin{equation}
 \mc{p}=\frac{1}{N^\frac{p}{2}}\left[C_p+\frac{D_p}{N}+o\left(\frac{1}{N}\right)\right].\label{scalingpbc}
\end{equation}
In Fig.~\ref{fig:pmag1pbc} and in Table~\ref{tab:datipbc}  we present our numerical results for the average optimal cost and its finite-size corrections for the assignment problem on the circumference in the $p>1$ case. The data have been obtained using Eq.~\eqref{scalingpbc} to extrapolate the $N\to+\infty$ limit for both $C_p$ and $D_p$. In Ref.~\cite{Caracciolo2014} the $p=2$ case was carefully analyzed using a particular scaling ansatz, and the scaling in Eq.~\eqref{scalingpbc} was numerically verified. In particular, they obtained $C_2=0.166668(3)$ and $D_2=-0.1645(13)$. We refer to Ref.~\cite{Caracciolo2014c} for further discussion on the correlation function on the circumference.

\begingroup\squeezetable
\begin{table}
\begin{ruledtabular}
\begin{tabular}{lcc}
$p$&$C_p$&$D_p$\\
\colrule
1.1&0.2972(2)&-0.103(5)\\
1.2&0.2756(2)&-0.114(6)\\
1.3&0.2561(3)&-0.116(6)\\
1.4&0.2392(3)&-0.126(8)\\
1.5&0.2231(4)&-0.122(9)\\
1.6&0.2098(2)&-0.134(5)\\
1.7&0.1976(2)&-0.143(5)\\
1.8&0.1864(2)&-0.146(4)\\
1.9&0.1756(2)&-0.142(4)\\
2.0&0.1671(3)&-0.161(8)\\
2.1&0.1580(1)&-0.154(3)\\
2.5&0.1305(2)&-0.166(4)\\
3.0&0.1067(2)&-0.178(6)\\
4.0&0.0788(3)&-0.200(7)\\
5.0&0.0648(3)&-0.226(7)\\
6.0&0.0580(4)&-0.26(1)\\
7.0&0.0563(7)&-0.31(2)\\
8.0&0.057(1)&-0.35(2)\\
9.0&0.061(2)&-0.43(4)\\
10&0.066(2)&-0.489(4)\\
\end{tabular}
\end{ruledtabular}
\caption{Numerical results for the average optimal cost and its finite size corrections in the assignment problem on the circumference. For each value of $p$, the average optimal cost has been evaluated averaging over at least $10^4$ instances with size between $N=10$ and $N=2.5\cdot 10^3$. Subsequently, a fit has been performed using Eq.~\eqref{scalingpbc} to extract $C_p$ and $D_p$.}
\label{tab:datipbc}
\end{table}
\endgroup
\paragraph{The $p<0$ case.} For $p<0$ a more detailed computation can be performed. At the leading order, the minimum is obtained for $t=\sfrac{1}{2}+O\left(\sfrac{1}{\sqrt N}\right)$, as in the case of open boundary condition. Under the assumption $t=\sfrac{1}{2}+\sfrac{\tau}{\sqrt N}$, we obtain
\begin{multline}
 \int_0^1\left[\frac{1}{2}-\frac{\left|\tau+\phi_{\sfrac{1}{2}+\sfrac{\tau}{\sqrt N}}^{(N)}(s)\right|}{\sqrt{N}}\right]^p\dd s=\\
 =\frac{1}{2^p} -\frac{p}{2^{p-1}\sqrt{N}}\int_0^1\left|\tau+\phi_{\sfrac{1}{2}}(s)\right|\dd s\\
 -\frac{p}{2^{p-1}\sqrt{N}}\int_0^1\left(\phi_{\sfrac{1}{2}}^{(N)}(s)-\phi_{\sfrac{1}{2}}(s)+\frac{\tau\left.\partial_t\phi_{t}(s)\right|_{t=\sfrac{1}{2}}}{\sqrt N}\right)\dd s \\+\frac{p(p-1)}{2^{p-1}N}\int_0^1\left(\tau+\phi_{\sfrac{1}{2}}(s)\right)^2\dd s+o\left(\frac{1}{N}\right).
\end{multline}
In the expression above we took into account that $\phi^{(N)}_{\sfrac{1}{2}}(s)-\phi_{\sfrac{1}{2}}(s)$ is infinitesimal quantity for large $N$, due to Eq.~\eqref{KMT}. Moreover, $\int_0^1\partial_t\phi_t(s)\dd s=\partial_t\int_0^1\phi_t(s)\dd s=0$, see Eq.~\eqref{bblimit}. We have then that the third contribution in the previous equation is $O(\sfrac{1}{N})$.
Minimizing respect to $\tau$, we obtain, up to higher order terms,
\begin{equation}
 \int_0^1\mathrm{sign}\left(\tau+\phi_{\sfrac{1}{2}}(s)\right)\dd s=0,
\end{equation}
and therefore the optimal value $\tau_\phi$ of $\tau$ depends on the instance 
$\phi_{\sfrac{1}{2}}$, i.e., on the properties of the Brownian bridge process, 
and not on $p$. The average optimal cost is
\begin{subequations}
\begin{equation}
 \mc{p}=\frac{1}{2^p} -\frac{p\lambda_1}{2^{p-1}\sqrt{N}} +\frac{p(p-1)\lambda_2}{2^{p-1}N}+o\left(\frac{1}{N}\right).
\end{equation}
where the quantities
\begin{align}
 \lambda_1&\coloneqq\overline{\int_0^1\left|\tau_\phi+\phi_{\sfrac{1}{2}}(s)\right|\dd s}=0.3217(5),\label{lambda1}\\
 \lambda_2&\coloneqq\overline{\int_0^1\left(\tau_\phi+\phi_{\sfrac{1}{2}}(s)\right)^2\dd s}=0.1717(5),\label{lambda2}
\end{align}\label{costopmin0PBC}\end{subequations}
are fixed numbers related to the Brownian 
bridge process only, which we evaluated numerically. We numerically verified 
Eq.~\eqref{costopmin0PBC}. Our numerical results are given in 
Fig.~\ref{fig:pmin0PBC} and they show a good agreement with the theoretical 
prediction.


\section{The random Euclidean matching problem}\label{sec:matching}
In the \remp in one dimension we associate to the set of $2N$ vertices of the complete graph $\mathcal K_{2N}$ a set of $2N$ points $\Xi_N\coloneqq\{x_i\}_{i=1,\dots,2N}$ independently and randomly generated on $\Lambda$ with uniform distribution. Again, we will assume that the points are labeled in such a way that $0\leq x_1<x_2<\dots<x_{2N}\leq 1$. In this case, a matching $\mu$ is any partition of $\Xi_{2N}$ in subsets of two elements only, its cardinality being $N$. 
We will consider the following \textit{matching cost} associated to $\mu$,
\begin{equation}
\mc{p}(\mu)\coloneqq \frac{1}{N}\sum_{(i,j)\in\mu}|x_i-x_j|^p,\quad p\in\mathds R.
\end{equation}
As in the assignment problem, we are interested in the {average}
\begin{equation}
\mc{p}\coloneqq\overline{\min_{\mu}\mc{p}(\mu)},
\end{equation}
and in its asymptotic behavior for $N\to+\infty$.

\subsection{Open boundary conditions} For $p>1$, the optimal solution on the interval has a simple structure. In particular, the couple $(x_i,x_j)$, $i<j$, belongs to the optimal matching if, and only if, $i$ is odd and $j=i+1$. This statement follows directly from the direct inspection of the $N=2$ case. We have indeed that given the generic configuration
\begin{center}
\begin{tikzpicture}
\draw[style=help lines,line width=1pt,black] (0,0) grid[step=1cm] (2.5,0);
\node[draw,circle,inner sep=1.5pt,fill=white,label=below:{\footnotesize $x_1$}] (x1) at (0.5,0) {};
\node[draw,circle,inner sep=1.5pt,fill=white,label=below:{\footnotesize $x_2$}] (x2) at (1,0) {};
\node[draw,circle,inner sep=1.5pt,fill=white,label=below:{\footnotesize $x_3$}] (x3) at (1.5,0) {};
\node[draw,circle,inner sep=1.5pt,fill=white,label=below:{\footnotesize $x_4$}] (x4) at (2,0) {};
\end{tikzpicture}\end{center}
the minimum cost configuration has always the structure
\begin{center}
\begin{tikzpicture}
\draw[style=help lines,line width=1pt,black] (0,0) grid[step=1cm] (2.5,0);
\node[draw,circle,inner sep=1.5pt,fill=white,label=below:{\footnotesize $x_1$}] (x1) at (0.5,0) {};
\node[draw,circle,inner sep=1.5pt,fill=white,label=below:{\footnotesize $x_2$}] (x2) at (1,0) {};
\node[draw,circle,inner sep=1.5pt,fill=white,label=below:{\footnotesize $x_3$}] (x3) at (1.5,0) {};
\node[draw,circle,inner sep=1.5pt,fill=white,label=below:{\footnotesize $x_4$}] (x4) at (2,0) {};
\draw[line width=1pt,gray] (x1) to[bend left=90] (x2);
\draw[line width=1pt,gray] (x3) to[bend left=90] (x4);
\end{tikzpicture}\end{center}
The study of the properties of the optimal matching is reduced therefore to the study of \textit{spacings} between successive random points on $\Lambda$. The optimal cost for $p>1$ is therefore given by
\begin{equation}
 \min_\mu\varepsilon_N^{(p)}(\mu)=\frac{1}{N}\sum_{i=1}^N\varphi_{2i-1}^p,\quad \varphi_i\coloneqq x_{i+1}-x_{i}.
\end{equation}
Let us first observe that the distribution of the \textit{ordered} set $\boldsymbol x=(x_1,\dots,x_{2N})$ is given by
\begin{equation}
\rho_N(\boldsymbol x)=(2N)!\prod_{i=0}^{2N}\theta(x_{i+1}-x_i),\quad x_0\equiv 0,\ x_{2N+1}\equiv 1.
\end{equation}
It follows that
\begin{equation}\label{varrhoN}
\varrho_N(\varphi_0,\dots,\varphi_{2N})=(2N)!\prod_{i=0}^{2N}\theta(\varphi_i)\quad \text{for }\sum_{i=0}^{2N}\varphi_i=1.
\end{equation}
In particular, this implies that for the spacing $\varphi_l$ we have
\begin{multline}
\varrho_{N}^{(1)}(\varphi_l)=(2N)!\left[\prod_{\substack{k=0\\k\neq l}}^{2N}\int_0^{+\infty}\dd \varphi_k\right]\delta\left(\sum_{j=0}^{2N}\varphi_j-1\right)\\
=(2N)!i^{2N}\lim_{\epsilon\to 0^+}\e^{\epsilon(1-\varphi_l)}\int_{-\infty}^{+\infty}\frac{\e^{-i\lambda (1-\varphi_l)}}{(\lambda+i\epsilon)^{2N}}\frac{\dd\lambda}{2\pi}\\
=\begin{cases}2N(1-\varphi_l)^{2N-1}&\text{for $0<\varphi_l<1$,}\\0&\text{otherwise.}\end{cases}
\end{multline}
Observe that the shape of the distribution is not dependent on $l$. Moreover,
\begin{multline}\label{m-mag1obc}
\varepsilon_N^{(p)}\equiv \overline{\varphi_l^p}=\frac{\Gamma(2N+1)\Gamma(1+p)}{\Gamma(2N+1+p)}\\
=\frac{1}{N^{p}}\, \frac{\Gamma(p+1)}{2^{p}}\left[ 1-\frac{p (p+1)}{4 N}+o\left(\frac{1}{N}\right)\right].
\end{multline}
The joint density distribution $\varrho^{(2)}$ of the couple $(\varphi_{i},\varphi_{j})$, $i\neq j$, can be similarly evaluated. As proven, for example, in Ref.~\cite{Pyke65}, we have that
\begin{multline}
\varrho_N^{(2)}(\varphi_{i},\varphi_{j})=\\
=2N(2N-1)(1-\varphi_{i}-\varphi_{j})^{2(N-1)}\theta(\varphi_i)\theta(\varphi_j)\theta(1-\varphi_i-\varphi_j),
\end{multline}
implying
\begin{equation}
\overline{\varphi_{i}\varphi_{j}}=\frac{1}{(2N+1)(2N+2)}.
\end{equation}
Observe once again that no dependence on $i$ and $j$ appears on the right hand side of the previous equations. It is clear that in this case $\varphi_{i}\sim N^{-1}$. We introduce the rescaled variables $\phi_i=2N\varphi_i$, whose asymptotic distribution, for $0\leq\phi_i\leq 2N$, is given by
\begin{subequations}
\begin{equation}
\hat \varrho_{N}^{(1)}(\phi_i)
=\e^{-\phi_i} \left[1 - \frac{\phi_i ( \phi_i-2)}{4 N} +o\left(\frac{1}{N}\right) \right].\end{equation}
Similarly, the joint probability distribution for $\phi_i,\phi_j\geq 0$, $0\leq\phi_i+\phi_j\leq 2N$, is
\begin{multline}
\hat \varrho_{N}^{(2)}(\phi_i,\phi_j)=\\
=\e^{-\phi_i-\phi_j}\left[1-\frac{(\phi_i\!+\!\phi_j)^2\!-\!4(\phi_i\!+\phi_j)\!+\!2}{4N}+o\left(\frac{1}{N}\right)\right].
\end{multline}
\end{subequations}
We can therefore write
\begin{subequations}
\begin{align}
\overline{\phi_i}&=1-\frac{1}{2N}+o\left(\frac{1}{N}\right),\\
\overline{\phi_i\phi_j}&=1-\frac{3}{2N}+o\left(\frac{1}{N}\right).
\end{align}\end{subequations}
This implies 
\begin{equation}
\lim_N\left(\overline{\phi_i\phi_j}-\overline{\phi_i}\,\overline{\phi_j}\right)=0.
\end{equation}

For $p<0$ it is easily seen that, for $N=2$, the optimal solution is always the crossing one, i.e., in the form 
\begin{center}
\begin{tikzpicture}
\draw[style=help lines,line width=1pt,black] (0,0) grid[step=1cm] (2.5,0);
\node[draw,circle,inner sep=1.5pt,fill=white,label=below:{\footnotesize $x_1$}] (x1) at (0.5,0) {};
\node[draw,circle,inner sep=1.5pt,fill=white,label=below:{\footnotesize $x_2$}] (x2) at (1,0) {};
\node[draw,circle,inner sep=1.5pt,fill=white,label=below:{\footnotesize $x_3$}] (x3) at (1.5,0) {};
\node[draw,circle,inner sep=1.5pt,fill=white,label=below:{\footnotesize $x_4$}] (x4) at (2,0) {};
\draw[line width=1pt,gray] (x1) to[bend left=90] (x3);
\draw[line width=1pt,gray] (x2) to[bend left=90] (x4);
\end{tikzpicture}\end{center}
This can be proved again by direct inspection, in the spirit of the analysis in Proposition \ref{PropA2} and observing that a crossing solution is always possible. It follows that the optimal matching on a set of $2N$ points on the interval is given by the set of couples $\{(x_i,x_{i+N})\}_{i=1,\dots,N}$, such that, in the pictorial representation above, each arc corresponding to a matched coupled crosses all the remaining $N-1$ arcs. The optimal cost per edge is
\begin{equation}
\min_\mu\mc{p}(\mu)=\frac{1}{N}\sum_{i=1}^N\left(x_{i+N}-x_i\right)^p.
\end{equation}
The analysis proceeds as in the $p>1$ case. To evaluate the average optimal cost, 
denoting by $\varphi_l\coloneqq x_{l+N}-x_l$, we have
\begin{equation}
\Pr[\varphi_l\in\dd \varphi] = B(N;2N,\varphi)\theta(\varphi)\theta(1-\varphi)\frac{N\dd\varphi}{\varphi}\label{distphipnegm}
\end{equation}
which is the probability that given $2N$ points at random  $N$ of them are  in an interval of length $\varphi$.
Of course the distribution of $\varphi_l$ does not depend on $l$. It follows that, for any real $\gamma$ such that $N>-\gamma$,
\begin{equation}
\overline{\varphi_l^\gamma}=\frac{\Gamma(2N+1)\Gamma(N+\gamma)}{\Gamma(N)\Gamma(2N+\gamma+1)}=\frac{1}{2^\gamma}+\frac{\gamma}{2^{\gamma+2}}\frac{\gamma-3}{N}+o\left(\frac{1}{N}\right), \label{phiallap}
\end{equation}
and therefore, for $p<0$ and $N>-p$,
\begin{equation}\label{m-min0obc}
\mc{p}=\frac{\Gamma(2N+1)\Gamma(N+p)}{\Gamma(N)\Gamma(2N+p+1)}=\frac{1}{2^p}+\frac{p(p-3)}{2^{p+2}N}+o\left(\frac{1}{N}\right).
\end{equation}
For $N\gg 1$ expectation values in the distribution given by Eq.~\eqref{distphipnegm} can be evaluated by the saddle-point method. By performing the shift around the saddle point value
\begin{equation}\label{varphiphi}
 \varphi_l = \frac{1}{2} + \frac{1}{2 \sqrt{N}} \phi_l,
\end{equation}
we recover the distribution for $\phi_l$ as
\begin{multline}
\Pr[\phi_l\in\dd\phi] \simeq  \\
\simeq \frac{\e^{-\phi^2}}{\sqrt{\pi}} \left[ 1 - \frac{1}{\sqrt{N}} \phi - \frac{4\phi^4 -8\phi^2 +1}{8N}+ o\left(\frac{1}{N}\right)\right]\dd\phi.
\end{multline}
For example, the evaluation of
\begin{multline}
\overline{\varphi_l^\gamma} = \frac{1}{2^\gamma} \overline{\left(1 + \frac{1}{\sqrt{N}} \phi_l\right)^\gamma}\\
\simeq \frac{1}{2^\gamma}\left[ 1 + \frac{\gamma}{\sqrt{N}} \overline{\phi_l} + \frac{\gamma (\gamma-1)}{2 N} \overline{\phi_l^2 }+ o\left(\frac{1}{N}\right)\right]
\end{multline}
will provide the result given in Eq.~\eqref{phiallap}, because
\begin{subequations}
\begin{align}
\overline{\phi_l} = & - \frac{1}{2 \sqrt{N} } + o\left(\frac{1}{N}\right)\\
\overline{\phi_l^2} = &\, \frac{1}{2} - \frac{1}{4N} + o\left(\frac{1}{N}\right)\, .
\end{align}
\end{subequations}

For the evaluation of correlations, we have, for $0\leq l<k\leq N$,
\begin{multline}
\Pr[x_l\in \dd x,x_{l+N}\in\dd x',x_{k}\in\dd y,x_{k+N}\in\dd y']=\\
=\dd x\dd x'\dd y\dd y'(2N)!\frac{x^{l-1}(y-x)^{k-l-1}(x'-y)^{N+l-k-1}}{\Gamma(l)\Gamma(k-l)\Gamma(N+l-k)}\\
\times\frac{(y'-x')^{k-l-1}(1-y')^{N-k}}{\Gamma(k-l)\Gamma(N-k+1)}\\
\times\theta(x)\theta(y-x)\theta(x'-y)\theta(y'-x')\theta(1-y').
\end{multline}
We get, therefore,
\begin{equation}\label{corrpnegm}
\overline{\varphi_l\varphi_k}=\frac{N(N+1)-|k-l|}{(2N+2)(2N+1)}.
\end{equation}
Eq.~\eqref{corrpnegm} suggests the introduction of the variables $x,y$, such that $Nx=l$, $Ny=k$, and, with reference to Eq.~\eqref{varphiphi}, of the field variable $\phi(x)\coloneqq\phi_{Nx}$. We have that in the large $N$ limit,
\begin{multline}
\overline{\phi(x)\phi(y)}-\overline{\phi(x)}\,\overline{\phi(y)}=\\
=\frac{1-2|y-x|}{2}+\frac{3|y-x|}{2N}+O\left(\frac{1}{N^2}\right).
\end{multline}
In Fig.~\ref{fig:mpmag1} and Fig.~\ref{fig:m-pmin0} we compare our theoretical results with the output of numerical simulations for the $p>1$ and the $p<0$ case, respectively. 

   \begin{figure*}[!ht]
     \subfloat[Numerical results for the finite size corrections to the average optimal cost (inset) in the matching problem with $p>1$ on $\Lambda$. We compare our numerical results with our numerical prediction given in Eq.~\eqref{m-mag1obc} (solid lines).\label{fig:mpmag1}]{ \includegraphics[width=0.45\textwidth]{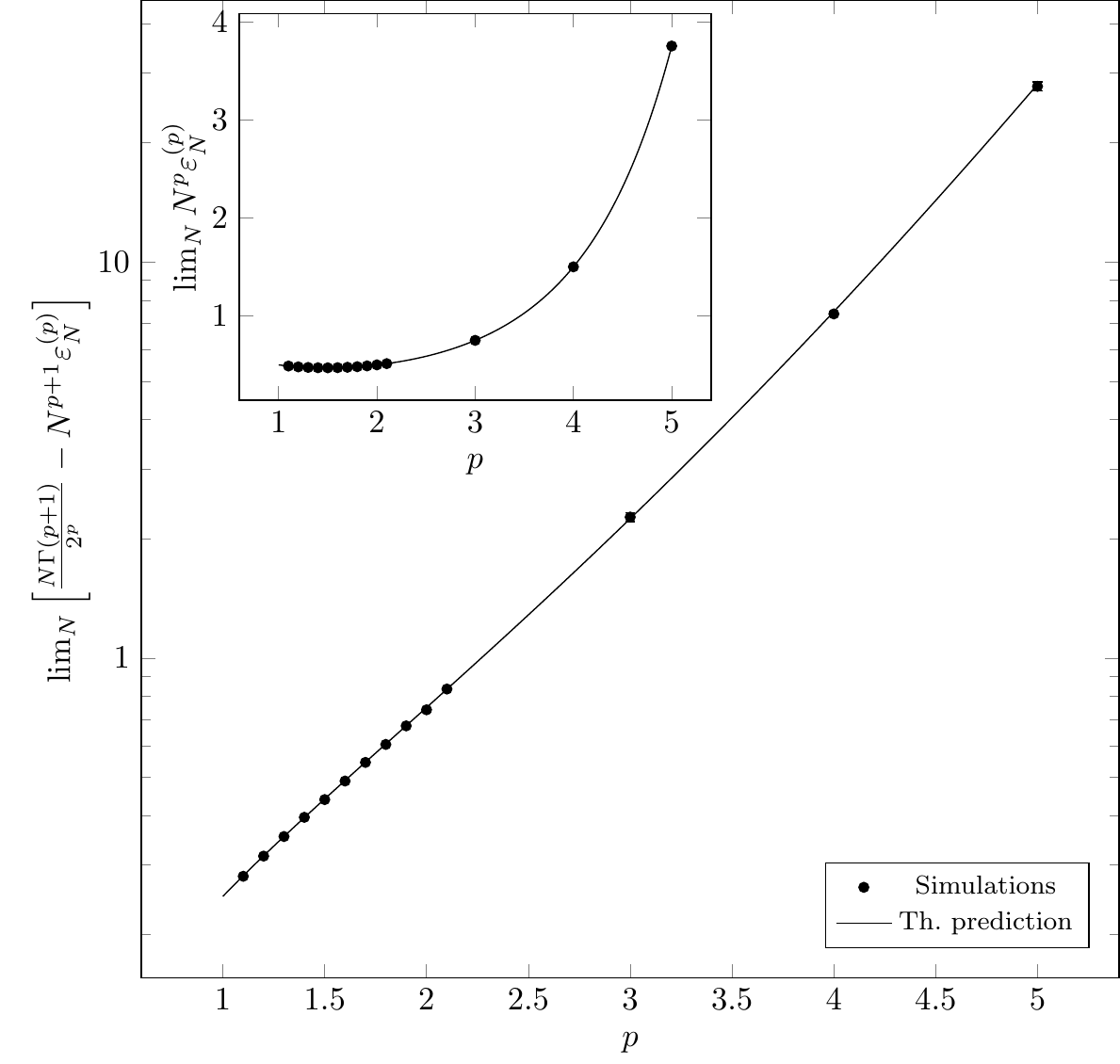}}
     \hfill
     \subfloat[Numerical results for average optimal cost in the matching problem on $\Lambda$ with $p<0$. The theoretical predictions (solid lines), for any value of $N>-p$, are given by Eq.~\eqref{m-min0obc}. Observe that for $N\gg 1$ $2^p\mc{p}$ is linear in $N^{-1}$ (inset). \label{fig:m-pmin0}]{\includegraphics[width=0.45\textwidth]{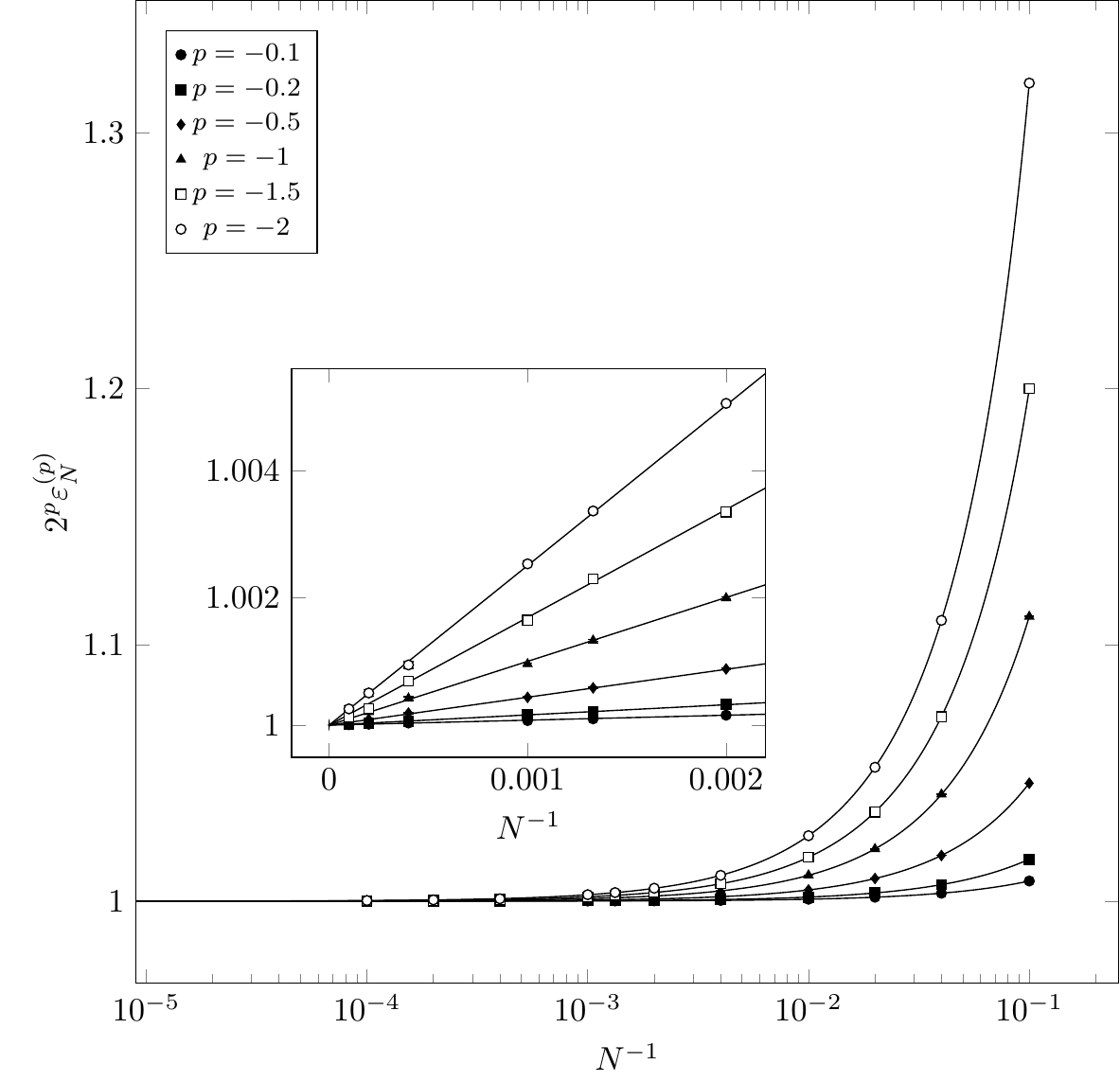}}\\
     \subfloat[Numerical results for average optimal cost in the matching problem on the circumference with $p>1$. The theoretical predictions (solid lines) for the asymptotic behavior are given by Eq.~\eqref{m-mag1pbc}. Observe that for $N\gg 1$ the corrections to the asymptotic cost scale as $\sfrac{1}{\sqrt{N}}$ (inset). \label{fig:m-pmag1pbc}]{%
       \includegraphics[width=0.45\textwidth]{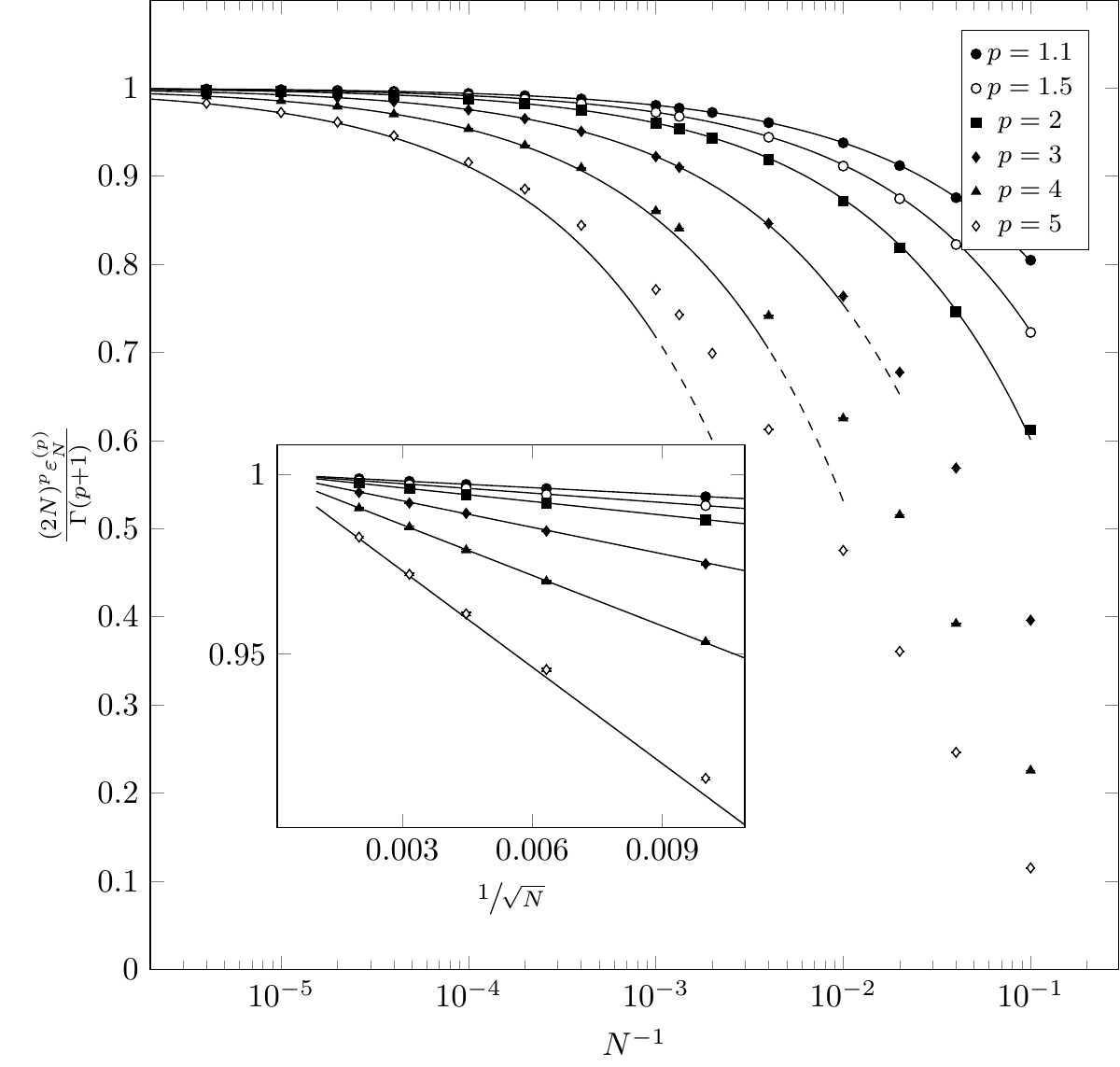}
     }\hfill
     \subfloat[Numerical results for average optimal cost in the matching problem on the circumference with $p<0$. The theoretical predictions (solid lines) are given by Eq.~\eqref{m-min0pbc} and they are correct for all values of $N>-p$, as expected. Observe that for $N\gg 1$ $2^p\mc{p}$ is linear in $\sfrac{1}{\sqrt{N}}$ (inset). \label{fig:m-pmin0pbc}]{%
       \includegraphics[width=0.45\textwidth]{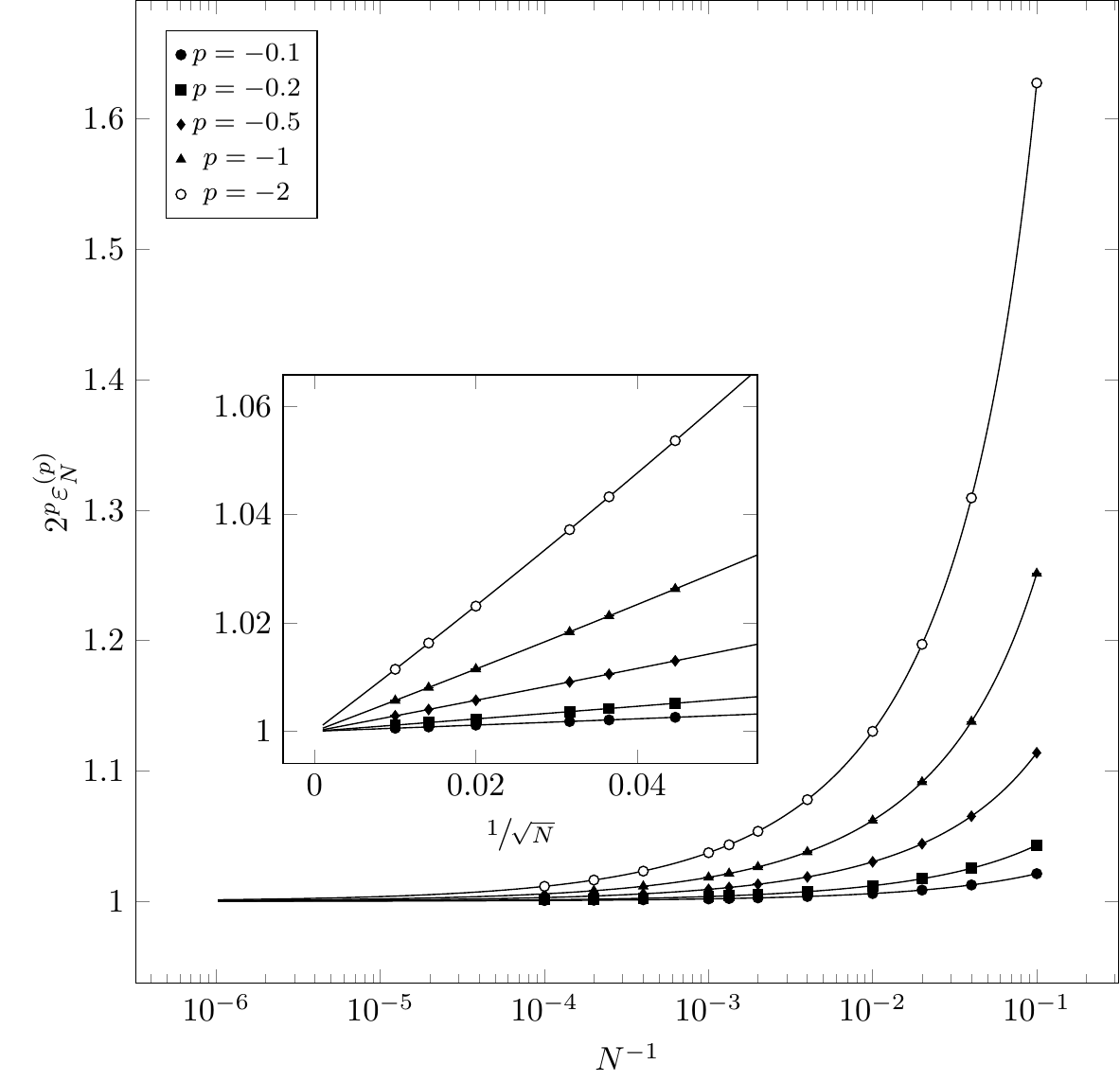}
     }
     \caption{Numerical results for the matching problem. Error bars are typically smaller than the markers in the figures.}
     \label{fig:matching}
   \end{figure*}
   
\subsection{Periodic boundary condition} The case of periodic boundary condition 
for $p>1$ can be easily obtained. Observe indeed that, in the case of $2N=4$ 
points on the circumference, given the crossing matching, we can always lower 
the cost considering one of the noncrossing solutions, i.e.,
\begin{center}
\begin{tikzpicture}[thick]
\begin{scope}
\newdimen\R
\R=0.5cm
\draw (0,0) circle (\R);
\node[draw,circle,inner sep=1.5pt,fill=white] (x1) at (45:\R) {};
\node[draw,circle,inner sep=1.5pt,fill=white] (x2) at (135:\R) {};
\node[draw,circle,inner sep=1.5pt,fill=white] (x3) at (225:\R) {};
\node[draw,circle,inner sep=1.5pt,fill=white] (x4) at (-45:\R) {};
\node[circle,inner sep=1pt,fill=white] (a) at (0:\R) {\tiny $a$};
\node[circle,inner sep=1pt,fill=white] (b) at (90:\R) {\tiny $b$};
\node[circle,inner sep=1pt,fill=white] (c) at (180:\R) {\tiny $c$};
\node[circle,inner sep=1pt,fill=white] (d) at (270:\R) {\tiny $d$};
\draw[line width=1pt,gray] (x1) to (x3);
\draw[line width=1pt,gray] (x2) to (x4);
\end{scope}
\draw [-stealth] (0.6,0) to (0.9,0);
\begin{scope}[shift={(1.5,0)}]
\R=0.5cm
\draw (0,0) circle (\R);
\node[draw,circle,inner sep=1.5pt,fill=white] (x1) at (45:\R) {};
\node[draw,circle,inner sep=1.5pt,fill=white] (x2) at (135:\R) {};
\node[draw,circle,inner sep=1.5pt,fill=white] (x3) at (225:\R) {};
\node[draw,circle,inner sep=1.5pt,fill=white] (x4) at (-45:\R) {};
\draw[line width=1pt,gray] (x1) to[out=225,in=135] (x4);
\draw[line width=1pt,gray] (x2) to[out=-45,in=45] (x3);
\end{scope}
\end{tikzpicture}
\end{center}
where the arrows denote the transition from a given matching to another one with lower cost. Indeed, with reference to the figure above, denoting by $a$, $b$, $c$, and $d$ the length of the four arcs with extremes the considered points, let us suppose, without loss of generality, that $a+b\leq\sfrac{1}{2}$ and $b+c\leq\sfrac{1}{2}$. Then, $a^p+c^p<(a+b)^p+(b+c)^p$ for $p>1$. Moreover, given the matching
\begin{center}
\begin{tikzpicture}[thick]
\begin{scope}
\newdimen\R
\R=0.5cm
\draw (0,0) circle (\R);
\node[draw,circle,inner sep=1.5pt,fill=white] (x1) at (55:\R) {};
\node[draw,circle,inner sep=1.5pt,fill=white] (x2) at (125:\R) {};
\node[draw,circle,inner sep=1.5pt,fill=white] (x3) at (175:\R) {};
\node[draw,circle,inner sep=1.5pt,fill=white] (x4) at (5:\R) {};
\node[circle,inner sep=0pt,fill=white] (a) at (30:\R) {\tiny $a$};
\node[circle,inner sep=0pt,fill=white] (b) at (90:\R) {\tiny $b$};
\node[circle,inner sep=0pt,fill=white] (c) at (150:\R) {\tiny $c$};
\node[circle,inner sep=0pt,fill=white] (d) at (270:\R) {\tiny $d$};
\draw[line width=1pt,gray] (x1) to[out=235,in=305] (x2);
\draw[line width=1pt,gray] (x3) to[out=355,in=185] (x4);
\end{scope}
\draw [-stealth] (0.6,0) to (0.9,0);
\begin{scope}[shift={(1.5,0)}]
\R=0.5cm
\draw (0,0) circle (\R);
\node[draw,circle,inner sep=1.5pt,fill=white] (x1) at (55:\R) {};
\node[draw,circle,inner sep=1.5pt,fill=white] (x2) at (125:\R) {};
\node[draw,circle,inner sep=1.5pt,fill=white] (x3) at (175:\R) {};
\node[draw,circle,inner sep=1.5pt,fill=white] (x4) at (5:\R) {};
\draw[line width=1pt,gray] (x1) to[out=235,in=185] (x4);
\draw[line width=1pt,gray] (x3) to[out=355,in=305] (x2);
\end{scope}
\end{tikzpicture}
\end{center}
if $a+b+c\leq\sfrac{1}{2}$, then $a^p+c^p\leq (a+b+c)^p+b^p$, i.e., there are no nested matchings in a half-circumference. Applying these rules iteratively to the case of $2N$ points on the circumference, we find that, ordering the points according to a reference orientation on the circumference, we have two possible optimal matching configurations, namely, for $i=1,\dots,N$, the $2i$-th point is associated either to the $(2i+1\Mod{2N})$-th point, or to the $(2i-1\Mod{2N})$-th point. Pictorially,
\begin{equation}
\begin{tikzpicture}[thick, baseline={([yshift=-.5ex]current bounding box.center)}]
\begin{scope}  
\newdimen\R
\R=0.5cm
\draw (0,0) circle (\R);
\node[draw,circle,inner sep=1.5pt,fill=white] (x1) at (30:\R) {};
\node[draw,circle,inner sep=1.5pt,fill=white] (x2) at (140:\R) {};
\node[draw,circle,inner sep=1.5pt,fill=white] (x3) at (170:\R) {};
\node[draw,circle,inner sep=1.5pt,fill=white] (x4) at (200:\R) {};
\node[draw,circle,inner sep=1.5pt,fill=white] (x5) at (290:\R) {};
\node[draw,circle,inner sep=1.5pt,fill=white] (x6) at (340:\R) {};
\end{scope}\draw [-stealth] (0.6,0) to (0.9,0);
\end{tikzpicture}
\min\left[
\begin{tikzpicture}[thick, baseline={([yshift=-.5ex]current bounding box.center)}]
\newdimen\R
\R=0.5cm
\draw (0,0) circle (\R);
\node[draw,circle,inner sep=1.5pt,fill=white] (x1) at (30:\R) {};
\node[draw,circle,inner sep=1.5pt,fill=white] (x2) at (140:\R) {};
\node[draw,circle,inner sep=1.5pt,fill=white] (x3) at (170:\R) {};
\node[draw,circle,inner sep=1.5pt,fill=white] (x4) at (200:\R) {};
\node[draw,circle,inner sep=1.5pt,fill=white] (x5) at (290:\R) {};
\node[draw,circle,inner sep=1.5pt,fill=white] (x6) at (340:\R) {};
\draw[line width=1pt,gray] (x1) to[out=210,in=160] (x6);
\draw[line width=1pt,gray] (x5) to[out=110,in=20] (x4);
\draw[line width=1pt,gray] (x3) to[out=350,in=320] (x2);
\end{tikzpicture},
\begin{tikzpicture}[thick, baseline={([yshift=-.5ex]current bounding box.center)}]
\newdimen\R
\R=0.5cm
\draw (0,0) circle (\R);
\node[draw,circle,inner sep=1.5pt,fill=white] (x1) at (30:\R) {};
\node[draw,circle,inner sep=1.5pt,fill=white] (x2) at (140:\R) {};
\node[draw,circle,inner sep=1.5pt,fill=white] (x3) at (170:\R) {};
\node[draw,circle,inner sep=1.5pt,fill=white] (x4) at (200:\R) {};
\node[draw,circle,inner sep=1.5pt,fill=white] (x5) at (290:\R) {};
\node[draw,circle,inner sep=1.5pt,fill=white] (x6) at (340:\R) {};
\draw[line width=1pt,gray] (x1) to[out=210,in=320] (x2);
\draw[line width=1pt,gray] (x3) to[out=350,in=20] (x4);
\draw[line width=1pt,gray] (x5) to[out=110,in=160] (x6);
\end{tikzpicture}\right]
\end{equation}
The distribution of $2N$ spacings $\{\varphi_i\}$ generated by $2N$ random points on the circumference is given by
\begin{equation}\label{spacingcirc}
 \varrho(\varphi_1,\dots\varphi_{2N})=\Gamma(2N)\delta\left(\sum_{i=1}^{2N}\varphi_i-1\right)\prod_{i=1}^N
\theta(\varphi_i).\end{equation}
We assume here that we choose one of the points as origin, and an orientation on the circumference, such that the intervals $\varphi_i$ are labeled accordingly. Let $p(\Phi_1,\dots,\Phi_{2N})$ be the probability for the quantities $\Phi_i\coloneqq N^p\varphi_i^p$, which can be straightforwardly obtained from Eq.~\eqref{spacingcirc}. The variables $\{\Phi_i\}_i$ have mean
\begin{subequations}
\begin{equation}
 \mu\coloneqq\overline{\Phi_i}=N^p\frac{\Gamma(2N)\Gamma(p+1)}{\Gamma(2N+p)},
\end{equation}
and variance
\begin{multline}
\overline{\left(\Phi_i-\mu\right)^2}=\\
= N^{2p}\left[\frac{\Gamma(2N)\Gamma(1+2p)}{\Gamma(2N+2p)}-\frac{\Gamma^2(2N)\Gamma^2(1+p)}{\Gamma^2(2N+p)}\right]\\
=\frac{\Gamma(2p+1)-\Gamma^2(p+1)}{2^{2p}}+o(1)\equiv\sigma^2+o(1).
\end{multline}
The variables are, however, not independent, due to the overall constraint $\sum_i\Phi_i^{\sfrac{1}{p}}=N$. We have that, for $i\neq j$,
\begin{multline}
\overline{\left(\Phi_i-\mu\right)\left(\Phi_j-\mu\right)}=\\
=N^{2p}\Gamma^2(1+p)\Gamma(2N)\left[\frac{1}{\Gamma(2N+2p)}-\frac{\Gamma(2N)}{\Gamma^2(2N+p)}\right]\\
=-\frac{p^2\Gamma^2(1+p)}{2^{2p+1}N}+o\left(\frac{1}{N}\right)=\frac{\rho}{N}+o\left(\frac{1}{N}\right).
\end{multline}
\end{subequations}
The optimal cost in the matching problem on the circumference with $p>1$ is given by
\begin{equation}
 \mc{p}=\overline{\min\left\{\frac{1}{N^{1+p}}\sum_{i=1}^N\Phi_{2i-1},\frac{1}{N^{1+p}}\sum_{i=1}^N\Phi_{2i}\right\}}.
\end{equation}
Using the results given in Appendix~\ref{App:minimum}, we have in this case that
\begin{multline}\label{m-mag1pbc}
 (2N)^p\mc{p}=\\
 =\Gamma(p+1)+\sqrt{\frac{\Gamma(2p+1)-\Gamma^2(p+1)}{\pi N}}+o\left(\frac{1}{\sqrt N}\right).
\end{multline}
In Fig.~\ref{fig:m-pmag1pbc} we compare our numerical results with the theoretical prediction in Eq.~\eqref{m-mag1pbc}. 

In the $p<0$ case, as in the case of open boundary conditions, we have that, 
given four points on the circumference, the optimal solution is always the 
crossing one. Let us consider indeed
\begin{center}
\begin{tikzpicture}[thick]
\begin{scope}
\newdimen\R
\R=0.5cm
\draw (0,0) circle (\R);
\node[draw,circle,inner sep=1.5pt,fill=white] (x1) at (45:\R) {};
\node[draw,circle,inner sep=1.5pt,fill=white] (x2) at (135:\R) {};
\node[draw,circle,inner sep=1.5pt,fill=white] (x3) at (225:\R) {};
\node[draw,circle,inner sep=1.5pt,fill=white] (x4) at (315:\R) {};
\node[circle,inner sep=0pt,fill=white] (a) at (90:\R) {\tiny $a$};
\node[circle,inner sep=0pt,fill=white] (b) at (180:\R) {\tiny $b$};
\node[circle,inner sep=0pt,fill=white] (c) at (270:\R) {\tiny $c$};
\node[circle,inner sep=0pt,fill=white] (d) at (0:\R) {\tiny $d$};
\draw[line width=1pt,gray] (x1) to (x3);
\draw[line width=1pt,gray] (x2) to (x4);
\end{scope}
\begin{scope}[shift={(-1.5,0)}]
\newdimen\R
\R=0.5cm
\draw (0,0) circle (\R);
\node[draw,circle,inner sep=1.5pt,fill=white] (x1) at (45:\R) {};
\node[draw,circle,inner sep=1.5pt,fill=white] (x2) at (135:\R) {};
\node[draw,circle,inner sep=1.5pt,fill=white] (x3) at (225:\R) {};
\node[draw,circle,inner sep=1.5pt,fill=white] (x4) at (315:\R) {};
\node[circle,inner sep=0pt,fill=white] (a) at (90:\R) {\tiny $a$};
\node[circle,inner sep=0pt,fill=white] (b) at (180:\R) {\tiny $b$};
\node[circle,inner sep=0pt,fill=white] (c) at (270:\R) {\tiny $c$};
\node[circle,inner sep=0pt,fill=white] (d) at (0:\R) {\tiny $d$};
\draw[line width=1pt,gray] (x1) to[out=225,in=315] (x2);
\draw[line width=1pt,gray] (x3) to[out=45,in=135] (x4);
\end{scope}
\draw [-stealth] (-0.9,0) to (-0.6,0);
\begin{scope}[shift={(1.5,0)}]
\newdimen\R
\R=0.5cm
\draw (0,0) circle (\R);
\node[draw,circle,inner sep=1.5pt,fill=white] (x1) at (45:\R) {};
\node[draw,circle,inner sep=1.5pt,fill=white] (x2) at (135:\R) {};
\node[draw,circle,inner sep=1.5pt,fill=white] (x3) at (225:\R) {};
\node[draw,circle,inner sep=1.5pt,fill=white] (x4) at (315:\R) {};
\node[circle,inner sep=0pt,fill=white] (a) at (90:\R) {\tiny $a$};
\node[circle,inner sep=0pt,fill=white] (b) at (180:\R) {\tiny $b$};
\node[circle,inner sep=0pt,fill=white] (c) at (270:\R) {\tiny $c$};
\node[circle,inner sep=0pt,fill=white] (d) at (0:\R) {\tiny $d$};
\draw[line width=1pt,gray] (x1) to[out=225,in=135] (x4);
\draw[line width=1pt,gray] (x3) to[out=45,in=315] (x2);
\end{scope}
\draw [-stealth] (0.9,0) to (0.6,0);
\end{tikzpicture}
\end{center}
and let us assume, without loss of generality, that $a+b\leq \sfrac{1}{2}$ and $b+c\leq\sfrac{1}{2}$. We have that $a^p+c^p\geq (a+b)^p+(b+c)^p$. With reference to the figure above, if $d\geq\sfrac{1}{2}$, then we also have $b^p+(a+b+c)^p>(a+b)^p+(b+c)^p$, where we have used the fact that $f(x,a)=x^p-(x+a)^p$ is a decreasing function for $x>0$ and $a>0$. If $d=1-a-b-c\leq\sfrac{1}{2}$, then $b\leq d$ and therefore we have $b^p+d^p\geq 2d^p\geq (a+b)^p+(b+c)^p$. This fact implies that the minimum cost matching is obtained coupling the $i$th point to the $i+N\Mod{2N}$ point on the circumference, where the points are supposed to be ordered according to a reference orientation on the circumference. For example, we will have that
\begin{center}
\begin{tikzpicture}[thick]
\begin{scope}
\newdimen\R
\R=0.5cm
\draw (0,0) circle (\R);
\node[draw,circle,inner sep=1.5pt,fill=white] (x1) at (30:\R) {};
\node[draw,circle,inner sep=1.5pt,fill=white] (x2) at (140:\R) {};
\node[draw,circle,inner sep=1.5pt,fill=white] (x3) at (170:\R) {};
\node[draw,circle,inner sep=1.5pt,fill=white] (x4) at (200:\R) {};
\node[draw,circle,inner sep=1.5pt,fill=white] (x5) at (290:\R) {};
\node[draw,circle,inner sep=1.5pt,fill=white] (x6) at (340:\R) {};
\end{scope}
\begin{scope}[shift={(1.5,0)}]
\newdimen\R
\R=0.5cm
\draw (0,0) circle (\R);
\node[draw,circle,inner sep=1.5pt,fill=white] (x1) at (30:\R) {};
\node[draw,circle,inner sep=1.5pt,fill=white] (x2) at (140:\R) {};
\node[draw,circle,inner sep=1.5pt,fill=white] (x3) at (170:\R) {};
\node[draw,circle,inner sep=1.5pt,fill=white] (x4) at (200:\R) {};
\node[draw,circle,inner sep=1.5pt,fill=white] (x5) at (290:\R) {};
\node[draw,circle,inner sep=1.5pt,fill=white] (x6) at (340:\R) {};
\draw[line width=1pt,gray] (x1) to[out=210,in=20] (x4);
\draw[line width=1pt,gray] (x5) to[out=110,in=320] (x2);
\draw[line width=1pt,gray] (x3) to[out=350,in=160] (x6);
\end{scope}
\draw [-stealth] (0.6,0) to (0.9,0);
\end{tikzpicture}
\end{center}
To find the average optimal cost, observe that, fixing the origin of our reference system in the point $i$, the distance of the point $i+N\Mod{N}$ from the $i$th point on the circumference is distributed as
\begin{multline}
 \Pr[\varphi_i\in\dd\varphi]=N\theta(\varphi)\theta\left(\frac{1}{2}-\varphi\right)\dd\varphi\times\\
 \times\left[\frac{B(N;2N,\varphi)}{\varphi}+\frac{B(N;2N,1-\varphi)}{1-\varphi}\right].
\end{multline}
As in the case of periodic boundary conditions, the previous distribution does not depend on $i$. We obtain, for $N>-p$, the average optimal cost straightforwardly as
\begin{multline}\label{m-min0pbc} \mc{p}=N\binom{2N}{N}\beta_{\sfrac{1}{2}}\left(N+p,N\right)\\
 =\frac{1}{2^p} \left[ 1 - \frac{p}{\sqrt{\pi N}} + \frac{p(p-1)}{4 N} + o\left(\frac{1}{N}\right)\right],
\end{multline}
where we have introduced the incomplete Beta function
\begin{equation}
 \beta_s(a,b)\coloneqq \int_0^s t^{a-1}(1-t)^{b-1}\dd t.
\end{equation}
In Fig.~\ref{fig:m-pmin0pbc} we show that the results of our numerical simulations are in agreement with Eq.~\eqref{m-min0pbc}. 

\section{Conclusions}\label{sec:conclusioni}
In the present paper we discussed the Euclidean matching problem and the Euclidean assignment problem on a set of $2N$ points both on the line and on the circumference. 

We first stated some fundamental properties of the Euclidean assignment problem on the line for a large class of cost functions $c(z)$, which we called $\mathcal{C}$-functions, and for strictly increasing cost functions. We proved that, for these classes of cost functions, the optimal matching $x_i\to y_{\pi(i)}$ between the set of points $0\leq x_1<x_2<\dots<x_N\leq 1$ and the set of points $0\leq y_1<y_2<\dots<y_N\leq 1$ can be expressed as a permutation in the form $\pi(i)=i+k\pmod{N}$ for some $k$, in the case of strictly increasing cost functions the optimal permutation being the identical permutation, $\pi(i)=i$. We considered then the assignment problem both on the line and on the circumference in presence of disorder, assuming the points  to be uniformly and randomly generated on the considered domain. We chose the cost function $c(z)=z^p$ with $p\in\mathds R\setminus[0,1]$, which is a $\mathcal C$-function for $p<0$ and a strictly increasing function for $p>1$. The analytical investigation allowed us to relate the optimal solution, in all the considered cases, to a well-known Gaussian stochastic process, namely the Brownian bridge process, in the $N\to+\infty$ limit. Then, we analytically derived the expression for the average optimal cost and its finite-size corrections for the considered range of values of $p$, and we gave an explicit expression of the correlation functions for the optimal solutions.
In particular, we computed
\begin{equation}
\varepsilon_N^{(p)} = 
\begin{cases}
\frac{\Gamma\left(1+\sfrac{p}{2}\right)}{p+1}\!\left[1\!-\!\frac{1}{N}\frac{p(p+2)}{8}\!+\! o\!\left(\frac{1}{N}\right)\right]\!\frac{1}{N^{\sfrac{p}{2}}}& \hbox{for } p>1,\\
\frac{1}{2^p}\left[1+\frac{1}{N}\frac{p(p-2)(p-4)}{3(p-3)}+o\left(\frac{1}{N}\right)\right] & \hbox{for } p<0, 
\end{cases}
\end{equation}
and the equivalent results on the unit circumference,
\begin{equation}
\varepsilon_N^{(p)} = 
\begin{cases}
\left[ C_p  + \frac{D_p}{{N}} + o\left(\frac{1}{N}\right) \right]\frac{1}{N^{\sfrac{p}{2}}} & \hbox{for } p>1,\\
\frac{1}{2^p}\left[1 -\frac{2p\lambda_1}{\sqrt{N}} +\frac{2p(p-1)\lambda_2}{N}+ o\left(\frac{1}{N}\right)\right] & \hbox{for } p<0,
\end{cases}
\end{equation}
where the constants $\lambda_1$ and $\lambda_2$ were defined in Eq.~\eqref{lambda1} and Eq.~\eqref{lambda2}, respectively. Unfortunately, in the $p>1$ case, only $C_2$ is known analytically.

We analyzed in a similar way the Euclidean matching problem.
In particular, for the average cost on the unit interval we computed the constants appearing in the expansions
\begin{equation}
\varepsilon_N^{(p)} = 
\begin{cases}
\frac{\Gamma(p+1)}{2^{p}}\left[1-\frac{p (p+1)}{4 N}+ o\left(\frac{1}{N^2}\right)\right]\frac{1}{N^p} & \hbox{for } p>1, \\
\frac{1}{2^p}\left[1+\frac{p(p-3)}{4N} + o\left(\frac{1}{N^2}\right)\right] & \hbox{for } p<0, 
\end{cases}
\end{equation}
and those for the problem on the unit circumference, where
\begin{equation}
\varepsilon_N^{(p)} = 
\begin{cases}
\frac{\Gamma(p+1)}{2^p}\left[1+\sqrt{\left(\frac{\Gamma(2p+1)}{\Gamma^2(p+1)}-1\right)\frac{1}{\pi N}}+ o\left(\frac{1}{\sqrt N}\right)\right]\frac{1}{N^p}\\\hfill\hbox{for } p>1,\\
\frac{1}{2^p}\left[1-\frac{p}{\sqrt{\pi N}}+\frac{p(p-1)}{2^{p+2}N}+ o\left(\frac{1}{N}\right)\right]\\\hfill \hbox{for } p<0.
\end{cases}
\end{equation}
The first remark is the different leading power of $N$ appearing here, that is $N^{-p}$, at variance with the assignment case where it was $N^{-\sfrac{p}{2}}$.
Second, we observe that, both in the case of the assignment problem with $p<0$ and in the case of the matching problem, the finite-size corrections to the average optimal cost change their scaling properties when open boundary conditions are replaced by periodic boundary conditions, i.e., when we consider the problem on the circumference instead of the interval. In particular, 
in the case of open boundary conditions the finite-size corrections scale as $O\left(\sfrac{1}{N}\right)$, whereas in the case of periodic boundary conditions, they scale as $O\left(\sfrac{1}{\sqrt N}\right)$. 
%
This fact can be observed both in Fig.~\ref{fig:assignment} and in Fig.~\ref{fig:matching}.

\section{Acknowledgments} The authors thank Carlo Lucibello, Giorgio Parisi, Filippo Santambrogio and Andrea Sportiello for useful discussions. The work of G.S.~was supported by the Simons Foundation (Grant No.~454949).
\appendix
\section{On the convexity property of $\mathcal C$-functions}\label{app:convexity}
In this appendix, we show that Eq.~\eqref{prima} is equivalent to convexity if the function $f$ is continuous on the interval $(0,1)$. Introducing
\begin{equation}
R(z_1,z_2)\coloneqq\frac{f(z_1)-f(z_2)}{z_1-z_2},
\end{equation}
Eq.~\eqref{prima} can be written as
\begin{equation}
 R(z_1,z_2)\leq R(z_1+ \eta,z_2+\eta).
\end{equation}
Observe that $R(z_1,z_2)$ is a symmetric function of its arguments, and, moreover, for $n\in\mathds N$, we can write
\begin{subequations}\label{nondec}
\begin{multline}
R(z_1,z_2)=\\
=\frac{1}{n}\sum_{k=0}^{n-1}R\left(z_1+k\frac{z_2-z_1}{n},z_2+(k+1)\frac{z_2-z_1}{n}\right)\\
 \leq R\left(z_1,z_2+\frac{z_2-z_1}{n}\right),
\end{multline}
and 
\begin{multline}
R\left(z_1,z_2+n(z_2-z_1)\right)=\\
=\frac{1}{n}\sum_{k=0}^{n-1}R\left(z_1+k(z_2-z_1),z_2+(k+1)(z_2-z_1)\right)\\
 \geq R\left(z_1,z_2\right).
\end{multline}
\end{subequations}
Eqs.~\eqref{nondec} are equivalent to say that the function $R(z_1,z_2)$ is 
monotonically non decreasing respect to each one of its arguments taking the 
other one fixed, provided that the ratio of the considered intervals is 
rational. If the function is continuous, we can extend this property to an 
arbitrary couple of intervals, and therefore, for $\eta>0$, we can simply state 
the stronger chain of inequalities 
\begin{equation}
R(z_1,z_2) \leq R(z_1, z_2 + \eta)\leq R(z_1+\eta, z_2 + \eta),
\end{equation}
that is $R$ is monotonically non decreasing respect to each of its arguments taking the other one fixed. This property is equivalent to convexity. Indeed, if we consider $\xi=tx_1+(1-t)x_2$ with $t\in[0,1]$ and $x_1<x_2$,
\begin{subequations}
 \begin{align}
  R(\xi,x_1)&\leq R(\xi,x_2),\\
\frac{f(\xi)-f(x_1)}{(1-t)(x_2-x_1)}&\leq \frac{f(x_2)-f(\xi)}{t(x_2-x_1)},\\
f(\xi)&\leq t f(x_1)+(1-t)f(x_2).
 \end{align}
\end{subequations}

\section{Derivation of Eq.~\eqref{distphi}}\label{app:distphi}
In this appendix we will sketch the derivation of Eq.~\eqref{distphi}. We have to evaluate
\begin{multline}
\Pr[\varphi_k\in\dd\varphi]
=\dd\varphi\binom{N}{k}^2 k^2\iint_0^1\! \delta\left(\varphi-y+x\right)\times\\(xy)^{k-1}[(1-x)(1-y)]^{N-k}\dd x\dd y
\end{multline}
for $N\gg 1$. Let us write now, for $N\gg 1$,
\begin{equation}
 k=Ns+\frac{1}{2},\quad s\in(0,1).
\end{equation}
The integral above can be written as
\begin{multline}
 \frac{\Gamma^2(N+1)}{\Gamma^2(Ns\!+\!\sfrac{1}{2})\Gamma^2(N\!+\!\sfrac{1}{2}\!-\!Ns)}\iint_0^1\!\frac{\delta\left(\varphi\!-\!y\!+\!x\right)}{\sqrt{xy(1-x)(1-y)}}\\
 \times \e^{Ns\ln(xy)+N(1-s)\ln[(1-x)(1-y)]}\dd x\dd y.
\end{multline}
We will evaluate the integral above using the saddle point method. In particular, the saddle point $(x_\text{sp},y_\text{sp})$ is obtained for
\begin{equation}
  x_\text{sp}=y_\text{sp}=s.
 \end{equation}
Observe now that, at fixed $s$, for $N\gg 1$,
\begin{multline}
 \frac{\Gamma(N+1)}{\Gamma(Ns\!+\!\sfrac{1}{2})\Gamma(N\!+\!\sfrac{1}{2}\!-\!Ns)}=\\
 =\sqrt{N}\frac{\e^{-N\left[s\ln s+(1-s)\ln(1-s)\right]}}{\sqrt{2\pi}}\\
 \times\left[1\!+\!\frac{1+2s(1-s)}{24 N s(1-s)}\!+\!o\left(\frac{1}{N}\right)\right],
\end{multline}
where we have used the Stirling expansion for $N\gg 1$
\begin{equation}
 N!=\sqrt{2\pi N}\left(\frac{N}{\e}\right)^N\left[1+\frac{1}{12N}+o\left(\frac{1}{N}\right)\right].
\end{equation}
We have, therefore,
\begin{widetext}
\begin{multline}
 \frac{\Gamma^2(N+1)}{\Gamma^2(Ns\!+\!\sfrac{1}{2})\Gamma^2(N\!+\!\sfrac{1}{2}\!-\!Ns)}\iint_0^1\!\frac{\delta\left(\varphi-y+x\right)}{\sqrt{xy(1-x)(1-y)}}\e^{Ns\ln(xy)+N(1-s)\ln[(1-x)(1-y)]}\dd x\dd y\\
 =\frac{N}{2\pi}\left[1+\frac{1+2s(1-s)}{24 N s(1-s)}+\!o\left(\frac{1}{N}\right)\right]^2\iint_0^1\!\frac{\delta\left(\varphi-y+x\right)}{\sqrt{xy(1-x)(1-y)}}\\
 \exp\left[-N\frac{(x\!-\!s)^2\!+\!(y\!-\!s)^2}{2s(1-s)}+N(1-2s)\frac{(x\!-\!s)^3\!+\!(y\!-\!s)^3}{3s^2(1-s)^2}-N\left(1-3s(1-s)\right)\frac{(x\!-\!s)^4\!+\!(y\!-\!s)^4}{4s^3(1-s)^3}+\dots\right]\dd x\dd y.
\end{multline}
The previous expression suggests the introduction of the set of variables
\begin{subequations}
\begin{align}
 \xi\coloneqq \sqrt N(x-s),\\
 \eta\coloneqq \sqrt N(y-s),
\end{align}
\end{subequations}
in such a way that the integral becomes
\begin{multline}
\frac{\Gamma^2(N+1)}{\Gamma^2(Ns\!+\!\sfrac{1}{2})\Gamma^2(N\!+\!\sfrac{1}{2}\!-\!Ns)}\iint_0^1\!\frac{\delta\left(\varphi-y+x\right)}{\sqrt{xy(1-x)(1-y)}}\e^{Ns\ln(xy)+N(1-s)\ln[(1-x)(1-y)]}\dd x\dd y\\
 =\frac{1}{2\pi s(1-s)}\left[1+\frac{1+2s(1-s)}{24 N s(1-s)}+\!o\left(\frac{1}{N}\right)\right]^2\iint_{-\infty}^{+\infty} \frac{\delta\left(\varphi-\frac{\eta}{\sqrt N}+\frac{\xi}{\sqrt N}\right)}{\sqrt{\left(1-\frac{\xi}{(1-s)\sqrt{N}}\right)\left(1-\frac{\eta}{(1-s)\sqrt{N}}\right)\left(1+\frac{\xi}{s\sqrt{N}}\right)\left(1+\frac{\eta}{s\sqrt{N}}\right)}}\\
 \times\exp\left[-\frac{\xi^2+\eta^2}{2s(1-s)}+\frac{1-2s}{3s^2(1-s)^2}\frac{\xi^3+\eta^3}{\sqrt N}-\frac{1-3s(1-s)}{4(1-s)^3s^3}\frac{\xi^4+\eta^4}{N}+o\left(\frac{1}{N}\right)\right]\dd \xi\dd\eta.
\end{multline}
\end{widetext}
Introducing $\phi=\sqrt N\varphi$, the distribution in Eq.~\eqref{distphi} is obtained through a series expansion for $N\gg 1$ and performing the Gaussian integrals.

\section{On the minimum of asymptotically uncorrelated exchangeable variables}\label{App:minimum}
Let us consider a vector $\boldsymbol{\Phi}$ of $2N$ continuous variables, $\boldsymbol{\Phi}=(\Phi_1,\dots,{\Phi}_{2N})$, and let us assume that their joint probability distribution density is given by $p(\boldsymbol\Phi)$. We assume that the $2N$ random variables are exchangeable, i.e., such that $p(\Phi_1,\dots,\Phi_{2N})=p(\Phi_{\pi(1)},\dots,\Phi_{\pi(2N)})$ for any permutation $\pi\in\mathcal S_{2N}$ \cite{Aldous1985}. In the following, we will denote the expectation respect to the probability density $p(\boldsymbol\Phi)$ by $\mathbb{E}(\bullet)$. Exchangeability implies
\begin{subequations}
 \begin{align}
  \mu&\coloneqq\mathbb E\left(\Phi_i\right),\\
  \sigma^2&\coloneqq\mathbb E\left[\left(\Phi_i-\mu\right)^2\right],
 \end{align}
which we suppose to remain finite for $N\to+\infty$. We also make the assumption that the $2N$ components of the vector $\boldsymbol{\Phi}$ are weakly correlated, i.e., the covariance is given by
\begin{equation}
 \frac{\rho}{N}\coloneqq\mathbb E\left[\left(\Phi_i-\mu\right)\left(\Phi_j-\mu\right)\right],
\end{equation}
\end{subequations}
in such a way that it vanishes as $O(\sfrac{1}{N})$ when 
$N\to+\infty$, asymptotically recovering independence. Given a subset 
$\{\Phi_{l(i)}\}_{i=1,\dots,K}$ of $1\leq K\leq 2N$ different components of 
$\boldsymbol\Phi$, it is easily seen that
\begin{equation}
 0\leq\mathbb E\left[\left(\sum_{i=1}^{K}\Phi_{l(i)}-K\mu\right)^2\right]=K\sigma^2+K\frac{K-1}{N}\rho.
\end{equation}
For $K=2N$, the relation above implies, for $N\to+\infty$, $\sigma^2+2\rho\geq 
0$, whereas from $K=N$ we obtain $\sigma^2+\rho>0$ in the same limit. Let us 
now partition the $2N$ components of $\boldsymbol\Phi$ in two subsets with the 
same cardinality, for example, the entries with even and odd labels, and consider
\begin{equation}\label{evenodd}
  \epsilon_{\text{o}}\coloneqq\frac{1}{N}\sum_{i=1}^N\Phi_{2i-1},\quad
  \epsilon_{\text{e}}\coloneqq\frac{1}{N}\sum_{i=1}^N\Phi_{2i}.
 \end{equation}
We want to evaluate the mean value of $\epsilon\coloneqq\min\{\epsilon_{\text{o}},\epsilon_{\text{e}}\}$ for $N\gg 1$. The joint distribution of $\epsilon_{\text{o}}$ and $\epsilon_{\text{e}}$ is
\begin{widetext}
\begin{multline}
 P(\epsilon_{\text{o}},\epsilon_{\text{e}})=\mathbb E\left[\delta\left(\epsilon_{\text{o}}-\frac{1}{N}\sum_{i=1}^N\Phi_{2i-1}\right)\delta\left(\epsilon_{\text{e}}-\frac{1}{N}\sum_{i=1}^N\Phi_{2i}\right)\right]\\=\iint_{-\infty}^{+\infty}\frac{\dd\lambda_{\text{e}}\dd\lambda_{\text{o}}}{4\pi^2}\mathbb E\left\{\exp\left[-i\lambda_{\text{e}}\left(\epsilon_{\text{e}}-\frac{1}{N}\sum_{k=1}^N\Phi_{2k}\right)-i\lambda_{\text{o}}\left(\epsilon_{\text{o}}-\frac{1}{N}\sum_{k=1}^N\Phi_{2k-1}\right)\right]\right\}\\
 =\iint_{-\infty}^{+\infty}\dd\lambda_{\text{e}}\dd\lambda_{\text{o}}\frac{\exp\left[-i\lambda_{\text{e}}(\epsilon_{\text{e}}-\mu)-i\lambda_{\text{o}}(\epsilon_{\text{o}}-\mu)-\frac{\lambda_{\text{e}}^2+\lambda_{\text{o}}^2}{2N}\sigma^2-\frac{(\lambda_{\text{e}}+\lambda_{\text{o}})^2}{2N}\rho\right]}{4\pi^2}
 \left[1+o\left(\frac{1}{N}\right)\right].
 \\=\frac{N}{2\pi\sigma\sqrt{\sigma^2+2\rho}}\exp\left[-\frac{N}{2\sigma^2\left(\sigma^2+2\rho\right)}\begin{pmatrix}\epsilon_{\text{e}}-\mu&\epsilon_{\text{o}}-\mu\end{pmatrix}\begin{pmatrix}\sigma^2+\rho&-\rho\\-\rho&\sigma^2+\rho\end{pmatrix}\begin{pmatrix}\epsilon_{\text{e}}-\mu\\\epsilon_{\text{o}}-\mu\end{pmatrix}\right]\left[1+o\left(\frac{1}{N}\right)\right],
\end{multline}
\end{widetext}
which is a bivariate Gaussian distribution function. The distribution $P(\epsilon)$ of the minimum $\epsilon\coloneqq \min\{\epsilon_e,\epsilon_o\}$ is therefore given by
\begin{equation}
 \begin{split}
  P(\epsilon)&=-\frac{\partial}{\partial\epsilon}\iint_{-\infty}^{+\infty}\theta(\epsilon_{\text{e}}-\epsilon)\theta(\epsilon_{\text{o}}-\epsilon)P(\epsilon_\text{e},\epsilon_\text{o})\dd\epsilon_\text{e}\dd\epsilon_\text{o}\\
 &= 2 \int_{\epsilon}^{\infty} P(\hat \epsilon, \epsilon)\dd\hat\epsilon,
 \end{split}
\end{equation}
which gives, in the limits of our approximations,
\begin{multline}
P(\epsilon)=\sqrt{\frac{N}{2\pi}}\frac{\e^{-\frac{N}{2}\frac{(\epsilon-\mu)^2}{\sigma^2+\rho}}}{\sqrt{\sigma^2+\rho}} \\
\times\left[ 1 - \mathrm{erf}\left(\sqrt{\frac{N}{2}}\frac{(\epsilon-\mu)\sigma}{\sqrt{(\sigma^2 + 2 \rho)(\sigma^2+ \rho)}}\right)\right],
\end{multline}
with $\mathrm{erf}(x)\coloneqq\sfrac{2}{\sqrt\pi}\int_0^{x}\exp(-z^2)\dd z$. In this form we see that the function $P(\epsilon)$ depends on $x=\sqrt{N} (\epsilon - \mu) $ and the distribution is a product of an even function and an odd function of $x$. In the spirit of the approximation, the domain $\mathcal D$ of $x$ is substituted, for $N\gg 1$, with the entire real line, up to exponentially small corrections, being the probability distribution concentrated around $\mu$. We immediately obtain
\begin{subequations}
\begin{align}
\int_{\mathcal D} P(\epsilon)\dd\epsilon&\simeq  \int_{-\infty}^\infty\frac{\e^{-\frac{1}{2}\frac{x^2}{\sigma^2+\rho}}}{\sqrt{\sigma^2+\rho}}\frac{\dd x}{\sqrt{2\pi}}=1.\\
N \int_{\mathcal D}(\epsilon-\mu)^2P(\epsilon)\dd\epsilon &\simeq\int_{-\infty}^\infty \, \frac{\e^{-\frac{1}{2}\frac{x^2}{\sigma^2+\rho}}}{\sqrt{\sigma^2+\rho}}\frac{\dd x}{\sqrt{2\pi}}=\sigma^2 + \rho.
\end{align}
More interestingly,
\begin{multline}
\sqrt{N} \int_{\mathcal D} (\epsilon-\mu)P(\epsilon) \dd\epsilon\simeq\\
\simeq- \int_{-\infty}^\infty x\frac{\e^{-\frac{1}{2}\frac{x^2}{\sigma^2+\rho}}}{\sqrt{\sigma^2+\rho}}\mathrm{erf}\left(\sqrt{\frac{1}{2}}\frac{x\, \sigma}{\sqrt{(\sigma^2 + 2 \rho)(\sigma^2+ \rho)}}\right)\frac{\dd x}{\sqrt{2\pi}}   \\ =- \frac{\sigma}{\sqrt{\pi}},\label{minimomedio}
\end{multline}
\end{subequations}
a result showing that, up to higher order terms, there is no influence of the weak correlation $\rho$ on the expectation value of $\epsilon$. This can be seen in a different way introducing the variables
\begin{equation}
 X_-\coloneqq\sqrt{N}\frac{\epsilon_\text{o}-\epsilon_\text{e}}{\sqrt{2}},\quad X_+\coloneqq \sqrt{N}\frac{\epsilon_\text{o}+\epsilon_\text{e}-2\mu}{\sqrt{2}}
\end{equation}
in the distribution $P(\epsilon_\text{e},\epsilon_\text{o})$, which gives the new distribution
\begin{equation}
 P_X(X_-,X_+)=\frac{ \exp \left[ -  \frac{1}{2(\sigma^2 + 2 \rho)} X_+^2 - \frac{1}{2\sigma^2}X_-^2 \right]}{2 \pi \sqrt{ \sigma^2 ( \sigma^2 + 2 \rho)}}.
\end{equation}
Using the fact that, due to exchangeability, $P(\epsilon_\text{o},\epsilon_\text{e})=P(\epsilon_\text{e},\epsilon_\text{o})$, we can replace
\begin{multline}
 2(\epsilon_\text{e}-\mu)\theta( \epsilon_\text{o}- \epsilon_\text{e})\rightarrow (\epsilon_\text{e}-\mu)\theta( \epsilon_\text{o}- \epsilon_\text{e})+(\epsilon_\text{o}-\mu)\theta( \epsilon_\text{e}- \epsilon_\text{o})\\
 =\epsilon_\text{o}-\mu-\left(\epsilon_\text{o}-\epsilon_\text{e}\right)\theta\left(\epsilon_\text{o}-\epsilon_\text{e}\right)\\
 \rightarrow\sqrt{2}\frac{X_+-X_-\theta(X_-)}{\sqrt N}.
\end{multline}
and therefore
\begin{multline}
\sqrt{N}\int_{-\infty}^{+\infty} (\epsilon-\mu)P(\epsilon)\dd\epsilon  =  \\
= 2 \sqrt{N} \iint_{-\infty}^{+\infty} (\epsilon_\text{e}-\mu) \, \theta( \epsilon_\text{o}- \epsilon_\text{e}) P(\epsilon_\text{e}, \epsilon_\text{o})\dd\epsilon_\text{e}\dd\epsilon_\text{o}  \\
= \sqrt{2}\iint_{-\infty}^{+\infty} \left[X_+-X_-\theta(X_-)\right]P_X(X_-,X_+)\dd X_+\dd X_-\\
=-\frac{1}{\sqrt{\pi}\sigma}\int_{0}^{+\infty}X_- \exp \left[-\frac{1}{2\sigma^2}X_-^2 \right] \dd X_-=-\frac{\sigma}{\sqrt{\pi}}.
\end{multline}

\bibliography{biblio.bib}
\end{document}